\documentclass[12pt]{article}
\usepackage{amsmath}
\usepackage{amssymb}
\usepackage{amsthm}
\usepackage{bbm,bm}
\usepackage{color}
\usepackage[dvipsnames]{xcolor}
\usepackage{caption}
\usepackage{subcaption}
\usepackage{graphicx}
\usepackage{graphics}
\usepackage[pdftex,pagebackref,letterpaper=true,colorlinks=true,pdfpagemode=none,urlcolor=blue,linkcolor=blue,
citecolor=blue,pdfstartview=FitH,plainpages=true]{hyperref} 
\usepackage[comma,authoryear]{natbib}
\usepackage{float}

\renewcommand{\baselinestretch}{1.2}

\newcommand{\bfm}[1]{\ensuremath{\mathbf{#1}}}

\def\bc{\bfm c}     
     
\def\be{\bfm e}

     \def\bM{\bfm M}

\def\bv{\bfm v}     \def\bV{\bfm V}
\def\bw{\bfm w}     
\def\bx{\bfm x}

\newcommand{\bfsym}[1]{\ensuremath{\boldsymbol{#1}}}

\def\balpha{\bfsym \alpha}
\def\bbeta{\bfsym \beta}

\def\bgamma{\bfsym \gamma}

\def\bdelta{\bfsym \delta}

\def\bSigma{\bfsym \Sigma}

\def\bphi{\bfsym \phi}

\def\bvarphi{\bfsym \varphi}

\def\cA{{\cal  A}}

\def\cC{{\cal  C}}
\def\cD{{\cal  D}}
\def\cE{{\cal  E}}
\def\cF{{\cal  F}}

\def\cH{{\cal  H}}
\def\cI{{\cal  I}}

\def\cK{{\cal  K}}
\def\cL{{\cal  L}}
\def\cM{{\cal  M}}
\def\cN{{\cal  N}}

\def\cP{{\cal  P}}

\def\cR{{\cal  R}}
\def\cS{{\cal  S}}
\def\cT{{\cal  T}}

\def\cW{{\cal  W}}
\def\cX{{\cal  X}}

\def \bbP {\mathbb{P}}
\def \bbE {\mathbb{E}}

\def \bbR {\mathbb{R}}

\renewcommand{\hat}{\widehat}
\renewcommand{\tilde}{\widetilde}

\def\hbbeta{\hat{\bbeta}}

\def \bomega    {\bfsym {\omega}}
\def \bphi      {\bfsym {\phi}}

\def\bSigma{\bfsym \Sigma}
\def \bDelta    {\bfsym {\Delta}}

\DeclareMathOperator{\argmin}{argmin}
\DeclareMathOperator{\Var}{Var}

\DeclareMathOperator{\Cov}{Cov}

\def\r#1{\textcolor{red}{\bf #1}}

\def\revp#1{\textcolor{orange}{\bf #1}}

\newcommand{\todist}{\stackrel{d}{\to}}

\theoremstyle{plain}
\newtheorem{theorem}{Theorem}[section]
\newtheorem{lemma}[theorem]{Lemma}

\newtheorem{corollary}[theorem]{Corollary}

\theoremstyle{definition}

\newtheorem{assumption}{Assumption}[section]

\theoremstyle{remark}
\newtheorem{remark}{Remark}
\begin{document}


{
  \title{\bf Communication-Efficient Distributed Estimation and Inference for Cox's Model}
  \author{Pierre Bayle\thanks{The authors gratefully acknowledge the support of NSF Grants DMS-2053832 and DMS-2210833 and ONR grant N00014-22-1-2340. Emails: \texttt{pbayle@alumni.princeton.edu}, \texttt{jqfan@princeton.edu},
  \texttt{ZHL318@pitt.edu}.} \qquad Jianqing Fan\footnotemark[1] \qquad Zhipeng Lou\footnotemark[1] \\
    Department of Operations Research and Financial Engineering\\
    Princeton University}
    \date{}
  \maketitle
}

\bigskip
\begin{abstract}
Motivated by multi-center biomedical studies that cannot share individual data due to privacy and ownership concerns, we develop communication-efficient iterative distributed algorithms for estimation and inference in the high-dimensional sparse Cox proportional hazards model. We demonstrate that our estimator, even with a relatively small number of iterations, achieves the same convergence rate as the ideal full-sample estimator under very mild conditions. To construct confidence intervals for linear combinations of high-dimensional hazard regression coefficients, we introduce a novel debiased method, establish central limit theorems, and provide consistent variance estimators that yield asymptotically valid distributed confidence intervals. In addition, we provide valid and powerful distributed hypothesis tests for any coordinate element based on a decorrelated score test. We allow time-dependent covariates as well as censored survival times. Extensive numerical experiments on both simulated and real data lend further support to our theory and demonstrate that our communication-efficient distributed estimators, confidence intervals, and hypothesis tests improve upon alternative methods.
\end{abstract}

\vfill

\newpage
\def\spacingset#1{\renewcommand{\baselinestretch}%
{#1}\small\normalsize} \spacingset{1.4}
\section{Introduction}
Massive datasets are ubiquitous and require proper storage, computation, and analysis. The amount of available data can be overwhelming and may not fit into a single machine, in which case storing and processing data across multiple machines provides a solution. Furthermore, information is often naturally collected in parallel: for a given study or survey, different counties or states may obtain data independently from their inhabitants, or competing firms may have their own customers' information.  Additionally, algorithms might not be scalable, and would take considerable time and computational power to run if the sample size is too large on a single machine. Splitting the data across multiple machines alleviates the problem.

Moreover, privacy is of utmost importance and is yet another explanation of the pressing need of carefully-designed distributed methods. Sharing data from multiple locations may cause privacy and ownership concerns, which need to be avoided in a variety of applications. These concerns thus hamper the full access of the whole data at any location. One striking example is in the field of medical science, where hospitals and laboratories work with clinical and genomic data and may not share patient-level information. As an illustration, the Observational Health Data Sciences and Informatics (OHDSI) network is an international collaboration which only transmits aggregated information from local health databases to a coordinating center, without sharing any patient-level data~\citep{Hripcsak2015observational}.

To address the distributed data challenges, it is pressing to develop communication-efficient distributed algorithms that yield estimators and inference procedures as good as if having access to the full data. The literature on distributed learning is flourishing and encompasses a variety of topics: $M$-estimation~\citep{Zhang2013communication, Shamir2014communication, Lee2017communication, Wang2017efficient, Battey2018distributed, Jordan2019communication, Fan2021communication}, Principal Component Analysis~\citep{Garber2017communication-pca, Fan2019distributed-eigen, Chen2021distributed-pca}, feature screening~\citep{Li2020distributed-screening}, to name a few.
We refer to~\citet{Gao2022review} for a recent review of the distributed learning literature.

Most methods rely on one round of communication between machines and a central processor, which are referred to as \textit{one-shot} approaches. Different machines carry out their statistical analysis locally on disjoint subsets of the dataset and send their local results to a central machine which then aggregates the information.
For example,~\citet{Zhang2013communication} averaged the estimators computed locally on each machine to form a global estimator.
In sparse high-dimensional linear and generalized linear models, using the debiasing approach~\citep{Javanmard2014confidence,vandeGeer2014asymptotically},~\citet{Lee2017communication} and~\citet{Battey2018distributed} debiased the local estimators before averaging them for parameter estimation. \citet{Battey2018distributed} further proposed distributed sparse high-dimensional inference for these models, based on debiasing.

To achieve optimal statistical convergence rate, these one-shot approaches require the number of machines to be small and therefore it considerably limits their applications. To overcome this difficulty, iterative procedures have been proposed~\citep{Shamir2014communication, Wang2017efficient, Jordan2019communication, Fan2021communication}, where there are multiple rounds of communication between the machines and central processor. Throughout the paper, we will prefer to use the term \textit{centers} rather than machines, as one core application of survival analysis is medical science.

There is a small body of literature on distributed learning for survival analysis, but previous works only tackle parameter estimation, not inference, and do not address fundamental issues such as communication efficiency and high-dimensional data. \citet{Lu2015webdisco} proposed a proof-of-concept web service and a distributed iterative algorithm to perform Cox regression, and~\citet{Duan2020learning} developed a one-shot distributed algorithm; these methods are privacy-preserving but they have major drawbacks. They do not work for high-dimensional data as regularization is not used. Moreover, they require transmitting Hessian matrices, which is not communication-efficient as it can be too costly to communicate $p^2$ entries, where $p$ is the number of predictors. Many datasets being high-dimensional nowadays, such as microarrays and SNPs in the fields of biomedicine and genomics, it is crucial to develop and analyze communication-efficient privacy-preserving distributed algorithms that are able to deal with the huge number of predictors.

In survival analysis, a widely used semi-parametric model is Cox's proportional hazards model~\citep{Cox1972regression}. The outcome variable is time-to-event, and often, some observations are censored. For example, a subject can leave the study before its end. \citet{Andersen1982cox} formulated the model into a counting process framework, which in particular allows for a systematic study of time-dependent covariates. Despite the very large number of predictors in many modern datasets, most predictors are often irrelevant to explain the outcome, leading to sparse models. Several regularized regression techniques have been extended to Cox's proportional hazards model~\citep{Tibshirani1997lasso-cox, Fan2002variable-cox}. \citet{Bayle2022factor} proposed a factor-augmented regularized model that is capable to handle high-dimensional correlated covariates. \citet{Bradic2011regularization} established model selection consistency and strong oracle properties for various penalty functions in the ultra-high dimensional setting. \citet{Huang2013oracle} and~\citet{Kong2014non} established oracle inequalities for LASSO under different conditions.

In the context of the high-dimensional linear regression model, fundamental statistical inference for linear functionals $\bc^{\top }\bbeta^{\star}$ of the population parameter vector has been well studied in the literature. \citet{Cai2017confidence} established a systematic theory of the minimax expected length for confidence intervals of $\bc^{\top}\bbeta^{\star}$; see also~\citet{Zhu2018linear} and~\citet{Javanmard2020flexible}. For the high-dimensional Cox model under sparsity, \citet{Yu2021confidence} and~\citet{Xia2022statistical} derived central limit theorems for linear functionals, and~\citet{Fang2017testing} proposed hypothesis testing procedures for low-dimensional components of the population parameter vector.

\paragraph{Our contributions} The contributions of our paper can be summarized as follows. First of all, we study a communication-efficient iterative distributed algorithm for parameter estimation in the high-dimensional sparse Cox proportional hazards model regularized with the $\ell_{1}$ (LASSO) penalty, and prove under mild conditions that our estimator achieves the same rate of convergence as an ideal estimator that would have access to the full dataset. In particular, we do not require the number of centers to be small. We then design novel communication-efficient distributed inference procedures for Cox's model: confidence intervals construction for $\bc^{\top}\bbeta^{\star}$ using a debiasing approach, and hypothesis testing for any coordinate element of $\bbeta^{\star}$ building upon decorrelated score test statistics. To this end, we establish central limit theorems and derive consistent variance estimators. Even though the topic of communication-efficient distributed estimation and inference in the context of Cox's proportional hazards model is crucial for many reasons, including but not limited to alleviating privacy issues in medical science, it has not been studied in past literature, and our contributions are the first of their kind. In addition to providing strong theoretical guarantees for these algorithms, we perform numerical experiments. They lend further support that our distributed estimator becomes as accurate as the full-sample estimator after only a few rounds of communication, our confidence intervals have good empirical coverage and small width, and our hypothesis tests have the correct size and high power.

The rest of the paper is organized as follows. Section~\ref{sec_setup} formulates Cox's model in a distributed environment. In Section~\ref{sec_iterative}, we introduce the iterative algorithm and establish convergence of our estimator with optimal rate. We carry out inference in Section~\ref{sec_inference}, constructing confidence intervals for linear functionals $\bc^{\top} \bbeta^{\star}$ of the population parameter, and providing hypothesis testing procedures for coordinates of $\bbeta^{\star}$. The application of our communication-efficient distributed estimation procedure, confidence intervals and tests to both simulated and real data in Section~\ref{sec_experiments} corroborates our theory. We prove our results in the Appendix.

\paragraph{Notation} For any integer $n$, we denote $[n] = \{1, \dots, n\}$. For a vector $\bgamma = (\gamma_{1}, \ldots, \gamma_{p})^{\top} \in \bbR^{p}$ and $q \geq 1$, denote the $\ell_{q}$ norm $\|\bgamma\|_{q} = (\sum_{i = 1}^{p} |\gamma_{i}|^{q})^{1/q}$ and $\|\bgamma\|_{\infty} = \max_{1\leq i \leq p} |\gamma_{i}|$, as well as $\bgamma^{\otimes 0} = 1$, $\bgamma^{\otimes 1} = \bgamma$, $\bgamma^{\otimes 2} = \bgamma\bgamma^\top$. For a matrix $\bM$, we denote by $\|\bM\|_{\max}=\max_{i,j}|M_{ij}|$ its max norm, and by $\|\bM\|_q$ its induced $q$-norm for $q\in \mathbb{N}^{\star}\cup\{\infty\}$; $\lambda_{\max}(M)$ is its maximum eigenvalue. For a matrix $\bM \in \bbR^{p \times p}$ and a vector $\bgamma = (\bgamma_{1}, \bgamma_{2})$ with $\bgamma_{i} \in \bbR^{p_i}$ where $p_1 + p_2 = p$, let $\bM_{\bgamma_{1}\bgamma_{2}} \in \bbR^{p_1 \times p_2}$ be the submatrix of $\bM$ corresponding to rows $1, \dots, p_1$ and columns $p - p_2 + 1, \dots, p$. We similarly define $\bM_{\bgamma_{1}\bgamma_{1}}$, $\bM_{\bgamma_{2}\bgamma_{1}}$ and $\bM_{\bgamma_{2}\bgamma_{2}}$. For a set or an event $A$, we use $\mathbb{I}\{A\}$ to represent the indicator function of $A$, and $A^{c}$ is the complement of $A$. Let $\nabla$ and $\nabla^2$ be the gradient and Hessian operators. When $f$ is a function of $\bgamma = (\bgamma_{1}, \bgamma_{2})$, $\nabla_{\bgamma_{i}} f$ is the vector of partial derivatives with respect to $\bgamma_{i}$, and $\nabla^2_{\bgamma_{i} \bgamma_{j}} f = (\nabla^2 f)_{\bgamma_{i} \bgamma_{j}}$ for $i, j \in \{1, 2 \}$. For two positive sequences $a_{n}$ and $b_{n}$, we write $a_{n} \lesssim b_{n}$ if there exists a constant $C > 0$ independent of $n$ such that $a_{n} \leq C \, b_{n}$ for all sufficiently large $n$. For two numbers $a$ and $b$, $a \lor b$ and $a \land b$ denote their maximum and minimum, respectively. $\mathcal{N}(0,1)$ refers to the standard normal distribution and $\Phi$ is its cumulative distribution function. Let $\todist$ and $\overset{\bbP}{\to}$ denote convergence in distribution, and in probability, respectively. For two sequences of random variables $X_{n}$ and $Y_{n}$ with non-negative $Y_{n}$, we write $X_{n} = O_{\bbP}\left(Y_n\right)$ if for any $\varepsilon > 0$, there exists a constant $C > 0$ independent of $n$ such that $\bbP(|X_n| \geq C \, Y_n) \leq \varepsilon$ for all sufficiently large $n$.

\section{Problem Setup}
\label{sec_setup}
\subsection{Cox's proportional hazards model}
Let $T$, $C$ denote the survival time and the censoring time, respectively, and let $\{\bx(t) \in \bbR^{p} : 0\leq t\leq \tau\}$ be a vector-valued predictable covariate process. Here $\tau$ is the study ending time, which we consider to be finite. As commonly assumed in the literature, $T$ and $C$ are conditionally independent given the covariates $\{\bx(t) : 0\leq t\leq \tau\}$. Let $Z = \min\{T, C\}$ be the observed time and $\delta = \mathbb{I}\{T \leq C\}$ be the censoring indicator. We observe $n$ independent and identically distributed (i.i.d)~samples $\{(\{\bx_{i}(t) : 0\leq t\leq \tau\}, Z_{i}, \delta_{i})\}_{i = 1}^{n}$ from the population $(\{\bx(t) : 0\leq t\leq \tau\}, Z, \delta)$, and we assume that there are no tied observations and the covariates are centered.

Cox's proportional hazards model~\citep{Cox1972regression} is a semi-parametric model widely used for time-to-event outcomes. In this model, the conditional hazard function $\lambda(t\mid \bx(t))$ of the survival time $T$ at time $t$ given covariate vector $\bx(t) \in \bbR^{p}$ is assumed to have the form
\begin{align*}
    \lambda(t\mid \bx(t)) = \lambda_{0}(t) \exp\{\bx(t)^{\top}\bbeta^{\star}\},
\end{align*}
where $\lambda_{0}(\cdot)$ is the baseline hazard function and $\bbeta^{\star} = (\beta_{1}^{\star}, \ldots, \beta_{p}^{\star})^{\top} \in \bbR^{p}$ is the population parameter vector.

We introduce usual counting process notation~\citep{Andersen1982cox}, namely $N_{i}(t) = \mathbb{I}\{Z_{i} \leq t, \delta_{i} = 1\}$ and $Y_{i}(t) = \mathbb{I}\{Z_{i} \geq t\}$. Then the average negative log-partial likelihood function is given by 
\begin{align*}
    \cL(\bbeta) = - \frac{1}{n} \sum_{i = 1}^n \int_{0}^{\tau} \{\bx_{i}(t)^{\top} \bbeta\} d N_{i}(t) + \frac{1}{n} \int_{0}^{\tau} \log\left[\sum_{i=1}^{n} Y_{i}(t) \exp\{\bx_{i}(t)^{\top} \bbeta\}\right] d \bar{N}(t),
\end{align*}
where $\bar{N}(t) = \sum_{i = 1}^{n} N_{i}(t)$. To estimate the parameter vector $\bbeta^{\star}$ in the high-dimensional setting where the dimension $p$ is larger than the sample size $n$, we consider the following regularized estimator
\begin{align*}
    \hat{\bbeta}_{\vartheta} = \underset{\bbeta\in\bbR^{p}}{\argmin} \{\cL(\bbeta) + \cP_{\vartheta}(\bbeta)\},
\end{align*}
where $\cP_{\vartheta}(\cdot)$ is a sparsity-inducing penalty function with regularization parameter $\vartheta > 0$. Most predictors are irrelevant to explain the outcome, and it is commonly assumed that $\bbeta^{\star}$ is a sparse vector. We define $\cS_{\star} = \{j \in [p]: \beta_{j}^{\star} \neq 0\}$ to be the support of $\bbeta^{\star}$ and let $|\cS_{\star}| = \sum_{j = 1}^{p} \mathbb{I}\{\beta_{j}^{\star} \neq 0\}$ be its cardinality. Various types of penalty functions have been investigated in the literature. \citet{Tibshirani1997lasso-cox} and~\citet{Fan2002variable-cox} introduced the LASSO and the SCAD penalties for survival data, respectively. In~\citet{Bradic2011regularization}, the authors established strong oracle properties in the high-dimensional setting. \citet{Huang2013oracle} established oracle inequalities for the Lasso estimator under quite mild conditions.

For $\ell = 0,1,2$, define the following quantities
\begin{align}
\label{S-s-param}
    S^{(\ell)}(\bbeta, t) = \frac{1}{n} \sum_{i = 1}^{n} Y_{i}(t) \{\bx_{i}(t)\}^{\otimes \ell} \exp\{\bx_{i}(t)^{\top}\bbeta\} \enspace \mathrm{and} \enspace s^{(\ell)}(\bbeta, t) = \bbE \{S^{(\ell)}(\bbeta, t)\}.
\end{align}
To ease the notations, also define the following
\begin{align*}
    \cX(\bbeta, t) = \frac{S^{(1)}(\bbeta, t)}{S^{(0)}(\bbeta, t)} \enspace \mathrm{and} \enspace \bV(\bbeta, t) = \frac{S^{(2)}(\bbeta, t)}{S^{(0)}(\bbeta, t)} - \cX(\bbeta, t)^{\otimes 2}. 
\end{align*}
Then the gradient vector and Hessian matrix of $\cL(\bbeta)$ are respectively given by 
\begin{align}
\label{eq_Gradient_Hessian}
    \nabla \cL(\bbeta) = -\frac{1}{n} \sum_{i = 1}^{n} \int_{0}^{\tau} \{\bx_{i}(t) - \cX(\bbeta, t)\} d N_{i}(t) \enspace \mathrm{and} \enspace \nabla^{2} \cL(\bbeta) = \frac{1}{n} \int_{0}^{\tau} \bV(\bbeta, t) d\bar{N}(t).  
\end{align}
Observe that $N_{i}(t)$ is a counting process with intensity process $\lambda_{i}(t, \bbeta^{\star}) = \lambda_0(t) Y_i(t) \exp\{\bx_i(t)^{\top} \bbeta^{\star}\}$, which does not admit jumps at the same time as $N_j(t)$ for $j \neq i$. For each $i \in [n]$, define the predictable compensator 
\begin{align*}
    \Lambda_i(t) = \int_{0}^{t} \lambda_i(s, \bbeta^\star) d s.
\end{align*}
Then $M_i(t) = N_i(t) - \Lambda_i(t)$ is an orthogonal local square-integrable martingale with respect to the filtration $\cF_{t,i} = \sigma\{N_i(u), \bx_i(u^+), Y_i(u^+) : 0 \leq u \leq t\}$. Let $\cF_{t} = \bigcup_{i=1}^n \cF_{t,i}$ be the smallest $\sigma$-algebra containing $\cF_{t,i}$. Then $\bar{M}(t) = \sum_{i=1}^n M_i(t)$ is a martingale with respect to $\cF_t$.

\subsection{Distributed environment}
Let $\cI_{1}, \ldots, \cI_{K}$ be a partition of $[n]$ corresponding to indices of the datapoints in each of the $K$ centers. Without loss of generality, we assume that $n/K$ is an integer called~\textit{subsample size} and denoted by $m$, with $|\cI_{1}| = \cdots = |\cI_{K}| = m$.

For the $k$-th center, the average negative log-partial likelihood function is
\begin{align*}
    \cL_{k}(\bbeta) = - \frac{1}{m} \sum_{i \in \cI_{k}} \int_{0}^{\tau} \{\bx_{i}(t)^{\top} \bbeta\} d N_{i}(t) + \frac{1}{m} \int_{0}^{\tau} \log \left[\sum_{i\in \cI_{k}} Y_{i}(t) \exp\{\bx_{i}(t)^{\top} \bbeta\} \right] d\bar{N}_{k}(t),   
\end{align*}
where $\bar{N}_{k}(t) = \sum_{i \in \cI_{k}} N_{i}(t)$ for each $k \in [K]$.
Having data split across numerous centers further reduces the ratio of the subsample size to the dimension, and makes the study and use of penalized methods even more crucial. In this paper, we focus on the penalty $\cP_{\vartheta}(\bbeta) = \vartheta \|\bbeta\|_{1}$ for some tuning parameter $\vartheta > 0$, which yields $\ell_{1}$-regularized Cox proportional hazards regression. The corresponding estimate for $\bbeta^{\star}$ based on the full sample is given by
\begin{align*}
    \hbbeta_{\vartheta} = \underset{\bbeta\in\bbR^{p}}{\argmin} \{\cL(\bbeta) + \vartheta \|\bbeta\|_{1}\}.
\end{align*}
Under fairly general regularity conditions, \citet{Huang2013oracle} established that
\begin{align}
\label{eq_hbbeta_lambda}
    \|\hbbeta_{\vartheta} - \bbeta^{\star}\|_{2} = O_{\bbP}\left(\sqrt{\frac{|\cS_{\star}|\log p}{n}}\right). 
\end{align}
Throughout this paper, we use $\hbbeta_{\vartheta}$ as the benchmark estimator in the context of distributed learning. Our first goal is to design a distributed estimator that attains the same convergence rate as $\hbbeta_{\vartheta}$.

\section{Iterative Algorithm and Convergence} 
\label{sec_iterative}
In the principal center $\cI_{1}$, a natural idea would be to approximate $\cL(\bbeta)$ via the subsample version $\cL_{1}(\bbeta)$. However, this approximation does not utilize the information across the centers. To obtain a better trade-off between communication and computation efficiencies, we use the gradient-enhanced loss (GEL) function~\citep{Jordan2019communication, Fan2021communication}. More specifically, given any vector $\bbeta^{\circ} \in \bbR^{p}$, the GEL function at the central processor $\cI_{1}$ at $\bbeta^{\circ}$ is defined by  
\begin{align*}
    \tilde{\cL}_{1}(\bbeta) =
    \cL_{1}(\bbeta) - \{\nabla \cL_{1}(\bbeta^{\circ}) - \nabla \cL(\bbeta^{\circ})\}^{\top} \bbeta.  
\end{align*}
This replaces the local gradient at the point $\bbeta^{\circ}$ at the principal center by the full-sample gradient, namely $\nabla \tilde \cL_{1}(\bbeta^{\circ}) =  \nabla \cL(\bbeta^{\circ})$. However, in view of~\eqref{eq_Gradient_Hessian}, the exact value of $\nabla \cL (\bbeta^{\circ})$ is not computable without accessing all the data points across the $K$ centers. Therefore, we propose to approximate $\nabla \cL(\bbeta^{\circ})$ via the following distributed version,
\begin{align*}
    \hat{\nabla \cL} (\bbeta^{\circ}) = \frac{1}{K}\sum_{k = 1}^{K} \nabla \cL_{k}(\bbeta^{\circ}).  
\end{align*}
This can easily be communicated from each individual center to the principal center $\mathcal{I}_{1}$. Consequently, for any initial estimator $\bbeta_{0}$ of $\bbeta^{\star}$ and any $t \geq 0$, we compute an estimator $\bbeta_{t + 1}$ as
\begin{align}
\label{eq_iterated_estimator}
    \bbeta_{t + 1} = \underset{\bbeta \in \mathbb{R}^{p}}{\argmin} \left[\cL_{1}(\bbeta) - \left\{\nabla \cL_{1} (\bbeta_{t}) - \hat{\nabla \cL} (\bbeta_{t})\right\}^{\top} \bbeta + \vartheta_{t + 1} \|\bbeta\|_{1}\right],
\end{align}
where $\vartheta_{t + 1} > 0$ is a regularization parameter and $\bbeta_{t}$ is the previous iterate. It is worth mentioning that this estimating procedure is invariant to the possibly different baseline hazard functions across the $K$ centers, which are treated as nuisance parameters in the semi-parametric modelling of Cox's proportional hazards model. Before presenting the main result on the convergence rate of the estimator $\bbeta_{t + 1}$ for $t \geq 0$, we first introduce some basic notation and regularity conditions. For any subset $\mathcal{S} \subset [p]$ and positive constant $\gamma < \infty$, we denote the convex cone $\mathcal{C}(\mathcal{S}, \gamma) = \{\bv \in \mathbb{R}^{p} : \|\bv_{\mathcal{S}^{c}}\|_{1} \leq \gamma \|\bv_{\mathcal{S}}\|_{1}\}$, where $\mathcal{S}^{c}$ is the complement of $\mathcal{S}$.

\begin{assumption}
\label{Assumption_covariate}
There exists a positive constant $B < \infty$ such that 
\begin{align*}
    \max_{i \in [n]} \sup_{t \in [0, \tau]} \|\bx_{i}(t)\|_{\infty} \leq B. 
\end{align*}
\end{assumption}

\begin{assumption}
\label{Assumption_eigenvalue}
There exists a constant $\varrho_{\star} > 0$ such that
\begin{align*}
    \min_{0 \neq \bv \in \cC(\cS_{\star}, 3)} \frac{\bv^\top \nabla^2 \cL_{1}(\bbeta^{\star}) \bv}{\|\bv\|_{2}^{2}} \geq \varrho_{\star}^{2}.
\end{align*}
\end{assumption}

Here, we will briefly comment on these two assumptions. Assumption~\ref{Assumption_covariate} is a mild boundedness condition. It is very common in the literature on the Cox's proportional hazards model; see, for instance, \citet{Zhao2012principled}, \citet{Kong2014non} and~\citet{Yu2021confidence}, among others. Assumption~\ref{Assumption_eigenvalue} is a restricted eigenvalue condition and has been justified to be reasonable in view of Theorem 4.1 in~\citet{Huang2013oracle} which requires that $|\cS_{\star}|\sqrt{(\log p)/m}$ be sufficiently small. Henceforth, we assume that there exists a positive constant $R_{0} < \infty$ such that $|\cS_{\star}|\sqrt{\{\log (pK)\}/m} \leq R_{0}$. In the following lemma, we derive a non-asymptotic upper bound for the $\ell_{2}$ error of $\bdelta_{t} = \bbeta_{t} - \bbeta^{\star}$ for each integer $t \geq 1$, which serves as a useful lemma to our convergence results.

\begin{lemma}
\label{Theorem_consistency}
For each integer $t \geq 0$, assume that
\begin{align}
\label{eq_deviation_t}
    \chi_{t + 1} := \frac{12\vartheta_{t + 1} B |\cS_{\star}|}{\varrho_{\star}^{2}} \leq \frac{1}{e} \enspace \mathrm{and} \enspace \|\nabla \cL_{1}(\bbeta^{\star}) - \nabla \cL_{1}(\bbeta_t) + \hat{\nabla \cL} (\bbeta_{t})\|_{\infty} \leq \frac{\vartheta_{t + 1}}{2}.
\end{align}
Then, under Assumptions~\ref{Assumption_covariate} and \ref{Assumption_eigenvalue}, we have $\bdelta_{t + 1} \in \cC(\cS_{\star}, 3)$, and the upper bound
\begin{align*}
    \|\bdelta_{t + 1}\|_2 \leq \frac{3\vartheta_{t + 1} e \sqrt{|\cS_{\star}|}}{2\varrho_{\star}^2}.
\end{align*}
\end{lemma}

\begin{assumption}
\label{Assumption_beta_1}
There exists a positive constant $M < \infty$ such that
\begin{align*}
    \max_{i \in [n]} \sup_{t \in [0, \tau]}|\bx_{i}(t)^{\top} \bbeta^{\star}| \leq M. 
\end{align*}
\end{assumption}

\begin{assumption}
\label{Assumption_covariance_matrix_covariate}
Let $\bSigma(t) = \Cov\{\bx(t)\}$ denote the covariance matrix, and let $\Lambda_{0}(t) = \int_{0}^{t} \lambda_{0}(s) ds$ be the cumulative baseline hazard function. There exists a positive constant $R_{1} < \infty$ such that 
\begin{align*}
    \lambda_{\max} \left(\int_{0}^{\tau} \bSigma(t) d\Lambda_{0}(t)\right) \leq R_{1}.
\end{align*}
\end{assumption}

\begin{assumption}
\label{Assumption_at_risk}
There exists a constant $\rho_0 > 0$ such that $\bbP\{Y_{1}(\tau) = 1\} \geq \rho_0$. 
\end{assumption}

Assumption~\ref{Assumption_beta_1} is a natural uniform upper bound on the true hazard~\citep{Fang2017testing}. In view of Assumption~\ref{Assumption_covariate}, a naive choice of $M$ would be $B \|\bbeta^{\star}\|_{1}$. Assumption~\ref{Assumption_covariance_matrix_covariate} imposes a natural upper bound on the maximum eigenvalue of the weighted covariance matrix for the time-dependent covariate vector $\bx(t)$, which is quite mild and commonly applied in the context of distributed learning~\citep{Battey2018distributed, Jordan2019communication, MR4209468}. Assumption~\ref{Assumption_at_risk} is standard and states that there is a non-zero probability that subjects are still at risk at the end of the study. See, for instance, \citet{Andersen1982cox}, \citet{Fan2002variable-cox}, \citet{Lin2013high}, \citet{Fang2017testing} and many others.

For simplicity of notation, we denote
\begin{align*}
     \mathcal{C}_{n, p} = \frac{6\eta_{0}eB}{\varrho_{\star}^{2}}\sqrt{\frac{|\cS_{\star}|\log p}{n}} \enspace \mathrm{and} \enspace \Psi_{m, p}(\bar{\omega}) = \frac{\cA_{1}}{\cA_{2}\exp(2B\bar{\omega})} \sqrt{\frac{|\cS_{\star}| \log (pK)}{m}},
\end{align*}
where $\bar{\omega}$ is defined in~\eqref{eq_cond_omega} below, $\mathcal{A}_{1}$ and $\mathcal{A}_{2}$ are positive constants depending on $B, M, \rho_{0}, R_{0}$ and $R_{1}$, whose explicit expressions are given in~\eqref{eq_A1_A2_definition}.

\begin{theorem}
\label{Theorem_iteration}
Let Assumptions~\ref{Assumption_covariate}--\ref{Assumption_at_risk} hold. Suppose that for each integer $t \geq 0$, 
\begin{align} 
\label{eq_cond_lambda}
    \frac{\vartheta_{t + 1}}{\eta_0} \leq \|\nabla\cL_{1}(\bbeta^{\star}) - \nabla \cL_{1}(\bbeta_{t}) + \hat{\nabla \cL} (\bbeta_{t})\|_{\infty} \leq \frac{\vartheta_{t + 1}}{2}, 
\end{align}
where $\eta_{0} > 2$ is a positive constant. Assume that
\begin{align}
\label{eq_cond_number_centers}
    \mathcal{C}_{1} |\cS_{\star}| \sqrt{\frac{\log (p K)}{m}} \leq \cA_{0} < 1
\end{align}
for some positive $\cA_{0}$, where $\cC_{1} = 3\eta_0 e \cA_{1}/\varrho_{\star}^{2}$. Moreover, we assume that there exists some positive $\bar{\omega} < \infty$ such that $8 B \Psi_{m, p}(\bar{\omega}) \sqrt{|\cS_{\star}|} \leq 1$ and
\begin{align}
\label{eq_cond_initial_estimator}
\Psi_{m, p}(\bar{\omega}) \geq \max\left\{\|\bdelta_{0}\|_{2}, \frac{\mathcal{C}_{n, p}}{1 - \cA_{0}}\right\},
\end{align}
where
\begin{align}
\label{eq_cond_omega}
\bar{\omega} = \max\left\{\|\bdelta_{0}\|_{1}, 8\sqrt{|\cS_{\star}|} \cA_{0}\|\bdelta_{0}\|_{2}, \frac{8\sqrt{|\cS_{\star}|}\mathcal{C}_{n, p}}{1 - \cA_{0}}\right\}.
\end{align}
Then, with probability at least $1 - 2K\exp(-m\rho_0^2/2) - 2p^{-1} - 12.442/(pK^{2})$, we have for all $t \geq 0$, $\bdelta_{t + 1} \in \cC(\cS_{\star}, 3)$ and
\begin{align}
\label{eq_bbeta_bound_iterated}
    \|\bdelta_{t + 1}\|_{2} &\leq \cA_{0}^{t + 1} \|\bdelta_{0}\|_{2} + \frac{6\eta_{0} e B}{(1 - \cA_{0})\varrho_{\star}^{2}} \sqrt{\frac{|\cS_{\star}|\log p}{n}}. 
\end{align}
\end{theorem}

\begin{remark}
\label{Remark_sufficient_condition}
It is worth mentioning that Theorem~\ref{Theorem_iteration} ensures $\sup_{t \geq 0} \|\bdelta_{t}\|_{2} \leq \Psi_{m, p}(\bar{\omega})$. Indeed, the base case with $t = 0$ is satisfied by~\eqref{eq_cond_initial_estimator}, and by mathematical induction, if $\|\bdelta_{t}\|_{2} \leq \Psi_{m, p}(\bar{\omega})$ for some $t \geq 0$, then it follows from this induction hypothesis together with~\eqref{eq_cond_initial_estimator} and the upper bound on $\|\bdelta_{t}\|_{2}$ in~\eqref{eq_bbeta_bound} that 
\begin{align*}
    \|\bdelta_{t + 1}\|_{2} \leq \cA_{0} \|\bdelta_{t}\|_{2} + \cC_{n, p} \leq \cA_{0} \Psi_{m, p}(\bar{\omega}) + (1 - \cA_{0}) \Psi_{m, p}(\bar{\omega}) = \Psi_{m, p}(\bar{\omega}),
\end{align*}
so that by mathematical induction, we have $\sup_{t\geq 0} \|\bdelta_{t}\|_{2} \leq \Psi_{m, p}(\bar{\omega})$. Combined with the fact that $\delta_{t} \in \cC(\cS_{\star}, 3)$ for all $t \geq 1$ and the definition~\eqref{eq_cond_omega} of $\bar{\omega}$, it is also satisfied that $\sup_{t \geq 0} \|\bdelta_{t}\|_{1} \leq \bar{\omega}$, which facilitates the analysis of the error bound for the sequence $\{\bdelta_{t}\}_{t \geq 1}$.
To make sure that condition~\eqref{eq_cond_initial_estimator} holds, we can enlarge $\cA_{1}$ if needed, as $\Psi_{m, p}(\cdot)$ is increasing in $\cA_{1}$. Our results will hold as long as the corresponding condition~\eqref{eq_cond_number_centers} is satisfied. One can take the initial estimator $\bbeta_{0}$ to be the local $\ell_{1}$-regularized maximum partial likelihood estimator of the principal center $\mathcal{I}_{1}$, that is, 
\begin{align}
\label{eq_beta0_LASSO}
    \bbeta_{0} = \underset{\bbeta \in \mathbb{R}^{p}}{\argmin} \{\cL_{1}(\bbeta) + \vartheta_{0} \|\bbeta\|_{1}\},
\end{align}
where $\vartheta_{0} > 0$ is a tuning parameter. Specifically, taking $\vartheta_{0} = c_{0} B \sqrt{(\log p)/m}$ for some constant $c_{0} \geq 8$, Lemma~\ref{Theorem_consistency}
ensures that with probability at least $1 - 2 p^{-1}$,
\begin{align*}
    \|\bdelta_{0}\|_{2} \leq \frac{3 c_{0} e B}{2\varrho_{\star}^{2}} \sqrt{\frac{|\cS_{\star}|\log p}{m}}. 
\end{align*}
Consequently, Theorem~\ref{Theorem_iteration} ensures that for any iterate $\bbeta_{t}$ with $t \geq \lceil (\log K)/(2\log (1/\cA_{0}))\rceil$, we have the following high-probability bound  
\begin{align*}
    \|\bdelta_{t}\|_{2} \leq \left\{\frac{3 c_{0} e B}{2\varrho_{\star}^{2}} + \frac{6\eta_{0} e B}{(1 - \cA_{0})\varrho_{\star}^{2}}\right\}\sqrt{\frac{|\cS_{\star}|\log p}{n}},
\end{align*}
which has the same convergence rate as the full-sample estimator $\hbbeta_{\vartheta}$ in~\eqref{eq_hbbeta_lambda}. With this initial estimator $\bbeta_{0}$, Theorem~\ref{Theorem_iteration} shows that the convergence rate of the one-step estimator $\bbeta_{1}$ is
\begin{align*}
    \|\bdelta_{1}\|_{2} = O_{\bbP}\left(\frac{|\cS_{\star}|^{3/2}K \sqrt{(\log p) \log (pK)} }{n} + \sqrt{\frac{|\cS_{\star}| \log p}{n}}\right) = O_{\bbP}\left(\sqrt{\frac{|\cS_{\star}|\log p}{n}}\right),
\end{align*}
where the last equation follows if the number of centers $K$ satisfies
\begin{align}
\label{eq_K_cond_1}
    |\cS_{\star}| \sqrt{\frac{K \log(p K)}{m}} \lesssim 1.
\end{align}
This is the same assumption used in~\citet{Jordan2019communication} to obtain the optimal rate of convergence for the one-step estimator under the conventional generalized linear model. Compared with~\eqref{eq_K_cond_1}, to achieve the same convergence rate, we only require condition~\eqref{eq_cond_number_centers} on the number of centers, which is a much weaker condition.  
\qed 
\end{remark}

\begin{corollary}
\label{Corollary_beta0_LASSO}
Let $\bbeta_{0}$ be defined in~\eqref{eq_beta0_LASSO} with $\vartheta_0 = c_{0} B \sqrt{(\log p)/m}$ for some constant $c_{0} \geq 8$.
Assume that~\eqref{eq_cond_lambda} is satisfied and $|\cS_{\star}|\sqrt{\log (pK)/m} \leq \cA^{\diamond}$, where $\cA^{\diamond}$ is a positive constant depending only on $c_{0}, \varrho_{\star}, \eta_{0}, \cA_{0}, \cA_{1}$ and $\cA_{2}$. Then, under the conditions of Theorem~\ref{Theorem_iteration}, with probability at least $1 - 2K\exp(-m\rho_0^2/2) - 4/p - 12.442/(pK^{2})$, we have 
\begin{align*}
    \sup_{t \geq \lceil \frac{\log K}{2\log(1/\cA_{0})}\rceil}\|\bdelta_{t}\|_{2} \leq \mathcal{C}_{0}\sqrt{\frac{|\cS_{\star}| \log p}{n}},
\end{align*}
where $\cC_{0}$ is a positive constant depending only on $c_{0}, B, \varrho_{\star}, \eta_{0}$ and $\cA_{0}$, and its explicit expression is given in~\eqref{eq_A1_A2_definition}.
\end{corollary}

\section{Distributed Statistical Inference}
\label{sec_inference}
In this section, we consider the statistical inference of the parameter vector $\bbeta^{\star}$ in the distributed setting and propose two communication-efficient inference procedures: the first constructs confidence interval for the linear functional $\bc^{\top} \bbeta^{\star}$ via a novel distributed bias-corrected $\ell_{1}$-regularized estimator, where $\bc$ represents a $p$-dimensional loading vector of interest; the second focuses on testing for any coordinate element of $\bbeta^{\star}$ based on a distributed decorrelated score test.

\subsection{Inference for linear functional}
We begin with the construction of a distributed bias-corrected estimator for $\bc^{\top}\bbeta^{\star}$. Recall that $\tilde{\bbeta} = \bbeta_{T}$. Throughout this section, we use $\hat{\bbeta} = \bbeta_{T + 1}$, defined in~\eqref{eq_iterated_estimator} with $t = T$, as the estimator for $\bbeta^{\star}$. Then, for the $k$-th center, we compute  
\begin{align}
\label{eq_estimation_omega}
    \hat{\bomega}_{k} = \underset{\bomega \in \bbR^{p}}{\argmin} \left\{\bomega^{\top} \nabla^{2} \cL_{k}(\hat{\bbeta})\bomega - 2\bc^{\top}\bomega + \vartheta_{k}^{\diamond} \|\bomega\|_{1}\right\}, 
\end{align}
where $\vartheta_{k}^{\diamond} > 0$ is a regularization parameter. Motivated by~\citet{vandeGeer2014asymptotically}, our bias-corrected estimator for $\bc^{\top}\bbeta^{\star}$ under the distributed setting is defined as  
\begin{align}
\label{eq_Debias_Linear_Functional}
    \tilde{\bc^{\top}\bbeta^{\star}} = \bc^{\top}\hat{\bbeta} + \frac{1}{K} \sum_{k = 1}^{K} \hat{\bomega}_{k}^{\top} \left\{\nabla \cL_{k}(\tilde{\bbeta}) - \nabla \cL_{k} (\hat{\bbeta}) - \hat{\nabla\cL}(\tilde{\bbeta})\right\}, 
\end{align}
in view of~\eqref{eq_iterated_estimator}. Let $\bomega^{\star} = (\omega_{1}^{\star}, \ldots, \omega_{p}^{\star})^{\top} \in \bbR^{p}$ denote the population version of $\hat{\bomega}_k$, which is defined by $\cH^{\star}\bomega^{\star} = \bc$, where $\cH^{\star}$ is the population Hessian matrix defined by
\begin{align}
\label{eq_L_star}
    \cH^{\star} = \bbE \left[\int_{0}^{\tau} Y_{1}(t) \exp\{\bx_{1}(t)^\top \bbeta^{\star}\} \{\bx_{1}(t) - \be(\bbeta^{\star}, t)\}^{\otimes 2} d\Lambda_0(t)\right],
\end{align}
and $\be(\bbeta^{\star}, t) = s^{(1)}(\bbeta^{\star}, t)/s^{(0)}(\bbeta^{\star}, t)$. Under the high-dimensional setting, we assume that $\bomega^{\star}$ is sparse. Let $\cS_{\diamond} = \{j \in [p] : \omega_{j}^{\star} \neq 0\}$ denote the support of $\bomega^{\star}$ and $|\cS_{\diamond}| = \sum_{j = 1}^{p} \mathbb{I}\{\omega_{j}^{\star} \neq 0\}$ be its cardinality. Similar to Assumption~\ref{Assumption_covariate} and Assumption~\ref{Assumption_eigenvalue}, we impose the following regularity conditions to study the consistency of $\hat{\bomega}_{k}$, $k = 1, \ldots, K$.

\begin{assumption}
\label{Assumption_cM_bomega}
There exists a positive constant $\cM_{\diamond} < \infty$ such that 
\begin{align*}
    \max_{i \in [n]} \sup_{t \in [0, \tau]} |\bx_{i}(t)^{\top} \bomega^{\star}| \leq \cM_{\diamond}.
\end{align*}
\end{assumption}

\begin{assumption}
\label{Assumption_eigenvalue_uniform_bomega} 
There exists a constant $\varrho_{\diamond} > 0$ such that 
\begin{align*}
    \min_{k \in [K]} \min_{0 \neq \bv \in \cC(\cS_{\diamond}, 4)} \frac{\bv^\top \nabla^2 \cL_{k}(\bbeta^{\star}) \bv}{\|\bv\|_{2}^{2}} \geq \varrho_{\diamond}^{2},
\end{align*}
\end{assumption}

\begin{lemma}
\label{Lemma_bv_consistency}
Let Assumptions~\ref{Assumption_covariate}--\ref{Assumption_eigenvalue_uniform_bomega} hold. Take $\vartheta_{k}^{\diamond} = \cI_{0} \cA_{3, \diamond} \sqrt{(\log p)/m}$ for some constant $\cI_{0} \geq 4$, 
where $\cA_{3, \diamond} < \infty$ is a positive constant whose explicit expression is given in~\eqref{eq_A1_A2_definition}. Then, for each $k \in [K]$, we have
\begin{align}\label{eq_bound_bomegak}
    \|\hat{\bomega}_{k} - \bomega^{\star}\|_{1} = O_{\bbP}\left(|\cS_{\star}|\sqrt{\frac{\log p}{nK}} + |\cS_{\diamond}|\sqrt{\frac{\log p}{m}}\right). 
\end{align}
\end{lemma}

Lemma~\ref{Lemma_bv_consistency} demonstrates that the convergence rate of the estimator $\hat{\bomega}_{k}$ consists of two parts. The second term shares the same convergence rate as the optimal rate for the $\ell_{1}$-regularized quadratic minimization problem, and the first term corresponds to the cost induced by the estimation of $\bbeta^{\star}$.

\begin{theorem}
\label{Theorem_distributed_CLT}
Let Assumptions~\ref{Assumption_covariate}--\ref{Assumption_eigenvalue_uniform_bomega} hold. Assume that 
\begin{align*}
    \frac{\cM_{\diamond}}{m\sqrt{\bc^{\top} \bomega^{\star}}} \to 0 \enspace \mathrm{and} \enspace \frac{|\cS_{\star}| \log p}{\sqrt{m}} \vee \frac{|\cS_{\diamond}| \log p}{\sqrt{m}} \to 0. 
\end{align*}
Then, we have 
\begin{align}\label{eq_central_limit_theorem_c}
    \sup_{z\in\bbR}\left|\bbP\left\{\frac{\sqrt{n}\left(\tilde{\bc^{\top}\bbeta^{\star}} - \bc^{\top}\bbeta^{\star}\right)}{\sqrt{\bc^{\top} \bomega^{\star}}} \leq z\right\} - \Phi(z)\right| \to 0.
\end{align}
\end{theorem}

Practically the asymptotic variance $\bc^{\top} \bomega^{\star}$ is typically unknown. Here we propose a communication-efficient feasible distributed estimator for $\bc^{\top}\bomega^{\star}$. Recall that $\hat{\bomega}_{k}$ is a $\ell_{1}$-regularized estimator for $\bomega^{\star}$ given in~\eqref{eq_estimation_omega} for each $k \in [K]$. Inspired by the similar idea of debiasing the linear functional in~\eqref{eq_Debias_Linear_Functional}, we define a distributed debiased estimator for $\bc^{\top}\bomega^{\star}$ based on $\hat{\bomega}_{1}, \ldots, \hat{\bomega}_{K}$ as  
\begin{align}
\label{eq_variance_estimation_linear}
    \hat{\bc^{\top}\bomega^{\star}} = \frac{1}{K} \sum_{k=1}^K \left\{2\bc^{\top}\hat{\bomega}_{k} - \hat{\bomega}_{k}\nabla^{2} \mathcal{L}_{k}(\hbbeta)\hat{\bomega}_{k}\right\}. 
\end{align}

\begin{lemma}
\label{Lemma_consistency_sigma}
Under the conditions of Lemma~\ref{Lemma_bv_consistency}, we have 
\begin{align}\label{eq_sigma_consistency_c}
    \left|\frac{\hat{\bc^{\top}\bomega^{\star}}}{\bc^{\top} \bomega^{\star}} - 1\right| = O_{\bbP}\left(\sqrt{\frac{|\cS_{\star}|\log p}{n}} + \frac{|\cS_{\diamond}|\log p}{m}\right). 
\end{align}
\end{lemma}

Theorem~\ref{Theorem_distributed_CLT} and Lemma~\ref{Lemma_consistency_sigma} together yield
\begin{align*}
    \sup_{z\in\bbR}\left|\bbP\left\{\frac{\sqrt{n}\left(\tilde{\bc^{\top}\bbeta^{\star}} - \bc^{\top}\bbeta^{\star}\right)}{\sqrt{\hat{\bc^{\top} \bomega^{\star}}}} \leq z\right\} - \Phi(z)\right| \to 0.
\end{align*}
Consequently, for any significance level $\alpha \in (0, 1)$, a $100(1 - \alpha)\%$ confidence interval for the linear functional $\bc^{\top} \bbeta^{\star}$ is defined by 
\begin{align*}
    \mathbb{CI}_{1 - \alpha}\left(\bc^{\top}\bbeta^{\star}\right) = \left[\tilde{\bc^{\top}\bbeta^{\star}} - z_{1 - \alpha/2}\sqrt{\frac{\hat{\bc^{\top}\bomega^{\star}}}{n}},\ \tilde{\bc^{\top}\bbeta^{\star}} + z_{1 - \alpha/2}\sqrt{\frac{\hat{\bc^{\top}\bomega^{\star}}}{n}}\right],
\end{align*}
where $z_{1 - \alpha/2}$ is the $(1 - \alpha/2)$-th quantile of the standard normal distribution. In particular, by taking $\bc$ to be a canonical basis vector, our methodology yields distributed confidence intervals for the individual parameters $\beta_{j}^{\star}$, $j = 1, \ldots, p$.

\subsection{Decorrelated score test}
Denote $\bbeta^{\star} = (\nu^{\star}, \bgamma^{\star\top})^{\top} \in \bbR^{p}$, where $\nu^{\star} \in \bbR$ and $\bgamma^{\star} \in \bbR^{p - 1}$. Without loss of generality, we consider testing the hypothesis
\begin{align}
\label{eq_hypothesis_testing}
    H_{0} : \nu^{\star} = 0 \enspace \mathrm{versus} \enspace H_{1} : \nu^{\star} \neq 0.   
\end{align}
In the low-dimensional setting, we can utilize the traditional score test. Specifically, let $\hat{\bgamma}_{0} = \arg\min_{\bgamma \in \bbR^{p - 1}} \cL(0, \bgamma)$ denote the maximum partial likelihood estimate for $\bgamma^{\star}$ under the constraint that $\nu = 0$. Under regularity conditions and the null hypothesis, it can be shown that 
\begin{align}
\label{eq_Traditional_Score_Test_Null_Distribution}
    \frac{\sqrt{n}\nabla_{\nu} \cL(0, \hat{\bgamma}_{0})}{\sigma_{\nu}} \todist \cN(0, 1), \enspace \mathrm{as} \enspace n \to \infty,  
\end{align}
where $\nabla_{\nu}$ means taking partial derivative with respect to the variable $\nu$ and $\sigma_{\nu}^{2} = \cH_{\nu\nu}^{\star} - \cH_{\bgamma\nu}^{\star\top} \cH_{\bgamma\bgamma}^{\star-1} \cH_{\bgamma\nu}^{\star}$ is the asymptotic variance. However, in high dimensions, the asymptotic distribution in~\eqref{eq_Traditional_Score_Test_Null_Distribution} becomes intractable with $\hat{\bgamma}_{0}$ replaced by some commonly-used regularized estimator for $\bgamma^{\star}$~\citep{MR3611489, Fang2017testing}. To resolve this issue, \citet{Fang2017testing} introduced the following decorrelated score test statistic for $\bbeta = (\nu, \bgamma^{\top})^{\top}$, 
\begin{align}
\label{eq_Population_Decorrelated_Score_Test_Statistic}
    \pi^{\star}(\nu, \bgamma) = \nabla_{\nu} \cL(\nu, \bgamma) - \bw^{\star \top} \nabla_{\bgamma} \cL(\nu, \bgamma), 
\end{align}
where $\bw^{\star}\in\bbR^{p - 1}$ is the solution of $\cH_{\bgamma\bgamma}^{\star}\bw^{\star} = \cH_{\bgamma\nu}^{\star}$. In the distributed setting, the primary goal of this section is to propose a distributed decorrelated score test for~\eqref{eq_hypothesis_testing}.

In what follows, we write $\hat{\bbeta} = (\hat{\nu}, \hat{\bgamma}^{\top})^{\top}$ and assume that $\bw^{\star}$ is sparse and denote its support by $\cS_{\sharp}$. Then, for the $k$-th center, we estimate $\bw^{\star}$ via 
\begin{align}
\label{eq_bw_LASSO}
    \hat{\bw}_{k} = \underset{\bw\in\bbR^{p - 1}}{\argmin}  \left\{\bw^{\top} \nabla_{\bgamma \bgamma}^2 \cL_{k}(\hbbeta)\bw - 2\bw^{\top} \nabla_{\nu \bgamma}^{2} \cL_{k}(\hbbeta) + \vartheta_{k}^{\sharp} \|\bw\|_{1}\right\},
\end{align}
where $\vartheta_{k}^{\sharp} > 0$ is a regularization parameter. Motivated by~\eqref{eq_Population_Decorrelated_Score_Test_Statistic}, the distributed decorrelated score test statistic for~\eqref{eq_hypothesis_testing} is then defined by 
\begin{align*}
    \bar{\pi}(0, \hat{\bgamma}) = \frac{1}{K}\sum_{k = 1}^{K} \left\{\nabla_{\nu} \tilde{\cL}_{k}(0, \hat{\bgamma}) - \hat{\bw}_{k}^{\top} \nabla_{\bgamma} \tilde{\cL}_{k}(0, \hat{\bgamma})\right\}, 
\end{align*}
where $\tilde{\cL}_{k}(\bbeta) = \cL_{k}(\bbeta) - \{\nabla \cL_{k} (\tilde{\bbeta}) - \hat{\nabla \cL} (\tilde{\bbeta})\}^{\top} \bbeta$, $k = 1, \ldots, K$, are the estimated GEL functions across the centers. Denote $\tilde{\bw}^{\star} = (1, -\bw^{\star\top})^{\top}\in \bbR^{p}$.

\begin{assumption}
\label{Assumption_cM_bw}
There exists a positive constant $\cM_{\sharp} < \infty$ such that
\begin{align*}
    \max_{i \in [n]} \sup_{t\in [0, \tau]} |\bx_{i}(t)^{\top} \tilde{\bw}^{\star}| \leq \cM_{\sharp}.  
\end{align*}
\end{assumption}

\begin{assumption}
\label{Assumption_eigenvalue_uniform_bw}
There exists a constant $\varrho_{\sharp} > 0$ such that 
\begin{align*}
    \min_{k \in [K]} \min_{0 \neq \bv \in \cC(\cS_{\sharp}, 4)} \frac{\bv^\top \nabla^2 \cL_{k}(\bbeta^{\star}) \bv}{\|\bv\|_{2}^{2}} \geq \varrho_{\sharp}^{2},
\end{align*}
\end{assumption}

In the following theorem, we establish the central limit theorem for the proposed test statistic $\bar{\pi}(0, \hat{\bgamma})$. Observe that $\sigma_{\nu}^{2} = \cH_{\nu\nu}^{\star} - \cH_{\bgamma\nu}^{\star\top}\bw^{\star}$ as $\bw^{\star} = \cH_{\bgamma\bgamma}^{\star-1} \cH_{\bgamma\nu}^{\star}$. Similar to~\eqref{eq_variance_estimation_linear}, our distributed estimator for $\sigma_{\nu}^{2}$ is defined by 
\begin{align*}
    \hat{\sigma}_{\nu}^{2} = \frac{1}{K}\sum_{k = 1}^{K} \left\{\nabla_{\nu\nu}^{2} \cL_{k}(\hbbeta) - 2\nabla_{\bgamma\nu}^{2}\cL_{k}(\hbbeta)^\top \hat{\bw}_{k} + \hat{\bw}_{k}^{\top}\nabla_{\bgamma\bgamma}^{2}\cL_{k}(\hbbeta)\hat{\bw}_{k}\right\}. 
\end{align*}

\begin{theorem}
\label{Theorem_test_average}
Let Assumptions~\ref{Assumption_covariate}--\ref{Assumption_at_risk} and~\ref{Assumption_cM_bw}--\ref{Assumption_eigenvalue_uniform_bw} hold. Assume that 
\begin{align}\label{eq_clt_cond}
    \frac{\cM_{\sharp}}{m\sigma_{\nu}} \to 0 \enspace \mathrm{and} \enspace \frac{|\cS_{\star}| \log p}{\sqrt{m}} \vee \frac{|\cS_{\sharp}| \log p}{\sqrt{m}} \to 0.
\end{align}
Then, under the null hypothesis $\nu^{\star} = 0$, we have 
\begin{align}
\label{eq_CLT_pi_nu}
    \sup_{z\in\bbR}\left|\bbP\left\{\frac{\sqrt{n}\bar{\pi}(0, \hat{\bgamma})}{\hat{\sigma}_{\nu}} \leq z\right\} - \Phi(z)\right| \to 0. 
\end{align}
\end{theorem}

For any significance level $\alpha \in (0, 1)$, the level-$\alpha$ test for~\eqref{eq_hypothesis_testing} is defined by
\begin{align*}
    \psi_{\alpha} = \mathbb{I}\left\{\left|\frac{\sqrt{n}\bar{\pi}(0, \hat{\bgamma})}{\hat{\sigma}_{\nu}}\right| > z_{1 - \alpha/2}\right\},
\end{align*}
and we reject the null hypothesis $H_{0}$ whenever $\psi_{\alpha} = 1$.

\begin{remark}
    It is worth mentioning that unlike the bias-corrected estimator $\tilde{\bc^{\top}\bbeta^{\star}}$ in~\eqref{eq_Debias_Linear_Functional}, which relies on the KKT condition~\citep{vandeGeer2014asymptotically} of~\eqref{eq_iterated_estimator} with $t = T$, construction of the decorrelated score test statistic $\bar{\pi}(0, \hat{\bgamma})$ is more flexible, allowing for the use of any feasible estimator under the distributed setting. In particular, the distribution theory in Theorem~\ref{Theorem_test_average} remains valid for $\bar{\pi}(0, \bgamma^{\diamond})$ as long as $\bgamma^{\diamond} \in \mathbb{R}^{p - 1}$ exhibits the same convergence rate as that of $\hat{\bgamma}$. The interested readers are referred to~\citet{MR3611489}, \citet{Fang2017testing} and~\citet{ Battey2018distributed} for further discussions.
    \qed
\end{remark}

\section{Numerical Experiments}
\label{sec_experiments}
In this section, we present the results of simulations and real data analysis assessing the performance of our iterative distributed estimation and inference procedures, and their comparison with the full-sample, the one-center, the average and the average-debiased  LASSO estimators, respectively denoted by $\hat{\bbeta}_{\vartheta}$, $\bbeta_0$, $\bbeta_a$ and $\bbeta_d$. We now define these estimators. The full-sample estimator is the benchmark, and uses all the dataset at once. Without loss of generality, the one-center estimator uses only the data from the first center. The average estimator is defined as a simple average of $K$ local estimators. Due to the lack of literature on communication-efficient distributed learning for Cox's model, we construct an estimator corresponding to the idea of averaging the debiased estimators, and we will see that our own procedures have superior performance. More precisely, on the one hand~\citet{Battey2018distributed} proposed distributed inference, but for generalized linear models; and on the other hand~\citet{Yu2021confidence} studied inference for Cox's model using the CLIME estimator~\citep{Cai2011constrained}, but in a non-distributed setting. Note that the computational complexity for this estimator is huge due to the CLIME step. For example, with the setting of Section~\ref{sec_simulations} and $K = 8$ centers, the median computation time for estimation followed by hypothesis testing is 16 seconds on each center for our communication-efficient distributed algorithm, while it is 196 seconds on each center for the average-debiased estimator. The full-sample estimation and hypothesis testing process, using all data at once, has a median computation time of 18 seconds. We merge the approaches from Section 3.1.2 of~\citet{Battey2018distributed} and Section 2.2 of~\citet{Yu2021confidence} to construct the average-debiased estimator for Cox's model.

\subsection{Simulations}
\label{sec_simulations}
To be able to compute the iterated estimators $\bbeta_{t}$ in~\eqref{eq_iterated_estimator}, we cannot simply use~\texttt{glmnet} due to the linear correction term appearing in the gradient-enhanced loss. Instead, we extend a penalized weighted least squares algorithm, inspired by~\citet{Simon2011regularization}, and we incorporate the linear correction.
For our inference algorithms, we also use~\texttt{CVXR}~\citep{Fu2020cvxr} when minimizing the penalized quadratic forms~\eqref{eq_estimation_omega} and~\eqref{eq_bw_LASSO}.

For our experiments based on simulated data, the number of centers is $K \in \{2, 4, 8\}$. The covariates are independently drawn from the $p$-dimensional standard multivariate Gaussian distribution where each entry's absolute value is clipped at one. We use the sample size $n = 1000$, dimension $p = 50$, and a constant baseline hazard function $\lambda_{0}(t) = 1$. Therefore the conditional survival time $T\mid\bx$ is an exponential random variable with parameter $\exp(\bx^{\top} \bbeta^{\star})$ where $\bx$ is the covariate vector. We set the distribution of the conditional censoring time $C\mid \bx$ to be exponential with parameter $\frac{3}{7} \exp(\bx^{\top} \bbeta^{\star})$. This ensures approximately $30\%$ of the samples to be censored on average. We simulate $400$ i.i.d replications. In the Appendix, additional simulated experiments for $n = 240$ with a $50\%$ censoring rate and $p = 300$ demonstrate the performance of our algorithms.

\subsubsection{Estimation}\label{subsubsection-estimation}
We consider $\bbeta^{\star} = (0,2,2,2,{\bf{0}}_{p-4}^{\top})^{\top}$, where ${\bf{0}}_{p-4}$ is the zero vector of dimension $p-4$. In Figure~\ref{fig_estimation}, we display median estimation errors as a function of the iteration $t \in \{0,\dots,10\}$. In Figure~\ref{fig_estimation_iterated}, the decreasing curves are $\|\bbeta_{t} - \bbeta^{\star}\|_2$ for 3 values of $K$, where $K$ represents the number of centers. As expected from our theory, the more we iterate, the lower the estimation error, and the estimator tends to be as good as the full-sample benchmark when sufficiently iterated. In Figure~\ref{fig_estimation_all}, we display the median estimation error of multiple estimators for $K = 8$ centers; our procedure improves upon all.

\newcommand{\subfigfracinest}{0.9}
\begin{figure}[H]
\centering
    \begin{subfigure}{\subfigfracinest\linewidth}
    \includegraphics[width=1\linewidth]{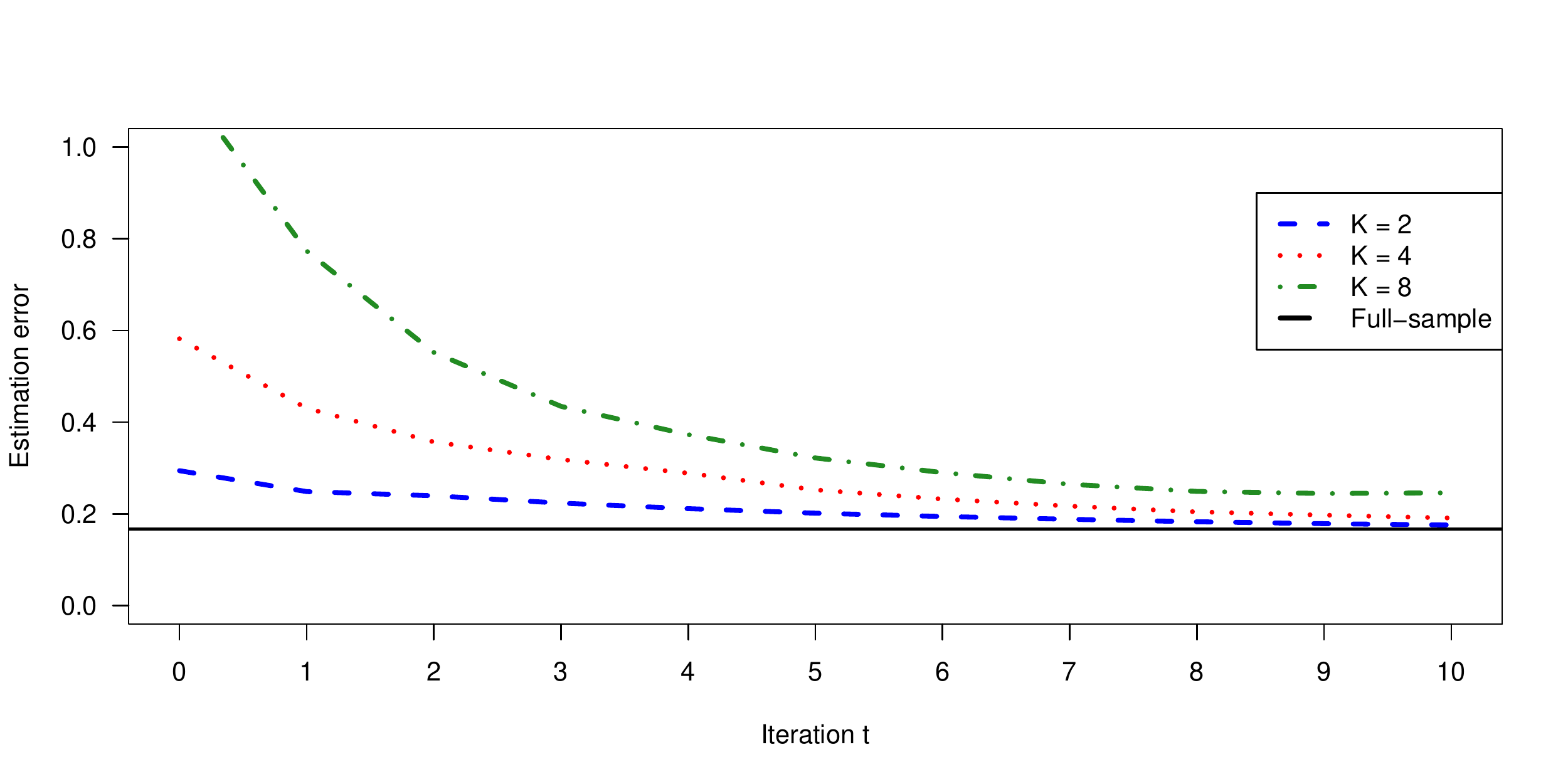}
    \caption{}
    \label{fig_estimation_iterated}
    \end{subfigure}
    
    \begin{subfigure}{\subfigfracinest\linewidth}
    \includegraphics[width=1\linewidth]{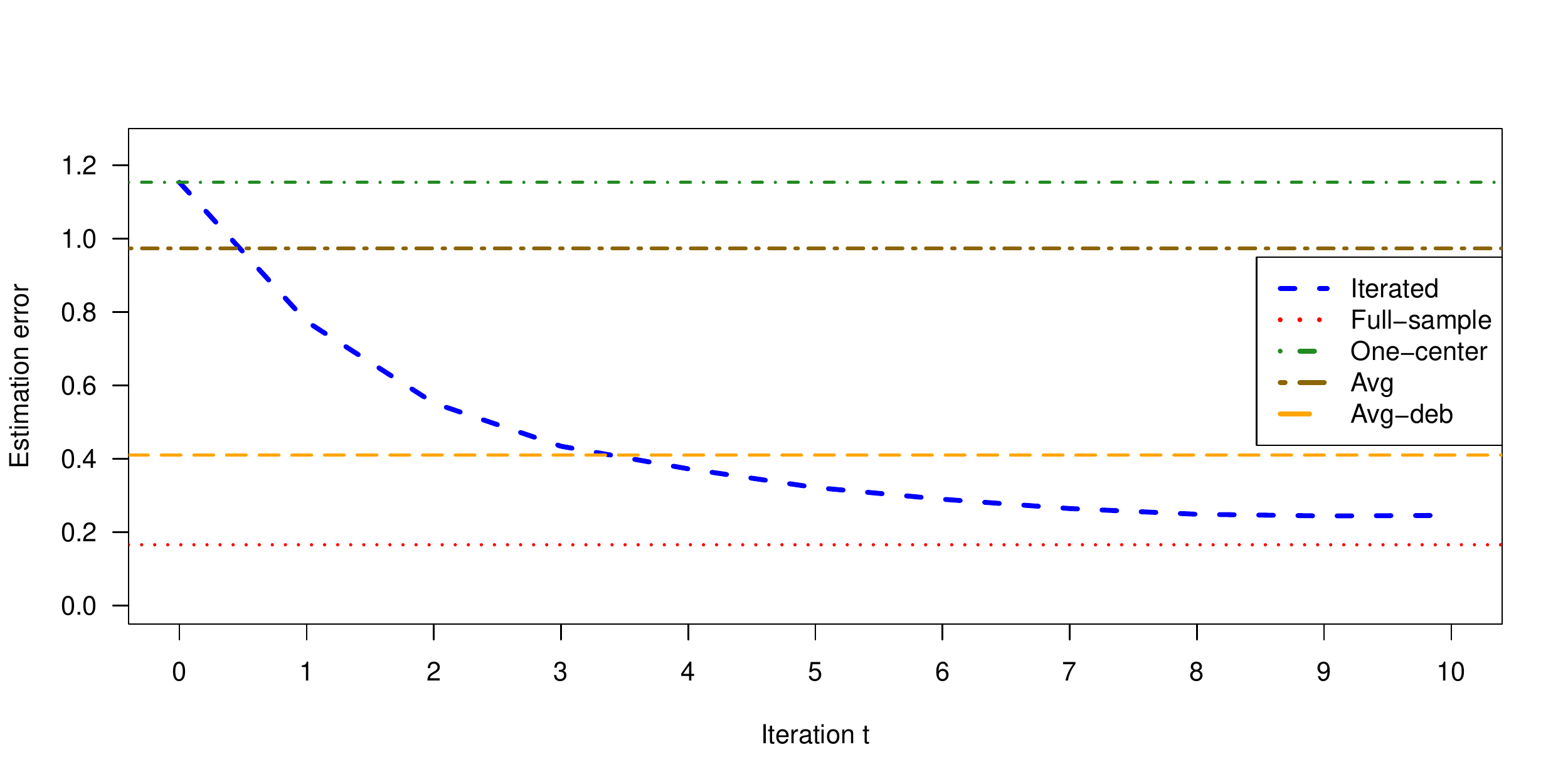}
    \caption{}
    \label{fig_estimation_all}
    \end{subfigure}
    
    \caption{\small{Median estimation error as a function of the iteration $t$. (a) For our iterated estimator, with 3 choices of $K$. The horizontal solid black line is the benchmark $\|\hbbeta_{\vartheta} - \bbeta^{\star}\|_2$, which has access to the full data.} \small{(b) For multiple estimators, with $K = 8$ centers.}}
    \label{fig_estimation}
\end{figure}

\subsubsection{Hypothesis testing}\label{subsubsection-hyptesting}
For the inference part, we take $\bbeta^{\star} = (\nu^{\star}, \bgamma^{\star\top})^{\top}$, where $\nu^{\star} \in \bbR$ and $\bgamma^{\star} = (2,2,2,{\bf{0}}_{p-4}^{\top})^{\top} \in \bbR^{p-1}$. We consider testing the hypothesis
\begin{align*}
    H_0 : \nu^{\star} = 0 \enspace \mathrm{versus} \enspace H_{1} : \nu^{\star} \neq 0.  
\end{align*}
We compare our test with three other tests: the decorrelated score test based on the benchmark estimator $\hbbeta_{\vartheta}$, the decorrelated score test based on the first center estimator $\bbeta_0$ which uses $m = n/K$ points, and the test based on the average-debiased estimator $\bbeta_d$.

\newcommand{\subfigfracin}{0.49}
\newcommand{\imgfrac}{1.0}
\newcommand{\imgfrachalf}{0.5}
\begin{figure}
\centering
    \begin{subfigure}{\subfigfracin\linewidth}
        \includegraphics[width=\imgfrac\linewidth]{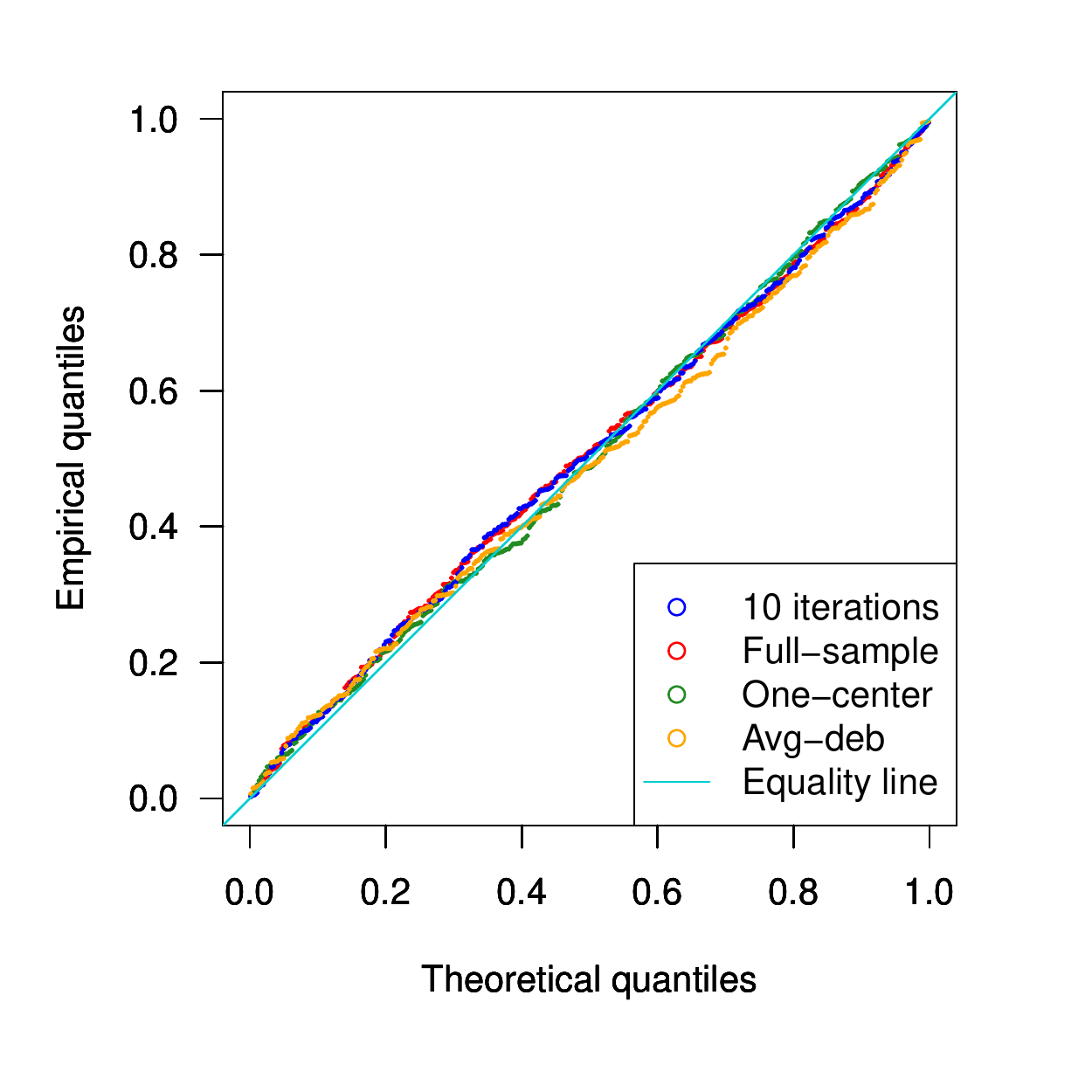}
        \caption{$K = 2$ centers}
    \end{subfigure}\hfill
    \begin{subfigure}{\subfigfracin\linewidth}
         \includegraphics[width=\imgfrac\linewidth]{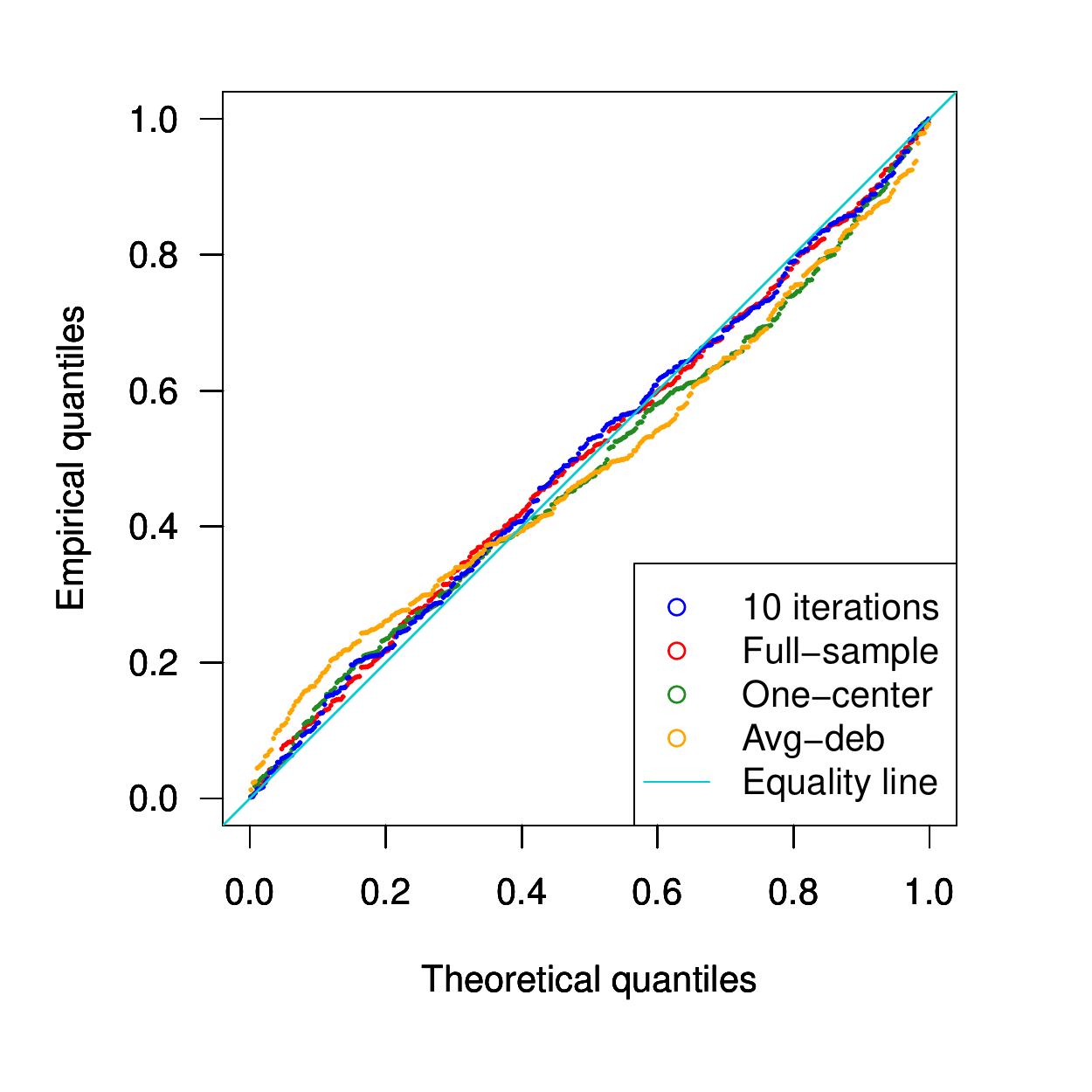}
         \caption{$K = 8$ centers}
    \end{subfigure}\hfill
    \caption{Q-Q plots of the $p$-values under $H_0$. ``$10$ iterations" (respectively ``Full-sample", ``One-center", ``Avg-deb") represents the test based on $\bbeta_{10}$ (respectively $\hbbeta_{\vartheta}$, $\bbeta_0$, $\bbeta_d$).}
    \label{fig_qqplot}
\end{figure}

Under $H_0$, the $p$-values theoretically follow a Uniform $(0,1)$ distribution. Figure~\ref{fig_qqplot} displays Q-Q plots of the $p$-values against their theoretical distribution, for $K 
\in \{2, 8\}$. We see that the distribution is indeed close to uniform.
The $p$-values coming from our testing methodology and from the full-sample estimator $\hbbeta_{\vartheta}$ follow the Uniform $(0,1)$ distribution more closely than the $p$-values coming from the one-center and the average-debiased estimators, and this is more visible when $K$ takes large values.

\begin{figure}
\centering
    \begin{subfigure}{\subfigfracin\linewidth}
        \includegraphics[width=\imgfrac\linewidth]{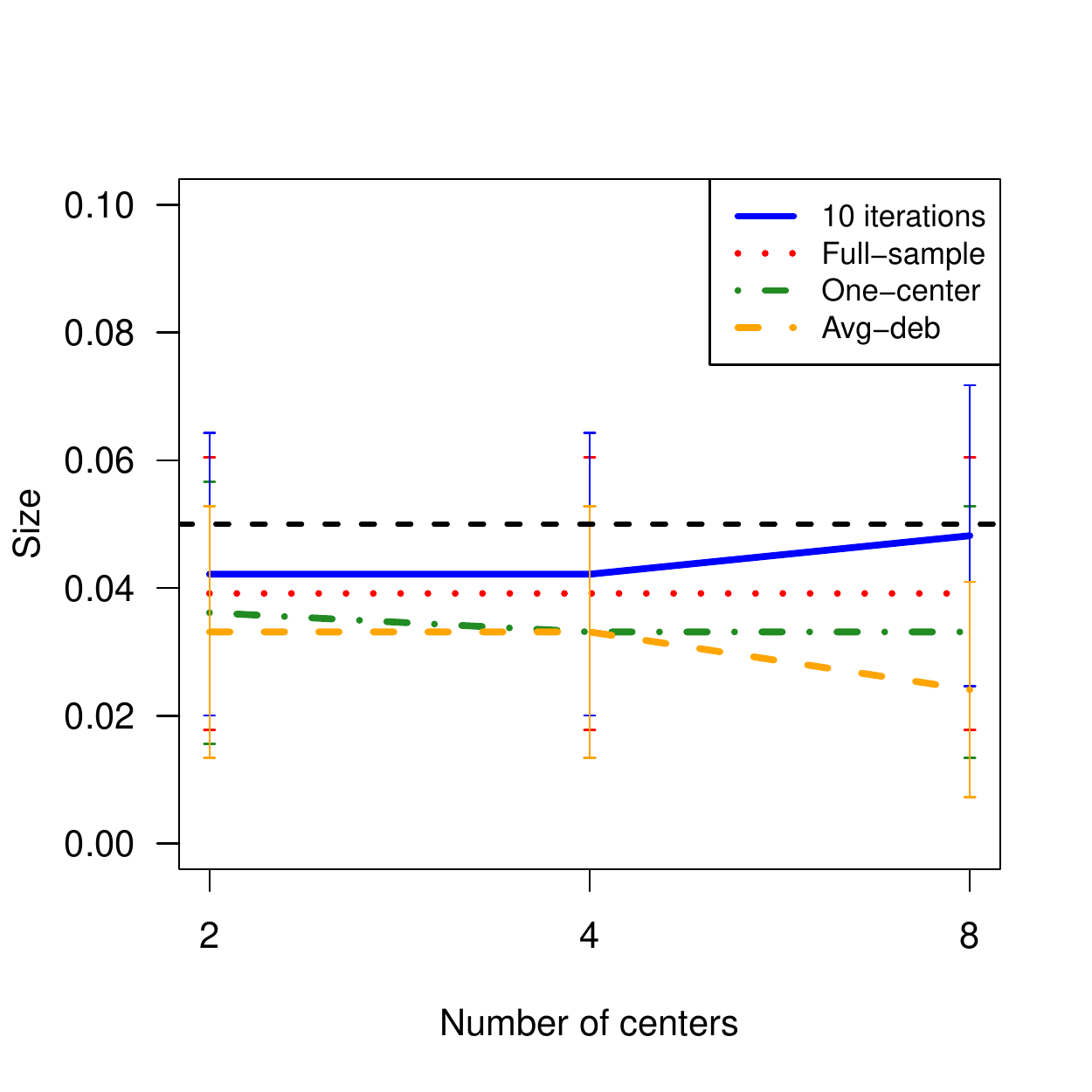}
        \caption{Size}
    \end{subfigure}\hfill
    \begin{subfigure}{\subfigfracin\linewidth}
         \includegraphics[width=\imgfrac\linewidth]{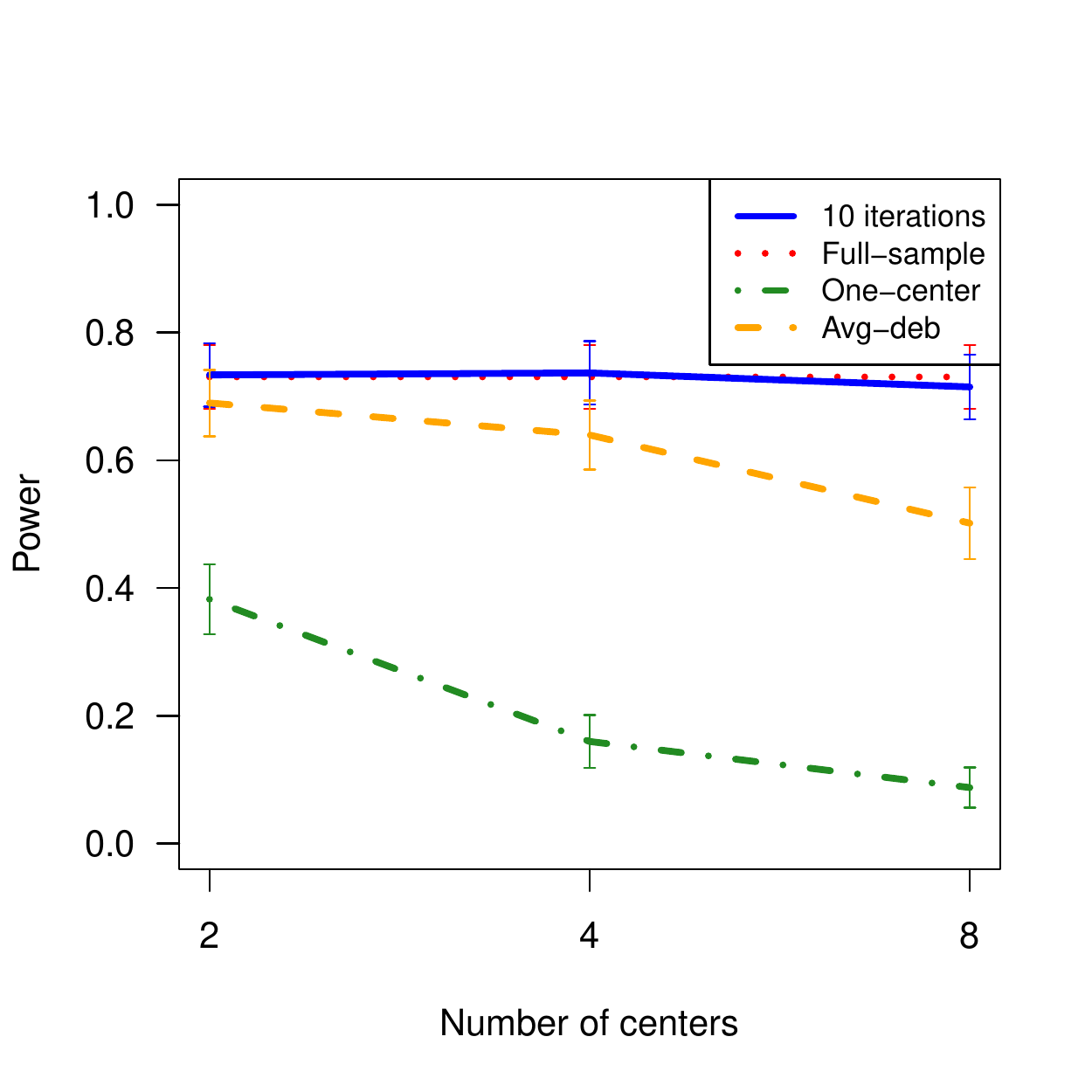}
         \caption{Power}
    \end{subfigure}\hfill
    \caption{Size and power as functions of the number of centers. ``$10$ iterations" (respectively ``Full-sample", ``One-center", ``Avg-deb") represents the test based on  $\bbeta_{10}$ (respectively $\hbbeta_{\vartheta}$, $\bbeta_0$, $\bbeta_d$).}
    \label{fig_inference}
\end{figure}

We now specify the significance level $\alpha = 0.05$. In Figure~\ref{fig_inference}, we show the estimated size and power of the various tests, along with $\pm 2$ standard error bars. The estimated size of all testing procedures are smaller than the nominal level, and our estimated test size is the closest to $\alpha$. To estimate the power of the various tests, we take $\nu^{\star} = 0.15$. For any number of centers, our test is as powerful as the test based on the full-sample estimator $\hbbeta_{\vartheta}$ and significantly more powerful than both the test based on $\bbeta_0$ and the one based on $\bbeta_d$.

\subsubsection{Linear functional}\label{subsubsection-linearfunc}
In this section, we present results obtained on the confidence interval for the linear functional $\bc^{\top} \bbeta^{\star}$. We take $\bbeta^{\star} = (0,2,2,2,{\bf{0}}_{p-4}^{\top})^{\top}$ and $\bc = (1,{\bf{0}}_{p-1}^{\top})^{\top}$, and consider $\alpha = 0.05$, which corresponds to a $95 \%$ confidence interval. The estimated coverage probability and median interval width based on $400$ replications are displayed in Figure~\ref{fig_linear-functional}. We note that all confidence intervals slightly overcover, with our confidence interval having an estimated coverage probability closer to $95 \%$ than the ones based on other estimators. The widths of the confidence intervals based on $\bbeta_{10}$, $\hat{\bbeta}_{\vartheta}$ and $\bbeta_d$ are very close, while $\bbeta_{0}$ delivers a substantially wider confidence interval, which is therefore less informative.

\begin{figure}
\centering
    \begin{subfigure}{\subfigfracin\linewidth}
        \includegraphics[width=\imgfrac\linewidth]{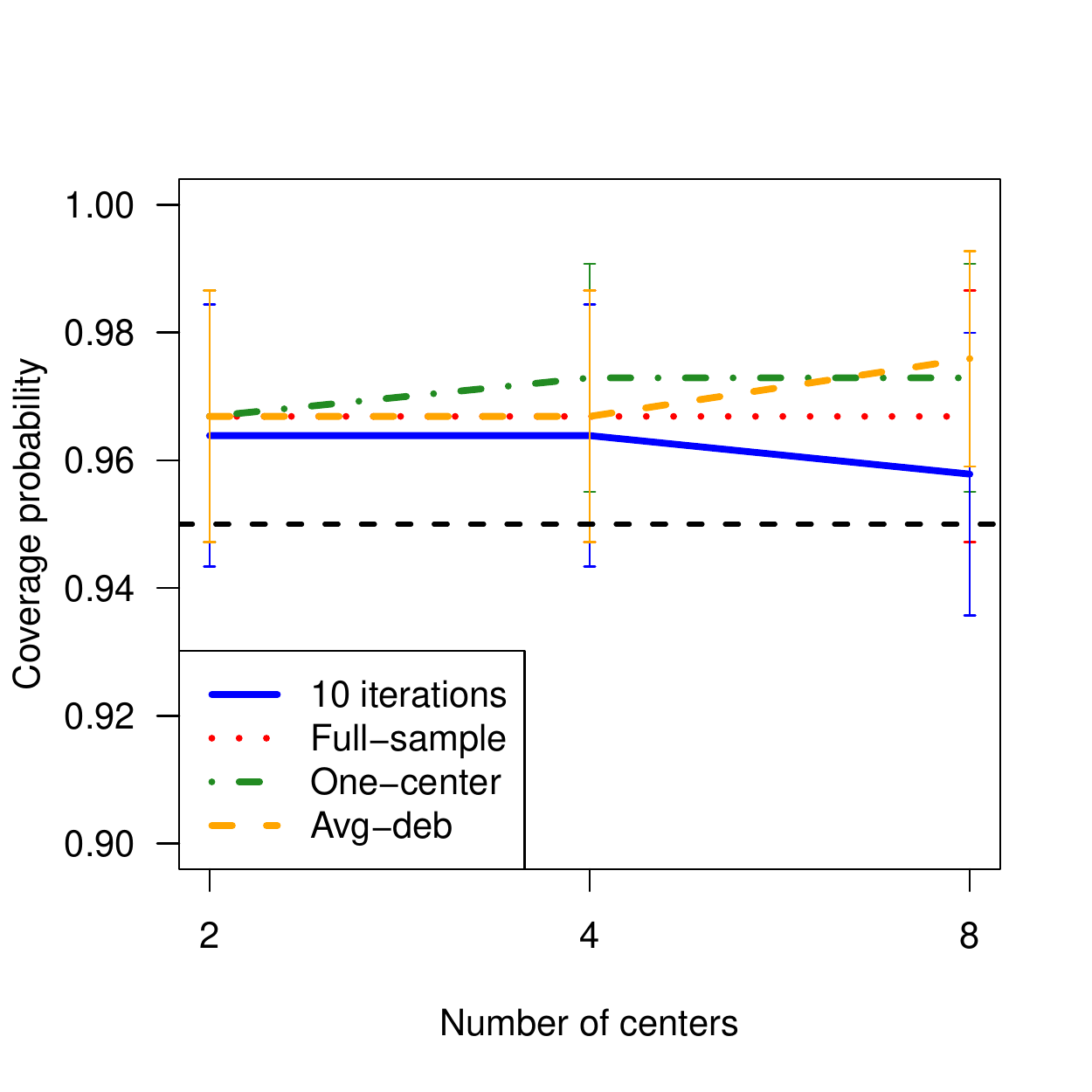}
        \caption{Coverage probability}
    \end{subfigure}\hfill
    \begin{subfigure}{\subfigfracin\linewidth}
        \includegraphics[width=\imgfrac\linewidth]{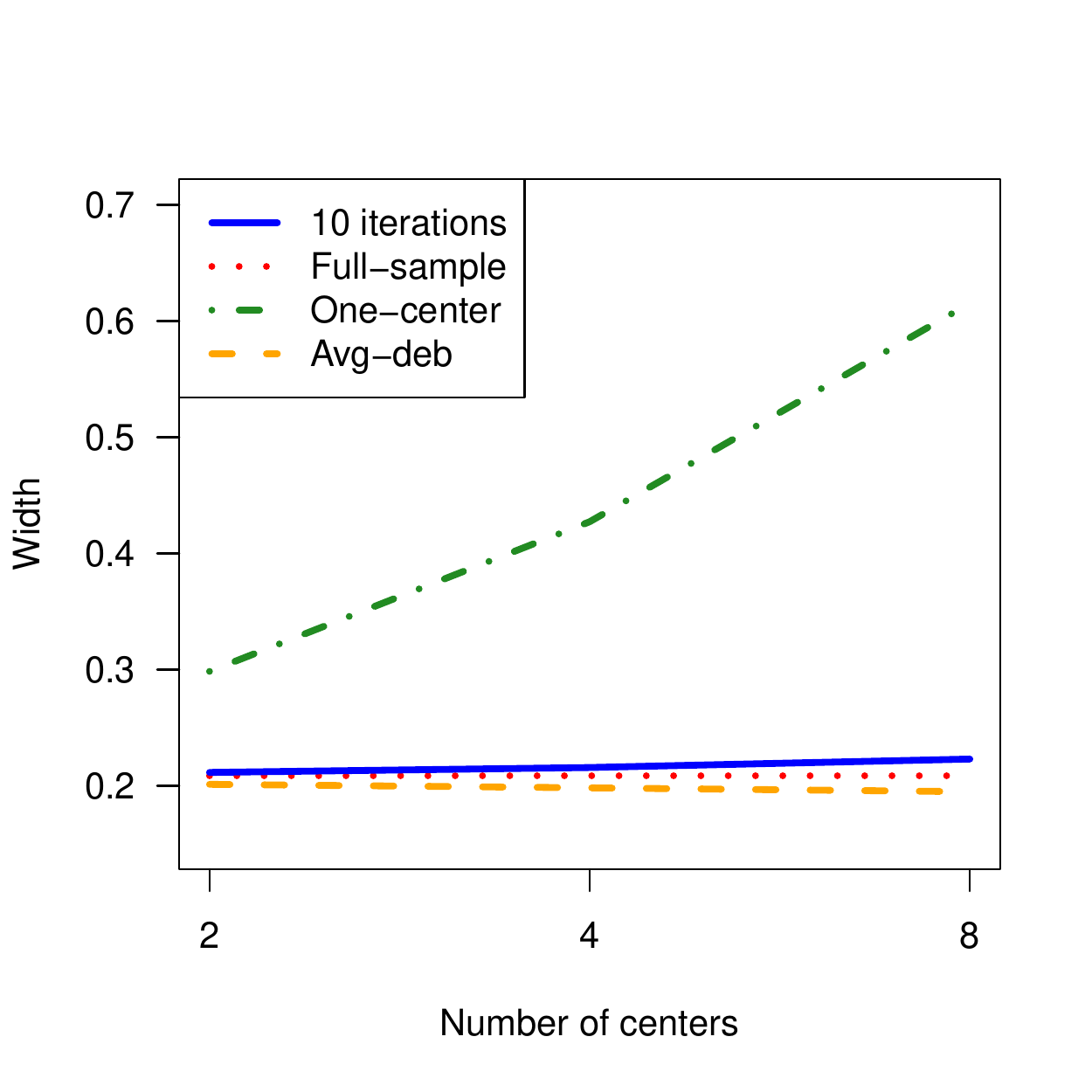}
        \caption{Confidence interval width}
    \end{subfigure}
    \caption{Estimated coverage probability and median interval width of the confidence interval for $\bc^{\top} \bbeta^{\star}$ and $\alpha = 0.05$, as a function of the number of centers. ``$10$ iterations" (respectively ``Full-sample", ``One-center", ``Avg-deb") represents the test based on  $\bbeta_{10}$ (respectively $\hbbeta_{\vartheta}$, $\bbeta_0$, $\bbeta_d$).}
    \label{fig_linear-functional}
\end{figure}

\subsection{Real data analysis}
We now use a real dataset, namely the diffuse large-B-cell lymphoma (DLBCL) dataset\footnote{The dataset is available at \url{https://llmpp.nih.gov/DLBCL/}.} of \citet{Rosenwald2002molecular} to evaluate our methodology. Gene-expression profiles, with a total of $7399$ microarray features, are related to survival time. A sample of $240$ patients with untreated DLBCL is available in this study, of which $138$ died. We preprocess the dataset as follows. The missing values for each covariate are replaced by the median of the observed values for this predictor, and we standardize the data. Moreover, we remove from our study the five patients who have follow-up time equal to zero.

We study the predictive performance of the procedures, using the concordance (C)-index~\citep{Harrell1982evaluating, Harrell1996multivariable} to this end. The C-index is the proportion of all usable patient pairs whose predictions and outcomes are concordant, that is when the patient whose estimated risk is lowest in the pair survives longer. Due to censoring, we use an inverse probability weighting adjustment~\citep{Uno2011OnTC}. C-index values are between $0$ and $1$, and $0.5$ would mean that a fitted model has no predictive discrimination. We randomly split the dataset into training and testing sets in 0.8 and 0.2 proportion. Using the training set, we then perform univariate screening to obtain the top $300$ predictors, and we apply the procedures on 4 centers. The out-of-sample C-index of the fitted models is computed thanks to the testing set. We repeat this whole procedure $400$ times, and the average values along with standard errors are given in Table~\ref{tab_C-index}. We observe that our iterated estimator has a better predictive performance than all other estimators.

\begin{table}[H]
\begin{center}
\caption{\small{Out-of-sample C-index.}}
\label{tab_C-index}
\footnotesize
\vspace{2mm}
\begin{tabular}{ccccc}
\hline\hline
 & {Iterated (3 iterations)} & {Full-sample} & {One-center} & {Average-debiased}
\\\cline{1-5}
{Average} & 0.583 & 0.581 & 0.510 & 0.533\\
{Standard error} & 0.003 & 0.003 & 0.002 & 0.003
\\\hline\hline
\end{tabular}
\end{center}
\end{table}

\section{Conclusion}
We proposed communication-efficient distributed estimation and inference algorithms for Cox's proportional hazards model. Under very mild conditions, we prove that our iterative estimator of the population parameter vector matches the ideal full-sample estimator after only a few rounds of communication. In the context of inference for linear functionals and hypothesis testing in the distributed framework, our central limit theorems and consistent variance estimators yield asymptotically valid confidence intervals and hypothesis tests. In addition to their validity, our intervals have smaller width and our tests are more powerful than alternative methods, and match the performance of the ones based on the full-sample estimator.

\bibliographystyle{apalike}
\bibliography{reference}

\newpage

\appendix

\begin{center}
\textbf{\LARGE{Appendix}}
\end{center}

\section{Additional notation}
We give here the expressions of the constants we defined in Section~\ref{sec_iterative} and Section~\ref{sec_inference}.
\begin{align}
\label{eq_A1_A2_definition}
    \cA_{1} &= 8\sqrt{6}B^2 + 32\sqrt{6}\exp(M)B^2 \Lambda_0(\tau) + \frac{768 R_{0}\exp(3M)B^2 \Lambda_0(\tau)}{\rho_{0}},\cr
    \cA_{2} &= \frac{8\exp(3M)B}{\rho_0} R_{1} + \frac{512\sqrt{6} R_{0}\exp(2M)B^3}{\rho_{0}} \{1 + \exp(M)\Lambda_{0}(\tau)\},\cr
    \cC_{0} &= \frac{3 c_{0} e B}{2\varrho_{\star}^{2}} + \frac{6\eta_{0} e B}{(1 - \cA_{0})\varrho_{\star}^{2}},\cr
    \cA_{3, \diamond} &= 2 B \cM_{\diamond} + 8 \exp(M) B \cM_{\diamond} \Lambda_{0}({\tau}) + \frac{64R_{0}\exp(3M)B\cM_{\diamond} \Lambda_{0}(\tau)}{\rho_{0}},\cr
    \cA_{3, \sharp} &= 2 B \cM_{\sharp} + 8 \exp(M) B \cM_{\sharp} \Lambda_{0}({\tau}) + \frac{64R_{0}\exp(3M)B\cM_{\sharp} \Lambda_{0}(\tau)}{\rho_{0}}.
\end{align}

We remind quantities appearing in $\cL(\bbeta)$, $\nabla \cL(\bbeta)$ and $\nabla^{2} \cL(\bbeta)$, and introduce their population counterparts.
\begin{align*}
    S^{(\ell)}(\bbeta, t) = \frac{1}{n} \sum_{i = 1}^{n} Y_{i}(t) \{\bx_{i}(t)\}^{\otimes \ell} \exp\{\bx_{i}(t)^{\top}\bbeta\} \enspace \mathrm{and} \enspace s^{(\ell)}(\bbeta, t) = \bbE \{S^{(\ell)}(\bbeta, t)\}, \enspace \mathrm{for} \enspace \ell = 0,1,2,
\end{align*}
\begin{align*}
    \cX(\bbeta, t) = \frac{S^{(1)}(\bbeta, t)}{S^{(0)}(\bbeta, t)} \enspace \mathrm{and} \enspace \be(\bbeta, t) = \frac{s^{(1)}(\bbeta, t)}{s^{(0)}(\bbeta, t)},
\end{align*}
\begin{align*}
    \bV(\bbeta, t) = \frac{S^{(2)}(\bbeta, t)}{S^{(0)}(\bbeta, t)} - \cX(\bbeta, t)^{\otimes 2} \enspace \mathrm{and} \enspace \bv(\bbeta, t) = \frac{s^{(2)}(\bbeta, t)}{s^{(0)}(\bbeta, t)} - \be(\bbeta, t)^{\otimes 2}.
\end{align*}

\section{Intermediate results and proofs}

In this section, we present results that will be used in the proofs of both Section~\ref{sec_iterative} and Section~\ref{sec_inference}.

\begin{lemma}[Lemma 3.2 in \citet{Huang2013oracle}]\label{Lemma_3.2_Huang}
For any vectors $\bbeta, \bomega \in \mathbb{R}^p$, let $D_{1}(\bbeta + \bomega, \bbeta) = \bomega^\top \{\nabla \cL_{1}(\bbeta + \bomega) - \nabla \cL_{1}(\bbeta)\}$. We have 
\begin{align*}
    \exp(-\eta_{\bomega}) \bomega^{\top} \nabla^{2} \cL_{1} (\bbeta) \bomega \leq D_{1}(\bbeta + \bomega, \bbeta) \leq \exp(\eta_{\bomega}) \bomega^{\top} \nabla^{2} \cL_{1} (\bbeta) \bomega,
\end{align*}
where $\eta_{\bomega} = \max_{t \in [0, \tau]} \max_{1\leq i, j \leq m} |\bomega^{\top}\{\bx_{i}(t) - \bx_{j}(t)\}|$. Moreover,
\begin{align*}
    \exp(-2\eta_{\bomega}) \nabla^{2} \cL_{1} (\bbeta) \preccurlyeq \nabla^{2} \cL_{1} (\bbeta + \bomega) \preccurlyeq \exp(2\eta_{\bomega}) \nabla^{2} \cL_{1} (\bbeta).
\end{align*}
\end{lemma}

\begin{lemma}[Lemma 3.3 in \citet{Huang2013oracle}]
\label{Lemma_3.3_Huang}
Let $\{\alpha_{i}(s), s \in [0, \tau]\}_{i = 1}^{n}$ be predictable processes with values in $[-1, 1]$ and $f_n(t) = n^{-1} \sum_{i=1}^n \int_{0}^{t} \alpha_{i}(s) d M_{i}(s)$. 
Then, for all $C_0 > 0$,
\begin{align*}
    \bbP\left(\sup_{0\leq t\leq \tau} |f_{n}(t)| > C_{0} x, \ \sum_{i = 1}^{n} \int_{0}^{\tau} Y_{i}(t) d N_{i}(t) \leq C_{0}^{2} n\right) \leq 2 \exp\left(-\frac{n x^{2}}{2}\right).
\end{align*}
\end{lemma}

\begin{lemma}\label{Lemma_HB}
Let $\bDelta(\bdelta) = \nabla^{2} \cL(\bbeta^{\star} + \bdelta) - \nabla^{2} \cL(\bbeta^{\star})$ and $\omega_{0} > 0$. Under Assumptions~\ref{Assumption_covariate}, \ref{Assumption_beta_1} and~\ref{Assumption_at_risk}, we have 
\begin{align}\label{eq_HB_l2}
    \bbP\left\{\sup_{0 < \|\bdelta\|_{1} \leq \omega_{0}}\frac{\|\bDelta(\bdelta) \bdelta\|_{\infty}}{\cD_{1}u\|\bdelta\|_{1}^{2} + \cD_{2}\|\bdelta\|_{2}^{2}} > \frac{4\exp(2M_{\diamond})B}{\rho_{0}}\right\} \leq 4 p^{2} \exp\left(- \frac{nu^{2}}{2}\right) + 2 \exp\left(-\frac{n \rho_{0}^{2}}{2}\right), 
\end{align}
where $M_{\diamond} = M + B \omega_{0}$, 
\begin{align*}
    \cD_{1} = 4 B^{2} \{1 + \exp(M) \Lambda_{0}(\tau)\} \enspace \mathrm{and} \enspace \cD_{2} = \exp(M) \lambda_{\max}\left\{\int_{0}^{\tau} \bSigma(t) d\Lambda_0(t)\right\}. 
\end{align*}
\end{lemma}

\begin{remark}
\label{remark_HB}
Under the same conditions as in Lemma~\ref{Lemma_HB}, for any vector $\balpha \in \bbR^{p}$ such that $\sup_{t \in [0, \tau]} |\bx(t)^{\top} \balpha| \leq M_{\balpha}$, we have 
\begin{align}
\label{equation_HB}
    \bbP\left\{\sup_{0 < \|\bdelta\|_{1} \leq \omega_{0}}\frac{\|\balpha^{\top}\bDelta(\bdelta) \bdelta\|_{\infty}}{\cD_{1}u\|\bdelta\|_{1}^{2} + \cD_{2}\|\bdelta\|_{2}^{2}} > \frac{4\exp(2M_{\diamond})M_{\balpha}}{\rho_{0}}\right\} \leq 4 p \exp\left(- \frac{nu^{2}}{2}\right) + 2 \exp\left(-\frac{n \rho_{0}^{2}}{2}\right), 
\end{align}
\qed
\end{remark}

\begin{proof}[Proof of Lemma~\ref{Lemma_HB}]
By Assumptions~\ref{Assumption_covariate} and \ref{Assumption_beta_1}, for any $\bbeta_{\theta} = \bbeta^{\star} + \theta \bdelta$ with $0\leq \theta \leq 1$,
\begin{align}
\label{eq_B_diamond}
    \max_{i \in [n]} |\bx_{i}(t)^\top \bbeta_{\theta}| &\leq \max_{i \in [n]} |\bx_{i}(t)^\top \bbeta^{\star}| + \max_{i \in [n]} \|\bx_{i}(t)\|_{\infty} \|\bbeta_{\theta} - \bbeta^{\star}\|_{1}\cr 
    &\leq M + \theta B \omega_{0} \leq M_{\diamond}. 
\end{align}
For simplicity of notation, we write $\Delta_{\infty}(\bdelta) = \|\{\nabla^2 \cL(\bbeta^{\star} + \bdelta) - \nabla^2 \cL(\bbeta^{\star})\}\bdelta\|_{\infty}$ and 
\begin{align*}
    \omega_{i}(\bbeta, t) = \frac{Y_{i}(t)\exp\{\bx_{i}(t)^\top \bbeta\}}{\sum_{\ell = 1}^{n} Y_{\ell}(t) \exp\{\bx_{\ell}(t)^\top \bbeta\}}, \enspace i = 1, \ldots, n. 
\end{align*}
With this notation, we have 
\begin{align*}
    \Delta_{\infty}(\bdelta) = \left\|\frac{1}{n}\int_{0}^{\tau} \{\bV(\bbeta^{\star} + \bdelta, t) - \bV(\bbeta^{\star}, t)\}\bdelta d\bar{N}(t)\right\|_{\infty},
\end{align*}
where for any $\bbeta\in\bbR^{p}$ and $t \in [0, \tau]$, 
\begin{align}\label{eq_bV}
    \bV(\bbeta, t) &= \sum_{i = 1}^{n} \omega_{i}(\bbeta, t) \bx_{i}(t)\bx_{i}(t)^{\top} - \left\{\sum_{i = 1}^{n} \omega_{i}(\bbeta, t) \bx_{i}(t)\right\}\left\{\sum_{\ell = 1}^{n} \omega_{\ell}(\bbeta, t) \bx_{\ell}(t)^\top\right\}\cr
    &= \sum_{i = 1}^{n} \omega_{i}(\bbeta, t)\left\{\bx_{i}(t) - \sum_{\ell = 1}^{n} \omega_{\ell}(\bbeta, t) \bx_{\ell}(t)\right\}^{\otimes 2}\cr 
    &=: \sum_{i = 1}^{n} \omega_{i}(\bbeta, t) \tilde{\bx}_{i}(\bbeta, t)^{\otimes 2}. 
\end{align}
Elementary calculations imply that the gradient of $\omega_{i}(\bbeta, t)$ with respect to $\bbeta$ is 
\begin{align}\label{eq_nabla_vartheta}
    \nabla \omega_{i}(\bbeta, t) = \omega_{i}(\bbeta, t) \left\{\bx_{i}(t) - \sum_{\ell = 1}^{n} \omega_{\ell}(\bbeta, t)\bx_{\ell}(t)\right\} = \omega_{i}(\bbeta, t) \tilde{\bx}_{i}(\bbeta, t).  
\end{align}
Write $\{\bV(\bbeta^{\star} + \bdelta, t) - \bV(\bbeta^{\star}, t)\}\bdelta = \{\varepsilon_{1}(t), \ldots, \varepsilon_{p}(t)\}^{\top}$. Then, for each $j \in [p]$, by the mean value theorem, 
\begin{align*}
    \varepsilon_{j}(t) = \sum_{k = 1}^{p} \{V_{jk}(\bbeta^{\star} + \bdelta, t) - V_{jk}(\bbeta^{\star}, t)\}\delta_{k} = \sum_{k = 1}^{p} \bdelta^{\top} \nabla V_{jk}(\bbeta^{\star} + \theta_{j}\bdelta, t) \delta_{k}
\end{align*}
for some $\theta_{j}\in [0, 1]$. In view of~\eqref{eq_bV}, we have 
\begin{align*}
    \varepsilon_{j}(t) &= \sum_{k = 1}^{p} \bdelta^{\top} \left\{\sum_{i = 1}^{n} \nabla \omega_{i}(\bbeta^{\star} + \theta_{j} \bdelta, t) \tilde{x}_{ij}(\bbeta^{\star} + \theta_j \bdelta, t)\tilde{x}_{ik}(\bbeta^{\star} + \theta_j \bdelta, t)\right\} \delta_{k}\cr
    &+\sum_{k = 1}^{p} \bdelta^\top \left\{\sum_{i = 1}^{n} \omega_{i}(\bbeta^{\star} + \theta_j \bdelta, t) \nabla \tilde{x}_{ij}(\bbeta^{\star} + \theta_{j} \bdelta, t)\tilde{x}_{ik}(\bbeta^{\star} + \theta_{j} \bdelta, t)\right\} \delta_{k}\cr
    &+\sum_{k = 1}^{p} \bdelta^\top \left\{\sum_{i = 1}^{n} \omega_{i}(\bbeta^{\star} + \theta_j \bdelta, t) \tilde{x}_{ij}(\bbeta^{\star} + \theta_j \bdelta, t) \nabla \tilde{x}_{ik}(\bbeta^{\star} + \theta_j \bdelta, t)\right\} \delta_{k}.
\end{align*}
For each $j \in [p]$, we observe that  
\begin{align}\label{eq_weighted_mean_0}
    \sum_{i = 1}^{n} \omega_{i}(\bbeta^{\star} + \theta_{j} \bdelta, t) \tilde{x}_{ij}(\bbeta^{\star} + \theta_j \bdelta, t) = 0, 
\end{align}
and that the gradient vector  
\begin{align*}
    \nabla \tilde{x}_{ij}(\bbeta^{\star} + \theta_j \bdelta, t) &= \nabla \left\{x_{ij}(t) - \sum_{\ell = 1}^{n} \omega_{\ell}(\bbeta^{\star} + \theta_{j} \bdelta, t) x_{\ell j}(t)\right\}\cr 
    &= -\sum_{\ell = 1}^{n} \nabla \omega_{\ell}(\bbeta^{\star} + \theta_j \bdelta, t) x_{\ell j}(t)
\end{align*}
is independent of the index $i \in [n]$. Hence we have 
\begin{align*}
    &\sum_{k = 1}^{p} \bdelta^\top \left\{\sum_{i = 1}^{n} \omega_{i}(\bbeta^{\star} + \theta_j \bdelta, t) \nabla \tilde{x}_{ij}(\bbeta^{\star} + \theta_j \bdelta, t)\tilde{x}_{ik}(\bbeta^{\star} + \theta_j \bdelta, t)\right\} \delta_{k} = 0,\cr
    &\sum_{k = 1}^{p} \bdelta^\top \left\{\sum_{i = 1}^{n} \omega_{i}(\bbeta^{\star} + \theta_j \bdelta, t) \tilde{x}_{ij}(\bbeta^{\star} + \theta_j \bdelta, t) \nabla \tilde{x}_{ik}(\bbeta^{\star} + \theta_j \bdelta, t)\right\} \delta_{k} = 0.
\end{align*}
Consequently, by \eqref{eq_nabla_vartheta}, we obtain
\begin{align}\label{eq_xi_j}
    \varepsilon_{j}(t) &= \sum_{k = 1}^{p} \bdelta^\top \left\{\sum_{i = 1}^{n} \nabla \omega_{i}(\bbeta^{\star} + \theta_j \bdelta, t) \tilde{x}_{ij}(\bbeta^{\star} + \theta_j \bdelta, t)\tilde{x}_{ik}(\bbeta^{\star} + \theta_j \bdelta, t)\right\} \delta_{k}\cr
    &= \sum_{k = 1}^{p} \bdelta^\top \left\{\sum_{i = 1}^{n} \omega_{i}(\bbeta^{\star} + \theta_j \bdelta, t) \tilde{\bx}_{i}(\bbeta^{\star} + \theta_j \bdelta, t) \tilde{x}_{ij}(\bbeta^{\star} + \theta_j \bdelta, t) \tilde{x}_{ik}(\bbeta^{\star} + \theta_j \bdelta, t)\right\} \delta_{k}\cr
    &= \sum_{i = 1}^{n} \omega_{i}(\bbeta^{\star} + \theta_j \bdelta, t) \tilde{x}_{ij}(\bbeta^{\star} + \theta_j \bdelta, t) \left\{\tilde{\bx}_{i}(\bbeta^{\star} + \theta_j \bdelta, t)^\top \bdelta\right\}^{2}.
\end{align}
Recall that $\bbE_{0}(\cdot) = \cdot - \bbE(\cdot)$. By Assumption~\ref{Assumption_covariate} and \eqref{eq_weighted_mean_0},  
\begin{align*}
    \max_{j \in [p]}|\varepsilon_{j}(t)| &\leq 2B \max_{j \in [p]} \sum_{i = 1}^{n} \omega_{i}(\bbeta^{\star} + \theta_j \bdelta, t) \left\{\tilde{\bx}_{i}(\bbeta^{\star} + \theta_j \bdelta, t)^\top \bdelta\right\}^{2}\cr
    &\leq 2B \max_{j \in [p]} \sum_{i = 1}^{n} \omega_{i}(\bbeta^{\star} + \theta_{j} \bdelta, t) \left[\bbE_{0}\{\bx_{i}(t)\}^\top \bdelta\right]^{2},
\end{align*}
where the second inequality follows from the fact that a weighted average minimizes the empirical weighted mean squared error.
This implies that  
\begin{align*}
    \Delta_{\infty}(\bdelta) &\leq \frac{2B}{n} \int_{0}^{\tau} \max_{j \in [p]} \sum_{i = 1}^{n} \omega_{i}(\bbeta^{\star} + \theta_{j} \bdelta, t) \left[\bbE_{0}\{\bx_{i}(t)\}^\top \bdelta\right]^{2} d\bar{N}(t)\cr
    &\leq \frac{2B}{n} \max_{i \in [n]}\max_{j \in [p]} \sup_{t\in[0, \tau]} \omega_{i}(\bbeta^{\star} + \theta_j \bdelta, t) \int_{0}^{\tau} \sum_{i = 1}^{n} \left[\bbE_{0}\{\bx_{i}(t)\}^\top \bdelta\right]^{2} d \bar{N}(t). 
\end{align*}
We first upper bound $\omega_{i}(\bbeta^{\star} + \theta_{j}\bdelta, t)$. Define  
\begin{align}
\label{eq_event_E}
    \cE = \left\{\sum_{i = 1}^{n} Y_{i}(\tau) > \frac{n\rho_0}{2}\right\}. 
\end{align}
Since $Y_{1}(\tau), \ldots, Y_{n}(\tau)$ are i.i.d.~random variables with $0 \leq Y_{i}(\tau) \leq 1$ and $\bbP(Y_{1}(\tau) = 1) \geq \rho_{0}$, it follows from Hoeffding's inequality that 
\begin{align*}
    \bbP(\cE) \geq 1 - \exp\left(-\frac{n\rho_0^2}{2}\right).
\end{align*}
Under $\cE$, it follows from \eqref{eq_B_diamond} that  
\begin{align*}
    \max_{i \in [n]} \max_{j \in [p]} \sup_{t\in [0, \tau]} \omega_{i}(\bbeta^{\star} + \theta_{j}\bdelta, t) \leq \frac{\exp(2M_{\diamond})}{\sum_{\ell = 1}^{n} Y_{\ell}(\tau)} \leq \frac{2\exp(2M_{\diamond})}{n\rho_{0}}.
\end{align*}
Consequently, we obtain 
\begin{align*}
    \Delta_{\infty}(\bdelta) &\leq \frac{4\exp(2M_{\diamond})B}{n \rho_{0}} \int_{0}^{\tau} \bdelta^\top \hat{\bSigma}(t) \bdelta d\bar{N}(t), \enspace \mathrm{where} \enspace \hat{\bSigma}(t) = \frac{1}{n}\sum_{i = 1}^{n} \bbE_{0}\{\bx_{i}(t)\}\bbE_{0}\{\bx_{i}(t)\}^\top.
\end{align*}
Recall that $d M_i(t) = d N_i(t) - Y_i(t) \exp\{\bx_{i}(t)^\top \bbeta^{\star}\} d\Lambda_0(t)$ and $\bar{M}(t) = \sum_{i = 1}^{n} M_{i}(t)$. Hence $d \bar{N}(t) = d \bar{M}(t) + n S^{(0)}(\bbeta^{\star}, t) d\Lambda_{0}(t)$ and    
\begin{align*}
    \int_{0}^{\tau} \bdelta^\top \hat{\bSigma}(t) \bdelta d\bar{N}(t) &= \int_{0}^{\tau} \bdelta^\top \hat{\bSigma}(t) \bdelta d \bar{M}(t) + n\int_{0}^{\tau} \bdelta^\top \hat{\bSigma}(t) \bdelta S^{(0)}(\bbeta^{\star}, t) d\Lambda_0(t)\cr
    &\leq \left\|\int_{0}^{\tau} \hat{\bSigma}(t) d\bar{M}(t)\right\|_{\max} \|\bdelta\|_{1}^{2} + n\int_{0}^{\tau} \bdelta^\top \hat{\bSigma}(t) \bdelta S^{(0)}(\bbeta^{\star}, t) d\Lambda_0(t)\cr
    &=: \Delta_{1} \|\bdelta\|_{1}^{2} + \Delta_{2}(\bdelta). 
\end{align*}
We first bound $\Delta_{1}$. By Assumption~\ref{Assumption_covariate}, we have $\sup_{t\in [0, \tau]}\|\hat{\bSigma}(t)\|_{\max} \leq 4B^{2}$. Hence, for any $u > 0$, it follows from Lemma~\ref{Lemma_3.3_Huang} that 
\begin{align*}
    \bbP(\Delta_{1} > 4 B^{2} u) \leq 2 p^{2} \exp\left(-\frac{u^{2}}{2n}\right). 
\end{align*}
We now upper bound $\Delta_{2}(\bdelta)$. By Assumptions~\ref{Assumption_covariate} and \ref{Assumption_beta_1}, 
\begin{align*}
    \Delta_{2}(\bdelta) &\leq n\exp(M) \int_{0}^{\tau} \bdelta^\top \hat{\bSigma}(t) \bdelta d\Lambda_0(t) \cr
    &= n\exp(M) \int_{0}^{\tau}\bdelta^\top \left\{\hat{\bSigma}(t) - \bSigma(t)\right\} \bdelta d\Lambda_0(t) + n\exp(M) \int_{0}^{\tau} \bdelta^\top \bSigma(t) \bdelta d \Lambda_0(t)\cr
    &\leq n\exp(M) \left\|\int_{0}^{\tau} \left\{\hat{\bSigma}(t) - \bSigma(t)\right\} d\Lambda_0(t)\right\|_{\max} \|\bdelta\|_{1}^{2}\cr
    &+ n \exp(M) \lambda_{\max}\left(\int_{0}^{\tau} \bSigma(t) d\Lambda_0(t)\right) \|\bdelta\|_{2}^{2}.
\end{align*}
By Assumption~\ref{Assumption_covariate} and Hoeffding's inequality,
\begin{align*}
    \bbP\left(\left\|\int_{0}^{\tau} \left\{\hat{\bSigma}(t) - \bSigma(t)\right\} d\Lambda_0(t)\right\|_{\max} > 4 B^{2} \Lambda_{0}(\tau) u\right) \leq 2 p^{2} \exp\left(- \frac{n u^{2}}{2}\right).  
\end{align*}
Putting all these pieces together, we obtain~\eqref{eq_HB_l2}.
\end{proof}

\begin{lemma}\label{Lemma_Hessian}
Let Assumptions \ref{Assumption_beta_1} and \ref{Assumption_at_risk} hold. For any vectors $\bphi_{1}, \bphi_{2} \in \bbR^{p}$, we assume that there exist $M_{1} > 0$ and $M_{2} > 0$ such that 
\begin{align}
\label{eq_bound_M}
    \max_{i \in [n]} \sup_{t \in [0, \tau]} |\bx_{i}(t)^{\top} \bphi_{1}| \leq M_{1} \enspace \mathrm{and} \enspace \max_{i \in [n]} \sup_{t \in [0, \tau]} |\bx_{i}(t)^{\top} \bphi_{2}| \leq M_{2}. 
\end{align}
For some $\chi_{0} > 0$, we define 
\begin{align*}
    \cE_{0} = \left\{\inf_{t \in [0, \tau]} S^{(0)}(\bbeta^{\star}, t) > \chi_{0}\right\}. 
\end{align*}
Let $\bDelta_{n} = \nabla^{2} \cL(\bbeta^{\star}) - \cH^{\star}$ where $\cH^{\star}$ is defined in~\eqref{eq_L_star}. Then, for any $u, v > 0$,  we have 
\begin{align}\label{eq_Hessian_max}
    \bbP\left(\left\{|\bphi_{1}^{\top}\bDelta_{n} \bphi_{2}| > A_{1} u + A_{2} v\right\} \cap \cE_{0}\right) \leq 4\exp\left(-\frac{n u^{2}}{2}\right) + 4.442 \exp\left(-\frac{nv}{2 + 2\sqrt{v}/3}\right),
\end{align}
where 
\begin{align*}
    A_{1} = M_{1} M_{2} + 4 \exp(M) M_{1} M_{2} \Lambda_{0}(\tau) \enspace \mathrm{and} \enspace A_{2} = \frac{4\exp(2M) M_{1} M_{2} \Lambda_{0}(\tau)}{\chi_{0}}. 
\end{align*}
\end{lemma}

\begin{proof}[Proof of Lemma~\ref{Lemma_Hessian}]
For simplicity of notation, write 
\begin{align*}
    \bV_{\diamond}(\bbeta^{\star}, t) = \sum_{i = 1}^{n} \omega_{i}(\bbeta^{\star}, t) \{\bx_{i}(t) - \be(\bbeta^{\star}, t)\}^{\otimes 2}. 
\end{align*}
Recall that $\sum_{i = 1}^{n} \omega_{i}(\bbeta^{\star}, t)\tilde{\bx}_{i}(\bbeta^{\star}, t) = 0$. Hence   
\begin{align*}
    \bV(\bbeta^{\star}, t) = \bV_{\diamond}(\bbeta^{\star}, t) - \left\{\frac{S^{(1)}(\bbeta^{\star}, t)}{S^{(0)}(\bbeta^{\star}, t)} - \be(\bbeta^{\star}, t)\right\}^{\otimes 2}.
\end{align*}
Consequently, we obtain 
\begin{align*}
    \nabla^2 \cL(\bbeta^{\star}) &= \frac{1}{n}\int_{0}^{\tau} \bV(\bbeta^{\star}, t) d\bar{M}(t) + \int_{0}^{\tau} \bV(\bbeta^{\star}, t) S^{(0)}(\bbeta^{\star}, t) d\Lambda_0(t)\cr
    &= \frac{1}{n}\int_{0}^{\tau} \bV(\bbeta^{\star}, t) d\bar{M}(t) + \int_{0}^{\tau} \bV_{\diamond}(\bbeta^{\star}, t) S^{(0)}(\bbeta^{\star}, t) d\Lambda_0(t)\cr
    &- \int_{0}^{\tau} \left\{\frac{S^{(1)}(\bbeta^{\star}, t)}{S^{(0)}(\bbeta^{\star}, t)} - \be(\bbeta^{\star}, t)\right\}^{\otimes 2} S^{(0)}(\bbeta^{\star}, t) d\Lambda_0(t)\cr
    &=: \cW_{1}(\bbeta^{\star}) + \cH_{n}(\bbeta^{\star}) - \cW_{2}(\bbeta^{\star})
\end{align*}
In view of \eqref{eq_bV}, it follows from \eqref{eq_bound_M} that 
\begin{align*}
    \sup_{t \in [0, \tau]}|\bphi_{1}^{\top} \bV (\bbeta^{\star}, t) \bphi_{2}| \leq \sup_{t \in [0, \tau]} \sqrt{\left\{\bphi_{1}^{\top} \bV (\bbeta^{\star}, t) \bphi_{1}\right\}\left\{\bphi_{2}^{\top} \bV (\bbeta^{\star}, t) \bphi_{2}\right\}} \leq M_{1} M_{2}. 
\end{align*}
Then, by Lemma~\ref{Lemma_3.3_Huang}, for any $u > 0$, 
\begin{align*}
    \bbP\left(|\bphi_{1}^{\top} \cW_{1}(\bbeta^{\star}) \bphi_{2}| > M_{1} M_{2} u\right) \leq 2 \exp\left(-\frac{n u^{2}}{2}\right). 
\end{align*}
By \eqref{eq_bound_M} and Assumption~\ref{Assumption_beta_1}, 
\begin{align*}
    \bphi_{1}^{\top} \cH_{n}(\bbeta^{\star}) \bphi_{2} = \frac{1}{n}\sum_{i = 1}^{n}\int_{0}^{\tau} Y_{i}(t) \exp\{\bx_{i}(t)^\top \bbeta^{\star}\} \bphi_{1}^{\top}\{\bx_{i}(t) - \be(\bbeta^{\star}, t)\}^{\otimes 2} \bphi_{2} d\Lambda_0(t)
\end{align*}
is an average of i.i.d.~random variables which are bounded by $4\exp(M) M_{1} M_{2}\Lambda_0(\tau)$. Hence, by Hoeffding's inequality,  
\begin{align*}
    \bbP\left\{|\bphi_{1}^{\top} \cH_{n}(\bbeta^{\star}) \bphi_{2} - \bphi_{1}^{\top} \cH^{\star} \bphi_{2}| > 4\exp(M) M_{1} M_{2} \Lambda_{0}(\tau)u\right\} \leq 2\exp\left(-\frac{nu^{2}}{2}\right).  
\end{align*}
We now bound $|\bphi_{1}^{\top} \cW_{2}(\bbeta^{\star}) \bphi_{2}|$. Under $\cE_{0}$, we have 
\begin{align*}
    \bphi_{1}^{\top} \cW_{2}(\bbeta^{\star}) \bphi_{2} &= \int_{0}^{\tau} \bphi_{1}^{\top} \left\{S^{(1)}(\bbeta^{\star}, t) - S^{(0)}(\bbeta^{\star}, t)\be(\bbeta^{\star}, t)\right\}^{\otimes 2} \bphi_{2} \frac{1}{S^{(0)}(\bbeta^{\star}, t)} d\Lambda_0(t)\cr
    &\leq \frac{1}{\chi_{0}} \int_{0}^{\tau} \bphi_{1}^{\top} \left\{S^{(1)}(\bbeta^{\star}, t) - S^{(0)}(\bbeta^{\star}, t)\be(\bbeta^{\star}, t)\right\}^{\otimes 2} \bphi_{2} d\Lambda_0(t)\cr
    &\leq \frac{1}{\chi_{0}}\sqrt{\int_{0}^{\tau} \bphi_{1}^{\top} \left\{S^{(1)}(\bbeta^{\star}, t) - S^{(0)}(\bbeta^{\star}, t)\be(\bbeta^{\star}, t)\right\}^{\otimes 2} \bphi_{1} d\Lambda_0(t)}\cr
    &\times \sqrt{\int_{0}^{\tau} \bphi_{2}^{\top} \left\{S^{(1)}(\bbeta^{\star}, t) - S^{(0)}(\bbeta^{\star}, t)\be(\bbeta^{\star}, t)\right\}^{\otimes 2} \bphi_{2} d\Lambda_0(t)}\cr
    &=: \frac{1}{\chi_{0}}\sqrt{\Phi_{1} \Phi_{2}}. 
\end{align*}
Observe that
\begin{align*}
    \bphi_{1}^{\top}\left\{S^{(1)}(\bbeta^{\star}, t) - S^{(0)}(\bbeta^{\star}, t)\be(\bbeta^{\star}, t)\right\} = \frac{1}{n}\sum_{i = 1}^{n} Y_i(t) \exp\{\bx_{i}(t)^\top \bbeta^{\star}\}\bphi_{1}^{\top} \{\bx_{i}(t) - \be(\bbeta^{\star}, t)\}
\end{align*}
is an average of i.i.d.~zero-mean random variables with 
\begin{align*}
    |Y_i(t) \exp\{\bx_{i}(t)^\top \bbeta^{\star}\}\bphi_{1}^{\top} \{\bx_{i}(t) - \be(\bbeta^{\star}, t)\}| \leq 2\exp(M)M_{1}. 
\end{align*}
Hence, by Lemma 4.2 in~\citet{Huang2013oracle} and~\eqref{eq_bound_M}, for any $v > 0$, 
\begin{align*}
    \bbP_{\Phi, 1} := \bbP\left(\Phi_{1} > 4\exp(2M)M_{1}^{2}\Lambda_{0}(\tau)v\right) \leq 2.221 \exp\left(-\frac{nv}{2 + 2\sqrt{v}/3}\right). 
\end{align*}
Similarly, we have 
\begin{align*}
    \bbP_{\Phi, 2} := \bbP\left(\Phi_{2} > 4\exp(2M)M_{2}^{2}\Lambda_{0}(\tau)v\right) \leq 2.221 \exp\left(-\frac{nv}{2 + 2\sqrt{v}/3}\right). 
\end{align*}
Consequently, we obtain 
\begin{align*}
    \bbP&\left(|\bphi_{1}^{\top} \cW_{2}(\bbeta^{\star})\bphi_{2}| > \frac{4\exp(2M) M_{1}M_{2}\Lambda_{0}(\tau) v}{\chi_{0}}\right)\cr
    &\leq \bbP_{\Phi, 1} + \bbP_{\Phi, 2} \leq 4.442\exp\left(-\frac{nv}{2 + 2\sqrt{v}/3}\right). 
\end{align*}
Putting all these pieces together, we obtain~\eqref{eq_Hessian_max}. 
\end{proof}

\section{Proofs of results in Section~\ref{sec_iterative}}

\subsection{Proof of Lemma~\ref{Theorem_consistency}}
\begin{proof}[Proof of Lemma~\ref{Theorem_consistency}]
For simplicity of notation, we omit the subscript of $\bdelta_{t + 1}$ in this proof. By Lemma 3.1 in \citet{Huang2013oracle} and~\eqref{eq_deviation_t}, we have 
\begin{align}
\label{eq_D1}
    \max\left\{D_{1}(\bbeta^{\star} + \bdelta, \bbeta^{\star}), \frac{\vartheta_{t + 1}}{2} \|\bdelta_{\cS_{\star}^{c}}\|_{1}\right\} \leq D_{1}(\bbeta^{\star} + \bdelta, \bbeta^{\star}) + \frac{\vartheta_{t + 1}}{2} \|\bdelta_{\cS_{\star}^{c}}\|_{1} \leq \frac{3\vartheta_{t + 1}}{2} \|\bdelta_{\cS_{\star}}\|_{1}.
\end{align}
Therefore $\bdelta \in \cC(\cS_{\star}, 3)$ and $D_{1}(\bbeta^{\star} + \bdelta, \bbeta^{\star}) \leq 3\vartheta_{t + 1}\|\bdelta_{\cS_{\star}}\|_{1}/2 \leq 3\vartheta_{t + 1}\sqrt{|\cS_{\star}|}\|\bdelta\|_{2}/2$. As $\cL_{1}(\bbeta)$ is convex, it follows that $D_{1}(\bbeta^{\star} + \bdelta, \bbeta^{\star}) \geq \theta^{-1} D_{1}(\bbeta_{\theta}, \bbeta^{\star})$ for any $\theta \in (0, 1]$, where $\bbeta_{\theta} = \bbeta^{\star} + \theta \bdelta$. Consequently, by Lemma~\ref{Lemma_3.2_Huang}, Assumptions~\ref{Assumption_covariate} and \ref{Assumption_eigenvalue}, 
\begin{align*}
    \frac{3\vartheta_{t + 1}\sqrt{|\cS_{\star}|}}{2} \|\bdelta\|_{2} &\geq \theta \exp\left(-\theta \max_{t\in [0, \tau]} \max_{1\leq i, j\leq m} |\bdelta^\top \bx_{i}(t) - \bdelta^\top \bx_{j}(t)|\right) \bdelta^{\top} \nabla^{2} \cL_{1} (\bbeta^{\star}) \bdelta\cr
    &\geq \theta\exp\left(-8\theta B\sqrt{|\cS_{\star}|}\|\bdelta\|_{2}\right) \bdelta^{\top} \nabla^{2} \cL_{1}(\bbeta^{\star}) \bdelta\cr
    &\geq \theta\exp\left(-8\theta B\sqrt{|\cS_{\star}|}\|\bdelta\|_{2}\right) \varrho_{\star}^{2} \|\bdelta\|_{2}^{2}. 
\end{align*}
Therefore, for any $\theta \in (0, 1]$, we have
\begin{align*}
    8\theta B \sqrt{|\cS_{\star}|}\|\bdelta\|_2 \exp\left(- 8\theta B \sqrt{|\cS_{\star}|}\|\bdelta\|_2\right) \leq \frac{12\vartheta_{t + 1} |\cS_{\star}| B}{\varrho_{\star}^{2}} = \chi_{t + 1}. 
\end{align*}
As $\chi_{t + 1} < e^{-1}$, it follows that $8B\sqrt{|\cS_{\star}|}\|\bdelta\|_{2} \leq \varpi_{t + 1}$, where $\varpi_{t + 1}$ is the smaller solution of equation $\varpi_{t + 1}\exp(-\varpi_{t + 1}) = \chi_{t + 1}$. Consequently, we obtain 
\begin{align*}
    \|\bdelta\|_2 \leq \frac{\varpi_{t + 1}}{8B\sqrt{|\cS_{\star}|}} = \frac{\exp(\varpi_{t + 1}) \chi_{t + 1}}{8B\sqrt{|\cS_{\star}|}} = \frac{3\exp(\varpi_{t + 1})\vartheta_{t + 1} \sqrt{|\cS_{\star}|}}{2\varrho_{\star}^{2}} \leq \frac{3 \vartheta_{t + 1} e \sqrt{|\cS_{\star}|}}{2\varrho_{\star}^{2}}
\end{align*}
and
\begin{align*}
    D_{1}(\bbeta^{\star} + \bdelta, \bbeta^{\star}) \leq \frac{3\vartheta_{t + 1}\sqrt{|\cS_{\star}|}}{2} \|\bdelta\|_{2} \leq \frac{9 \vartheta_{t + 1}^{2} e |\cS_{\star}|}{4 \varrho_{\star}^{2}}. 
\end{align*}
\end{proof}

\subsection{Proof of Theorem~\ref{Theorem_iteration}}
\begin{proof}[Proof of Theorem \ref{Theorem_iteration}]
For each $t \geq 0$, by \eqref{eq_cond_lambda} and the proof of Lemma~\ref{Theorem_consistency},
\begin{align*}
    \|\bdelta_{t + 1}\|_{2} \leq \frac{3 \eta_0 e \sqrt{|\cS_{\star}|}}{2\varrho_{\star}^{2}} \Upsilon_{t} \enspace \mathrm{and} \enspace D_{1}(\bbeta_{t + 1}, \bbeta^{\star}) \leq \frac{9\eta_{0}^{2} e |\cS_{\star}|}{4\varrho_{\star}^{2}} \Upsilon_{t}^{2}, 
\end{align*}
where $\Upsilon_{t} = \|\nabla\cL_{1}(\bbeta^{\star}) - \nabla \cL_{1}(\bbeta_{t}) + \hat{\nabla \cL}(\bbeta_{t})\|_{\infty}$. By Lemma~\ref{Lemma_HB}, 
with probability at least $1 - 2K\exp(-m\rho_{0}^{2}/2) - 4/(pK^{2})$, we have
\begin{align}
\label{eq_Upsilon_t}
    \Upsilon_{t} &\leq \|\hat{\nabla^{2} \cL} (\bbeta^{\star}) - \nabla^{2} \cL_{1}(\bbeta^{\star})\|_{\max} \|\bdelta_{t}\|_{1} + \|\hat{\nabla \cL}(\bbeta^{\star})\|_{\infty} + \cA_{2} \exp(2B\bar{\omega}) \|\bdelta_{t}\|_{2}^{2},
\end{align}
where $\hat{\nabla^{2} \cL}(\bbeta^{\star}) = K^{-1} \sum_{k = 1}^{K} \nabla^{2} \cL_{k}(\bbeta^{\star})$. By Lemma~\ref{Lemma_Hessian}, with probability at least $1 - 2K\exp(-m\rho_{0}^{2}/2) - 8.442/(pK^{2})$, we have
\begin{align*}
    \|\hat{\nabla^{2} \cL}(\bbeta^{\star}) - \nabla^{2} \cL_{1}(\bbeta^{\star})\|_{\max} \leq 2\max_{k \in [K]}\|\nabla^{2} \cL_{k}(\bbeta^{\star}) - \cH^{\star}\|_{\max} \leq \frac{\cA_{1}}{4}\sqrt{\frac{\log (pK)}{m}}.
\end{align*}
We now bound $\|\hat{\nabla \cL}(\bbeta^{\star})\|_{\infty}$. We have
\begin{align}\label{eq_distributed_gradient}
    \hat{\nabla\cL}(\bbeta^{\star}) = -\frac{1}{n}\sum_{k = 1}^{K} \sum_{i \in \cI_{k}} \int_{0}^{\tau} \{\bx_{i}(t) - \cX_{k}(\bbeta^{\star}, t)\} d M_{i}(t).
\end{align}
By Assumption~\ref{Assumption_covariate} and Lemma~\ref{Lemma_3.3_Huang},
\begin{align}\label{eq_nabla_ell}
    \bbP\left(\|\hat{\nabla \cL}(\bbeta^{\star})\|_{\infty} > 4 B\sqrt{(\log p)/n}\right) \leq \frac{2}{p}.
\end{align}
Putting all these pieces together, we obtain
\begin{align}
\label{eq_bbeta_bound}
    \|\bdelta_{t + 1}\|_{2} \leq \frac{3\eta_0 e \sqrt{|\cS_{\star}|}}{2\varrho_{\star}^{2}} \left(\cA_{1} \sqrt{\frac{|\cS_{\star}|\log (pK)}{m}} \|\bdelta_{t}\|_{2} + 4B\sqrt{\frac{\log p}{n}} + \cA_{2}\exp(2B\bar{\omega}) \|\bdelta_{t}\|_{2}^{2}\right),
\end{align}
Consequently, we have $\|\bdelta_{t + 1}\|_{2} \leq \cA_{0} \|\bdelta_{t}\|_{2} + \cC_{n, p}$ for each $t \geq 0$. Consequently, as $\cA_{0} < 1$, we obtain  
\begin{align*}
    \|\bdelta_{t + 1}\|_{2} &\leq \cA_{0}^{t + 1} \|\bdelta_{0}\|_{2} + \cC_{n, p} \sum_{s \leq t} \cA_{0}^{s} \leq \cA_{0}^{t + 1} \|\bdelta_{0}\|_{2} + \frac{\cC_{n, p}}{1 - \cA_{0}} \enspace \mathrm{and}\cr
    D_{1}(\bbeta_{t + 1}, \bbeta^{\star}) &\leq \frac{1}{e} (\cA_{0} \|\bdelta_{t}\|_{2} + \cC_{n, p})^{2} \leq \frac{1}{e} \left(\cA_{0}^{t + 1} \|\bdelta_{0}\|_{2} + \frac{\cC_{n, p}}{1 - \cA_{0}}\right)^{2}. 
\end{align*}
\end{proof}

\section{Proofs of results in Section~\ref{sec_inference}}
\subsection{Proof of Lemma~\ref{Lemma_bv_consistency}}
We first prove a more general result, which relates bounds on the estimation of $\bomega^{\star}$ to bounds on the estimation of $\bbeta^{\star}$.

\begin{lemma}
\label{Lemma_bw_LASSO}
Let Assumptions~\ref{Assumption_covariate}--\ref{Assumption_cM_bomega} hold. Let $\hat{\bomega} \in \bbR^{p}$ be the solution of  
\begin{align*}
    \hat{\bomega} = \underset{\bomega\in\bbR^{p}}{\argmin} \left\{\bomega^{\top} \nabla^{2} \cL (\hbbeta) \bomega - 2\bc^{\top} \bomega + \vartheta^{\diamond} \|\bomega\|_{1}\right\}, 
\end{align*}
where $\vartheta^{\diamond} > 0$ is a tuning parameter and $\hbbeta\in\bbR^{p}$ is an estimator for $\bbeta^{\star}$ which satisfies that 
\begin{align*}
    \|\hbbeta - \bbeta^{\star}\|_{1} \leq \chi_{n}^{\circ} \enspace \mathrm{and} \enspace (\hbbeta - \bbeta^{\star})^{\top} \nabla^{2} \cL(\bbeta^{\star}) (\hbbeta - \bbeta^{\star}) \leq \chi_{n}^{\diamond}
\end{align*}
for some $\chi_{n}^{\circ} > 0$ and $\chi_{n}^{\diamond} > 0$, respectively. Assume that there exists a positive constant $\varrho_{\diamond}$ such that 
\begin{align}
\label{Assumption_eigenvalue_diamond}
    \min_{0 \neq \bv \in \cC(\cS_{\diamond}, 4)} \frac{\bv^{\top} \nabla^{2} \cL(\bbeta^{\star}) \bv}{\|\bv\|_{2}^{2}} \geq \varrho_{\diamond}^{2}. 
\end{align}
Take $\vartheta^{\diamond} = \cI_{0}\cA_{3}\sqrt{(\log p)/n}$ for some constant $\cI_{0} \geq 4$, where $\cA_{3}$ is defined in~\eqref{eq_A1_A2_definition}. Then, for $\bvarphi = \hat{\bomega} - \bomega^{\star}$, with probability at least $1 - 2\exp(-n\rho_{0}^{2}/2) - 8.442 p^{-1}$, we have  
\begin{align}
\label{eq_bound_Q}
    \|\bvarphi\|_{1} &\leq \frac{20\cT_{n}^{2}}{\vartheta^{\diamond}} + \frac{45\vartheta^{\diamond}|\cS_{\diamond}|\exp(4B\chi_{n}^{\circ})}{\varrho_{\diamond}^{2}} \cr
    \bvarphi^{\top} \nabla^{2} \cL(\hbbeta) \bvarphi &\leq 2\cT_{n}^{2} + \frac{9(\vartheta^{\diamond})^{2}|\cS_{\diamond}|\exp(4B\chi_{n}^{\circ})}{2\varrho_{\diamond}^{2}}, 
\end{align}
where $\cT_{n} = 4 \cM_{\diamond} \exp(4B\chi_{n}^{\circ}) \sqrt{\chi_{n}^{\diamond}}$.
\end{lemma}

\begin{proof}[Proof of Lemma~\ref{Lemma_bw_LASSO}]
By the definition of $\hat{\bomega}$, we have   
\begin{align*}
    \hat{\bomega}^\top \nabla^{2} \cL(\hbbeta) \hat{\bomega} - 2 \bc^{\top} \hat{\bomega} + \vartheta^{\diamond} \|\hat{\bomega}\|_{1} \leq \bomega^{\star \top} \nabla^{2} \cL(\hbbeta) \bomega^{\star} - 2 \bc^{\top} \bomega^{\star} + \vartheta^{\diamond} \|\bomega^{\star}\|_{1}.
\end{align*}
Combined with the fact that $\|\bomega^{\star}\|_{1} = \|\bomega_{\cS_{\diamond}}^{\star}\|_{1} \leq \|\hat{\bomega}_{\cS_{\diamond}}\|_{1} + \|\bvarphi_{\cS_{\diamond}}\|_{1}$, we obtain 
\begin{align*}
    Q(\bvarphi) := \bvarphi^{\top} \nabla^{2} \cL(\hbbeta) \bvarphi &\leq 2 \bvarphi^{\top} \bc - 2\bvarphi^{\top} \nabla^{2} \cL(\hbbeta) \bomega^{\star} + \vartheta^{\diamond} \|\bomega^{\star}\|_{1} - \vartheta^{\diamond} \|\hat{\bomega}\|_{1}\cr
    &\leq 2\bvarphi^{\top} \left\{\bc - \nabla^{2}\cL(\hbbeta)\bomega^{\star}\right\} + \vartheta^{\diamond} \|\bvarphi_{\cS_{\diamond}}\|_{1} - \vartheta^{\diamond} \|\bvarphi_{\cS_{\diamond}^{c}}\|_{1}. 
\end{align*}
Decompose
\begin{align}
\label{eq_deviation}
    \bvarphi^{\top}\left\{\bc - \nabla^{2}\cL(\hbbeta)\bomega^{\star}\right\} &= \bvarphi^{\top}\left\{\bc - \nabla^{2}\cL(\bbeta^{\star})\bomega^{\star}\right\} + \bvarphi^{\top}\left\{\nabla^{2}\cL(\bbeta^{\star}) - \nabla^{2}\cL(\hbbeta)\right\}\bomega^{\star} \cr
    &=: \cR_{1} + \cR_{2}. 
\end{align}
Recall that $\cH^{\star}\bomega^{\star} = \bc$. Hence, by Lemma~\ref{Lemma_Hessian}, with probability at least $1 - 2\exp(-n\rho_{0}^{2}/2) - 8.442 p^{-1}$, we have 
\begin{align}
\label{eq_calK_1}
    |\cR_{1}| \leq \|\{\nabla^{2}\cL(\bbeta^{\star}) - \cH^{\star}\}\bomega^{\star}\|_{\infty} \|\bvarphi\|_{1} \leq \cA_{3}\sqrt{\frac{\log p}{n}} \|\bvarphi\|_{1} \leq \frac{\vartheta^{\diamond}}{4} \|\bvarphi\|_{1}, 
\end{align}
which further implies that 
\begin{align*}
    Q(\bvarphi) &\leq 2|\cR_{2}| + \frac{\vartheta^{\diamond}}{2} \|\bvarphi\|_{1} + \vartheta^{\diamond} \|\bvarphi_{\cS_{\diamond}}\|_{1} - \vartheta^{\diamond} \|\bvarphi_{\cS_{\diamond}^{c}}\|_{1} \cr
    &= 2|\cR_{2}| + \frac{3\vartheta^{\diamond}}{2} \|\bvarphi_{\cS_{\diamond}}\|_{1} - \frac{\vartheta^{\diamond}}{2} \|\bvarphi_{\cS_{\diamond}^{c}}\|_{1}. 
\end{align*}
Let $\bdelta = \hbbeta - \bbeta^{\star} \in \bbR^{p}$. By a similar argument as~\eqref{eq_xi_j},  
\begin{align*}
    \cR_{2} = \frac{1}{n}\int_{0}^{\tau} \sum_{i = 1}^{n} \omega_{i}(\bbeta^{\circ}, t) \bvarphi^{\top} \tilde{\bx}_{i}(\bbeta^{\circ}, t) \bomega^{\star\top} \tilde{\bx}_{i}(\bbeta^{\circ}, t) \tilde{\bx}_{i}(\bbeta^{\circ}, t)^{\top}\bdelta d\bar{N}(t),
\end{align*}
where $\bbeta^{\circ} = \bbeta^{\star} + \varpi (\hbbeta - \bbeta^{\star})$ for some $\varpi \in [0, 1]$. By the Cauchy-Schwarz inequality and Assumption~\ref{Assumption_cM_bomega}, we have  
\begin{align}\label{eq_bound_cR2}
    |\cR_{2}| &\leq \frac{1}{n} \int_{0}^{\tau} \sqrt{\sum_{i = 1}^{n} \omega_{i}(\bbeta^{\circ}, t)\left\{\bvarphi^{\top} \tilde{\bx}_{i}(\bbeta^{\circ}, t)\right\}^{2}} \sqrt{\sum_{i = 1}^{n} \omega_{i}(\bbeta^{\circ}, t)\left\{\bomega^{\star\top} \tilde{\bx}_{i}(\bbeta^{\circ}, t) \tilde{\bx}_{i}(\bbeta^{\circ}, t)^\top \bdelta\right\}^{2}} d \bar{N}(t)\cr
    &\leq \frac{2\cM_{\diamond}}{n} \int_{0}^{\tau} \sqrt{\sum_{i = 1}^{n} \omega_{i}(\bbeta^{\circ}, t)\left\{\bvarphi^{\top} \tilde{\bx}_{i}(\bbeta^{\circ}, t)\right\}^{2}} \sqrt{\sum_{i = 1}^{n} \omega_{i}(\bbeta^{\circ}, t)\left\{\tilde{\bx}_{i}(\bbeta^{\circ}, t)^\top \bdelta\right\}^{2}} d \bar{N}(t)\cr
    &\leq \frac{2\cM_{\diamond}}{n} \sqrt{\int_{0}^{\tau}\sum_{i = 1}^{n} \omega_{i}(\bbeta^{\circ}, t)\left\{\bvarphi^{\top} \tilde{\bx}_{i}(\bbeta^{\circ}, t)\right\}^{2} d\bar{N}(t)} \sqrt{\int_{0}^{\tau} \sum_{i = 1}^{n} \omega_{i}(\bbeta^{\circ}, t)\left\{\tilde{\bx}_{i}(\bbeta^{\circ}, t)^\top \bdelta\right\}^{2} d \bar{N}(t)}\cr
    &= 2 \cM_{\diamond} \sqrt{\bvarphi^{\top} \nabla^2 \cL(\bbeta^{\circ}) \bvarphi} \sqrt{\bdelta^{\top} \nabla^{2} \cL(\bbeta^{\circ}) \bdelta},
\end{align}
where the last equality follows from the definition that for any $\bbeta \in \bbR^{p}$, 
\begin{align*}
    \nabla^{2}\cL (\bbeta) = \frac{1}{n} \int_{0}^{\tau}\sum_{i = 1}^{n} \omega_{i}(\bbeta, t) \tilde{\bx}_{i}(\bbeta, t)^{\otimes 2} d\bar{N}(t).
\end{align*}
Note that $\|\hbbeta - \bbeta^{\circ}\|_{1} \leq \|\bdelta\|_{1} \leq \chi_{n}^{\circ}$. Hence, it follows from Lemma~\ref{Lemma_3.2_Huang} and Assumption~\ref{Assumption_covariate} that
\begin{align*}
    \bvarphi^{\top} \nabla^{2} \cL (\bbeta^{\circ}) \bvarphi &\leq \exp\left(2\sup_{t\in[0, \tau]}\max_{i, i' \in [n]} |(\hbbeta - \bbeta^{\circ})^{\top} \{\bx_{i}(t) - \bx_{i'}(t)\}|\right) \bvarphi^\top \nabla^{2} \cL (\hbbeta) \bvarphi\cr
    &\leq \exp\left(4\|\hbbeta - \bbeta^{\circ}\|_{1} \sup_{t\in [0, \tau]} \max_{i \in [n]} \|\bx_{i}(t)\|_{\infty}\right) \bvarphi^{\top} \nabla^{2} \cL (\hbbeta) \bvarphi\cr 
    &\leq \exp(4B\chi_{n}^{\circ}) Q(\bvarphi).
\end{align*}
Similarly, we have 
\begin{align*}
    \bdelta^{\top} \nabla^{2} \cL(\bbeta^{\circ}) \bdelta \leq \exp(4B\chi_{n}^{\circ}) \bdelta^{\top} \nabla^{2} \cL(\bbeta^{\star}) \bdelta \leq \exp(4B\chi_{n}^{\circ})\chi_{n}^{\diamond}. 
\end{align*}
Consequently, we obtain
\begin{align}
\label{eq_Q_bphi}
    Q(\bvarphi) \leq \cT_{n} \sqrt{Q(\bvarphi)} + \frac{3\vartheta^{\diamond}}{2} \|\bvarphi_{\cS_{\diamond}}\|_{1} - \frac{\vartheta^{\diamond}}{2} \|\bvarphi_{\cS_{\diamond}^{c}}\|_{1}.
\end{align}
It suffices to consider the case where $Q(\bvarphi) > \cT_{n}^{2}$. In view of~\eqref{eq_Q_bphi}, we have $\bvarphi \in \cC(\cS_{\diamond}, 3)$.   
By~\eqref{Assumption_eigenvalue_diamond} and Lemma~\ref{Lemma_3.2_Huang}, it follows that   
\begin{align*}
    \|\bvarphi_{\cS_{\diamond}}\|_{1} \leq \sqrt{|\cS_{\diamond}|} \|\bvarphi\|_{2} \leq \frac{\sqrt{|\cS_{\diamond}|}\exp(2B\chi_{n}^{\circ})}{\varrho_{\diamond}} \sqrt{Q(\bvarphi)}.
\end{align*}
Substituting this into \eqref{eq_Q_bphi}, we obtain \eqref{eq_bound_Q} in view of 
\begin{align*}
    Q(\bvarphi) \leq \cT_{n} \sqrt{Q(\bvarphi)} + \frac{3\vartheta^{\diamond}\sqrt{|\cS_{\diamond}|}\exp(2B\chi_{n}^{\circ})}{2\varrho_{\diamond}} \sqrt{Q(\bvarphi)}. 
\end{align*}
We now upper bound $\|\bvarphi\|_{1}$. If $\|\bvarphi_{\cS_{\diamond}^{c}}\|_{1} \leq 4 \|\bvarphi_{\cS_{\diamond}}\|_{1}$, it follows from Assumption~\ref{Assumption_eigenvalue_diamond} that 
\begin{align*}
    \|\bvarphi\|_{1} \leq 5\|\bvarphi_{\cS_{\diamond}}\|_{1} \leq \frac{5\sqrt{|\cS_{\diamond}|}\exp(2B\chi_{n}^{\circ})}{\varrho_{\diamond}} \sqrt{Q(\bvarphi)}. 
\end{align*}
In the other case where $\|\bvarphi_{\cS_{\diamond}^{c}}\|_{1} > 4 \|\bvarphi_{\cS_{\diamond}}\|_{1}$, since $Q(\bvarphi) \geq 0$, it follows from \eqref{eq_Q_bphi} that 
\begin{align*}
    \vartheta^{\diamond} \|\bvarphi_{\cS_{\diamond}}\|_{1} \leq \vartheta^{\diamond} \|\bvarphi_{\cS_{\diamond}^{c}}\|_{1} - 3 \vartheta^{\diamond} \|\bvarphi_{\cS_{\diamond}}\|_{1} \leq 2 \cT_{n} \sqrt{Q(\bvarphi)}. 
\end{align*}
Consequently, we obtain 
\begin{align*}
    \|\bvarphi\|_{1} = \|\bvarphi_{\cS_{\diamond}}\|_{1} + \|\bvarphi_{\cS_{\diamond}^{c}}\|_{1} \leq 4 \|\bvarphi_{\cS_{\diamond}}\|_{1} + \frac{2\cT_{n}\sqrt{Q(\bvarphi)}}{\vartheta^{\diamond}} \leq \frac{10\cT_{n}\sqrt{Q(\bvarphi)}}{\vartheta^{\diamond}}. 
\end{align*}
\end{proof}

\begin{proof}[Proof of Lemma~\ref{Lemma_bv_consistency}]
Let $\bdelta = \hbbeta - \bbeta^{\star}$. By the definition of $\hbbeta$ and Corollary~\ref{Corollary_beta0_LASSO},   
\begin{align*}
    \|\bdelta\|_{1} = O_{\bbP}\left(|\cS_{\star}|\sqrt{\frac{\log p}{n}}\right) \enspace \mathrm{and} \enspace \bdelta^{\top} \nabla^{2} \cL_{1}(\bbeta^{\star}) \bdelta = O_{\bbP}\left(\frac{|\cS_{\star}| \log p}{n}\right). 
\end{align*}
Then it follows from Lemma~\ref{Lemma_Hessian} that
\begin{align*}
    \max_{k \in [K]} \bdelta^{\top} \nabla^{2} \cL_{k}(\bbeta^{\star}) \bdelta &\leq \bdelta^{\top} \nabla^{2} \cL_{1}(\bbeta^{\star}) \bdelta + \max_{2\leq k\leq K} \|\nabla^{2} \cL_{k}(\bbeta^{\star}) - \nabla^{2} \cL_{1}(\bbeta^{\star})\|_{\max} \|\bdelta\|_{1}^{2}\cr
    &= O_{\bbP}\left(\frac{|\cS_{\star}|\log p}{n} + \sqrt{\frac{\log p}{m}} |\cS_{\star}|^{2} \frac{\log p}{n}\right) = O_{\bbP}\left(\frac{|\cS_{\star}|\log p}{n}\right). 
\end{align*}
Consequently, we obtain~\eqref{eq_bound_bomegak} by applying Lemma~\ref{Lemma_bw_LASSO}.   
\end{proof}

\subsection{Proof of Theorem~\ref{Theorem_distributed_CLT}}

We first prove the following central limit theorem on $\bphi^{\top} \hat{\nabla \cL}(\bbeta^{\star})$ for any vector $\bphi$.

\begin{lemma}
\label{Lemma_gradient_clt}
Let Assumptions~\ref{Assumption_beta_1} and~\ref{Assumption_at_risk} hold. Assume that 
\begin{align}
\label{eq_Mv}
    \max_{i \in [n]} \sup_{t \in [0, \tau]} |\bx(t)^{\top} \bphi| \leq M(\bphi) \enspace \mathrm{and} \enspace \frac{M(\bphi)}{m\sigma(\bphi)} \to 0,  
\end{align}
where $\sigma(\bphi)^{2} = \bphi^{\top} \cH^{\star} \bphi$. Then we have 
\begin{align*}
    \sup_{z\in\bbR} \left|\bbP\left\{\frac{\sqrt{n}\bphi^{\top} \hat{\nabla \cL}(\bbeta^{\star})}{\sigma(\bphi)} \leq z\right\} - \Phi(z)\right| \to 0.
\end{align*}
\end{lemma}

\begin{proof}[Proof of Lemma~\ref{Lemma_gradient_clt}]
In view of~\eqref{eq_distributed_gradient}, we decompose 
\begin{align*}
    \bphi^{\top} \hat{\nabla \cL}(\bbeta^{\star}) &= - \frac{1}{n}\sum_{i = 1}^{n} \int_{0}^{\tau} \bphi^{\top}\{\bx_{i}(t) - \be(\bbeta^{\star}, t)\} d M_{i}(t)\cr
    &- \frac{1}{n} \sum_{k = 1}^{K} \int_{0}^{\tau} \bphi^{\top}\{\be(\bbeta^{\star}, t) - \cX_{k}(\bbeta^{\star}, t)\} d \bar{M}_{k}(t) \cr
    &=: \frac{1}{n}\sum_{i = 1}^{n} \xi_{i}(\bphi) + \Delta_{n}(\bphi),
\end{align*}
where $\xi_{i}(\bphi) = -\int_{0}^{\tau} \bphi^{\top}\{\bx_{i}(t) - \be(\bbeta^{\star}, t)\} d M_{i}(t)$ and 
\begin{align*}
    \cX_{k}(\bbeta^{\star}, t) = \frac{\sum_{\ell \in \cI_{k}}Y_{\ell}(t)\exp\{\bx_{\ell}(t)^{\top}\bbeta^{\star}\}\bx_{\ell}(t)}{\sum_{\ell \in \cI_{k}}Y_{\ell}(t) \exp\{\bx_{\ell}(t)^{\top} \bbeta^{\star}\}}. 
\end{align*}
Note that $\bbE\{\Delta_{n}(\bphi)\} = 0$ and 
\begin{align*}
    \Var\{\Delta_{n}(\bphi)\} &= \frac{1}{n} \bbE \int_{0}^{\tau} \left[\bphi^{\top} \{\be(\bbeta^{\star}, t) - \cX_{1}(\bbeta^{\star}, t)\}\right]^{2} S_{1}^{(0)}(\bbeta^{\star}, t) d \Lambda_{0}(t)\cr
    &=: \frac{1}{n} \bbE(W_{1}),
\end{align*}
where $S_{1}^{(0)}(\bbeta^{\star}, t) = m^{-1} \sum_{\ell \in \cI_{1}}Y_{\ell}(t) \exp\{\bx_{\ell}(t)^{\top} \bbeta^{\star}\}$.
Similar to $\cE$ in \eqref{eq_event_E}, we define 
\begin{align*}
    \cE_{1} = \left\{\frac{1}{m} \sum_{i \in \mathcal{I}_{1}} Y_{i}(\tau)  > \frac{\rho_{0}}{2}\right\},
\end{align*}
which satisfies that $\bbP(\cE_{1}) \geq 1 - \exp(-m\rho_{0}^{2}/2)$. 
Under $\cE_{1}$, we have 
\begin{align*}
    W_{1} &= \int_{0}^{\tau} \frac{1}{S_{1}^{(0)}(\bbeta^{\star}, t)} \left\{\bphi^{\top} S_{1}^{(1)}(\bbeta^{\star}, t) - S_{1}^{(0)}(\bbeta^{\star}, t) \bphi^{\top} \be(\bbeta^{\star}, t)\right\}^{2} d \Lambda_{0}(t)\cr
    &\leq \frac{2\exp(M)}{\rho_{0}} \int_{0}^{\tau} \left\{\bphi^{\top} S_{1}^{(1)}(\bbeta^{\star}, t) - S_{1}^{(0)}(\bbeta^{\star}, t) \bphi^{\top} \be(\bbeta^{\star}, t)\right\}^{2} d \Lambda_{0}(t). 
\end{align*}
Then, by Lemma 4.2 in \citet{Huang2013oracle} and \eqref{eq_Mv}, 
\begin{align*}
    \bbP(\{W_{1} > \bar{M}(\bphi) \nu\} \cap \cE_{1}) \leq 2.221\exp\left(-\frac{m \nu}{2 + 2\sqrt{\nu}/3}\right),
\end{align*}
where $\bar{M}(\bphi) = 8\exp(3M)M(\bphi)^{2} \Lambda_{0}(\tau)/\rho_{0}$.
Elementary calculations imply that  
\begin{align*}
    \bbE(W_{1} \mathbb{I}\{\cE_{1}\}) &= \int_{0}^{\infty} \bbP(\{W_{1} > \nu\} \cap \cE_{1}) d \nu\cr
    &\leq 2.221 \bar{M}(\bphi) \left(\frac{4}{m} + \frac{32}{9m^{2}}\right). 
\end{align*}
Combined with the fact that $W_{1} \leq 4 M(\bphi)^{2} \exp(M) \Lambda_{0}(\tau)$, we obtain 
\begin{align*}
    \bbE(W_{1}) &\leq \bbE(W_{1} \mathbb{I}\{\cE_{1}\}) + 4 M(\bphi)^{2} \exp(M) \Lambda_{0}(\tau) \bbP(\cE_{1}^{c})\cr
    &\leq 2.221 \bar{M}(\bphi) \left(\frac{4}{m} + \frac{32}{9m^{2}}\right) + 8 M(\bphi)^{2} \exp(M) \Lambda_{0}(\tau) \exp\left(-\frac{m\rho_{0}^{2}}{2}\right).  
\end{align*}
Then, by~\eqref{eq_Mv}, we have $\sqrt{n} \Delta_{n}(\bphi) \overset{\bbP}{\to} 0$. Recall~\eqref{eq_L_star} for $\cH^{\star}$. Note that $\xi_{1}(\bphi), \ldots, \xi_{n}(\bphi)$ are i.i.d.~zero mean random variables with variance 
\begin{align*}
    \Var\{\xi_{i}(\bphi)\} &= \bbE \int_{0}^{\tau} \left[\bphi^{\top}\{\bx_{1}(t) - \be(\bbeta^{\star}, t)\}\right]^{2} Y_{1}(t) \exp\{\bx_{1}(t)^\top \bbeta^{\star}\} d \Lambda_{0}(t)\cr
    &= \bphi^{\top} \cH^{\star} \bphi.
\end{align*}
Moreover, by Assumption~\ref{Assumption_covariate} and Lemma~\ref{Lemma_3.3_Huang}, it follows that $\bbE |\xi_{i}(\bphi)|^{3}\lesssim M(\bphi)^{3}$. Then, by the Lyapunov central limit theorem, we have   
\begin{align*}
    \frac{1}{\sqrt{n}} \sum_{i = 1}^{n} \frac{\xi_{i}(\bphi)}{\sigma(\bphi)} \todist \cN(0, 1), \enspace \mathrm{as} \enspace n \to \infty. 
\end{align*}
\end{proof}

\begin{proof}[Proof of Theorem~\ref{Theorem_distributed_CLT}]
By the definition of $\tilde{\bc^{\top}\bbeta^{\star}}$ and the mean value theorem, we have 
\begin{align*}
    -\frac{1}{K} \sum_{k = 1}^{K} \hat{\bomega}_{k}^{\top} \nabla \tilde{\cL}_{k} (\bbeta^{\star}) &= - \frac{1}{K} \sum_{k = 1}^{K} \hat{\bomega}_{k}^{\top} \nabla \tilde{\cL}_{k}(\hat{\bbeta}) + \frac{1}{K} \sum_{k = 1}^{K} \hat{\bomega}_{k}^{\top} \left\{\nabla \tilde{\cL}_{k}(\hat{\bbeta}) - \nabla \tilde{\cL}_{k}(\bbeta^{\star})\right\}\cr
    &= - \frac{1}{K} \sum_{k = 1}^{K} \hat{\bomega}_{k}^{\top} \nabla \tilde{\cL}_{k}(\hat{\bbeta}) + \frac{1}{K} \sum_{k = 1}^{K} \hat{\bomega}_{k}^{\top} \nabla^{2} \cL_{k}(\bbeta^{\circ}) (\hat{\bbeta} - \bbeta^{\star})\cr
    &= \bc^{\top} \hat{\bbeta} - \frac{1}{K} \sum_{k = 1}^{K} \hat{\bomega}_{k}^{\top} \nabla \tilde{\cL}_{k}(\hat{\bbeta}) - \bc^{\top}\bbeta^{\star} + \frac{1}{K} \sum_{k = 1}^{K} \left\{\nabla^{2} \cL_{k}(\bbeta^{\circ})\hat{\bomega}_{k} - \bc\right\}^{\top}\bdelta\cr
    &= \tilde{\bc^{\top} \bbeta^{\star}} - \bc^{\top} \bbeta^{\star} + \frac{1}{K} \sum_{k = 1}^{K} \left\{\nabla^{2} \cL_{k}(\bbeta^{\circ})\hat{\bomega}_{k} - \bc\right\}^{\top}\bdelta, 
\end{align*}
where $\bbeta^{\circ} = \bbeta^{\star} + \varpi (\hbbeta - \bbeta^{\star})$ for some $\varpi \in [0, 1]$. Let $\Delta_{c} = \tilde{\bc^{\top}\bbeta^{\star}} - \bc^{\top}\bbeta^{\star} + \bomega^{\star\top} \hat{\nabla \cL}(\bbeta^{\star})$. We have 
\begin{align*}
    \Delta_{c} &= -\frac{1}{K} \sum_{k = 1}^{K} \bvarphi_{k}^{\top} \nabla \tilde{\cL}_{k}(\bbeta^{\star}) + \frac{1}{K} \sum_{k = 1}^{K} \{\bc - \nabla^{2} \cL_{k}(\bbeta^{\circ})\hat{\bomega}_{k}\}^{\top} \bdelta\cr
    &=: \Delta_{c}^{\diamond} + \Delta_{c}^{\circ},
\end{align*}
where $\bvarphi_{k} = \hat{\bomega}_{k} - \bomega^{\star}$. By Lemma~\ref{Lemma_gradient_clt}, it follows that  
\begin{align*}
    \sup_{z \in \bbR} \left|\bbP\left\{\frac{\sqrt{n}\bomega^{\star\top}\hat{\nabla \cL}(\bbeta^{\star})}{\sqrt{\bc^{\top} \bomega^{\star}}} \leq z\right\} - \Phi(z)\right|\to 0.
\end{align*}
Hence it suffices to prove that $|\Delta_{c}^{\diamond} + \Delta_{c}^{\circ}| = o_{\bbP}\left(n^{-1/2}\right)$.
By Lemma~\ref{Lemma_HB} and a similar argument as that of~\eqref{eq_Upsilon_t}, with probability $1 - o(1)$, we have 
\begin{align*}
    \max_{k \in [K]} \|\nabla \tilde{\cL}_{k}(\bbeta^{\star})\|_{\infty} &\leq \max_{k \in [K]} \|\hat{\nabla^{2} \cL}(\bbeta^{\star}) - \nabla^{2} \cL_{k}(\bbeta^{\star})\|_{\max} \|\tilde{\bbeta} - \bbeta^{\star}\|_{1}^{2} + \|\hat{\nabla \cL}(\bbeta^{\star})\|\cr
    &+ \cA_{2} \exp(2B\|\tilde{\bbeta} - \bbeta^{\star}\|_{1}) \|\tilde{\bbeta} - \bbeta^{\star}\|_{2}^{2}. 
\end{align*}
Hence, it follows from Lemma~\ref{Lemma_Hessian} that 
\begin{align}
\label{eq_uniform_deviation_bound}
    \max_{k \in [K]} \|\nabla \tilde{\cL}_{k}(\bbeta^{\star})\|_{\infty} = O_{\bbP}\left(\sqrt{\frac{\log p}{m}} \frac{|\cS_{\star}|^{2}\log p}{n} + \sqrt{\frac{\log p}{n}} + \frac{|\cS_{\star}|\log p}{n}\right) = O_{\bbP}\left(\sqrt{\frac{\log p}{n}}\right).
\end{align}
By Lemma~\ref{Lemma_bv_consistency} and~\eqref{eq_uniform_deviation_bound}, 
\begin{align*}
    |\Delta_{c}^{\diamond}| &\leq \max_{k \in [K]} \|\bvarphi_{k}\|_{1} \max_{k \in [K]}\|\nabla \tilde{\cL}_{k} (\bbeta^{\star})\|_{\infty} \cr
    &= O_{\bbP}\left(
    \left(|\cS_{\star}|\sqrt{\frac{\log p}{nK}} + |\cS_{\diamond}|\sqrt{\frac{\log p}{m}}\right)
    \sqrt{\frac{\log p}{n}}\right) = o_{\bbP} (n^{-1/2}). 
\end{align*}
It remains to bound $|\Delta_{c}^{\circ}|$. Decompose 
\begin{align*}
    \Delta_{c}^{\circ} &= \left\{\bc - \hat{\nabla^{2} \cL}(\bbeta^{\star}) \bomega^{\star}\right\}^{\top}\bdelta - \frac{1}{K} \sum_{k = 1}^{K} \bvarphi_{k}^{\top} \nabla^{2} \cL_{k}(\bbeta^{\circ}) \bdelta + \bomega^{\star\top} \left\{\hat{\nabla^{2} \cL}(\bbeta^{\star}) - \hat{\nabla^{2} \cL}(\bbeta^{\circ})\right\}\bdelta\cr
    &=: \Delta_{c, 1}^{\circ} + \Delta_{c, 2}^{\circ} + \Delta_{c, 3}^{\circ}.
\end{align*}
Recall that $\bc = \cH^{\star} \bomega^{\star}$. Hence, by Lemma~\ref{Lemma_Hessian}, it follows that 
\begin{align*}
    |\Delta_{c, 1}^{\circ}| \leq \|\bc - \hat{\nabla^{2} \cL}(\bbeta^{\star}) \bomega^{\star}\|_{\infty} \|\bdelta\|_{1} = O_{\bbP}\left(\frac{|\cS_{\star}| \log p}{\sqrt{nm}}\right) = o_{\bbP}\left(n^{-1/2}\right). 
\end{align*}
By the Cauchy-Schwarz inequality and Lemma~\ref{Lemma_bv_consistency}, we have  
\begin{align*}
    |\Delta_{c, 2}^{\circ}| &\leq \max_{k \in [K]} \sqrt{\bvarphi_{k}^{\top}\nabla^{2}\cL_{k}(\bbeta^{\circ})\bvarphi_{k}}\sqrt{\bdelta^{\top}\nabla^{2} \cL_{k}(\bbeta^{\circ})\bdelta}\cr 
    &= O_{\bbP}\left(\frac{|\cS_{\star}|\log p}{n} + (\log p)\sqrt{\frac{|\cS_{\star}||\cS_{\diamond}|}{mn}}\right) = o_{\bbP}\left(n^{-1/2}\right). 
\end{align*}
Following a similar argument as that of~\eqref{eq_bound_cR2}, it is straightforward to verify that $|\Delta_{c, 3}^{\circ}| = o_{\bbP}\left(n^{-1/2}\right)$. Putting all these pieces together, we obtain~\eqref{eq_central_limit_theorem_c}.
\end{proof}

\subsection{Proof of Lemma~\ref{Lemma_consistency_sigma}}

\begin{proof}[Proof of Lemma~\ref{Lemma_consistency_sigma}]
Recall that $\cH^{\star}\bomega^{\star} = \bc$ and $\bvarphi_{k} = \hat{\bomega}_{k} - \bomega^{\star}$ for each $1\leq k\leq K$. Hence  
\begin{align*}
    \hat{\bc^{\top}\bomega^{\star}} - \bc^{\top}\bomega^{\star} &= \frac{1}{K} \sum_{k = 1}^{K} 2\left\{\bc - \nabla^{2} \cL_{k}(\hbbeta)\hat{\bomega}_{k}\right\}^{\top}\bvarphi_{k} + \frac{1}{K} \sum_{k = 1}^{K} \bvarphi_{k}^{\top} \nabla^{2} \cL_{k}(\hbbeta) \bvarphi_{k}\cr
    &+ \bomega^{\star \top}\left\{\cH^{\star} - \hat{\nabla^{2} \cL}(\bbeta^{\star})\right\} \bomega^{\star} + \bomega^{\star \top} \left\{\hat{\nabla^{2} \cL}(\bbeta^{\star}) - \hat{\nabla^{2} \cL}(\hbbeta)\right\} \bomega^{\star}\cr
    &=: \Delta_{1}^{\sigma} + \Delta_{2}^{\sigma} + \Delta_{3}^{\sigma} + \Delta_{4}^{\sigma}. 
\end{align*}
By Lemma~\ref{Lemma_bv_consistency} and the fact that $\vartheta_{k}^{\diamond} \asymp \sqrt{(\log p)/m}$ for each $1\leq k\leq K$, we have 
\begin{align*}
    |\Delta_{1}^{\sigma}| &\leq 2 \max_{k \in [K]} \|\bc - \nabla^{2} \cL_{k}(\hbbeta)\hat{\bomega}_{k}\|_{\infty} \|\bvarphi_{k}\|_{1} = O_{\bbP}\left(\frac{|\cS_{\star}|\log p}{n} + \frac{|\cS_{\diamond}|\log p}{m}\right)\cr
    |\Delta_{2}^{\sigma}| &\leq \max_{k \in [K]} \{\bvarphi_{k}^{\top} \nabla^{2} \cL_{k}(\hbbeta) \bvarphi_{k}\} = O_{\bbP}\left(\frac{|\cS_{\star}|\log p}{n} + \frac{|\cS_{\diamond}|\log p}{m}\right). 
\end{align*}
By Lemma~\ref{Lemma_Hessian} and Assumption~\ref{Assumption_cM_bomega}, it follows that $|\Delta_{3}^{\sigma}| = O_{\bbP}\left(n^{-1/2}\right)$. We now bound $|\Delta_{4}^{\sigma}|$. By Assumption~\ref{Assumption_cM_bomega} and a similar argument as~\eqref{eq_bound_cR2}, we have 
\begin{align*}
    |\Delta_{4}^{\sigma}| \leq \max_{k \in [K]} \sqrt{\bomega^{\star\top}\nabla^{2}\cL_{k}(\bbeta^{\circ})\bomega^{\star}} \sqrt{\bdelta^{\top} \nabla^{2}\cL_{k}(\bbeta^{\circ})\bdelta} = O_{\bbP}\left(\sqrt{|\cS_{\star}|(\log p)/n}\right), 
\end{align*}
where $\bbeta^{\circ} = \bbeta^{\star} + \varpi (\hbbeta - \bbeta^{\star})$ for some $\varpi \in [0, 1]$. Putting all these pieces together, we obtain~\eqref{eq_sigma_consistency_c}. 
\end{proof}

\subsection{Proof of Theorem~\ref{Theorem_test_average}}

The proofs of Lemma~\ref{Lemma_bw_consistency} and Lemma~\ref{Lemma_sigma_nu} below follow the same ideas as Lemma~\ref{Lemma_bv_consistency} and Lemma~\ref{Lemma_consistency_sigma}, hence are omitted.

\begin{lemma}
\label{Lemma_bw_consistency}
Let Assumptions~\ref{Assumption_covariate}--\ref{Assumption_at_risk} and~\ref{Assumption_cM_bw}--\ref{Assumption_eigenvalue_uniform_bw} hold. Take $\vartheta_{k}^{\sharp} = \cI_{0} \cA_{3, \sharp} \sqrt{(\log p)/m}$ for some constant $\cI_{0} \geq 4$, where $\cA_{3, \sharp} < \infty$ is a positive constant. Then we have
\begin{align}
\label{eq_bound_bwk}
    \|\hat{\bw}_{k} - \bw^{\star}\|_{1} = O_{\bbP}\left(|\cS_{\star}|\sqrt{\frac{\log p}{nK}} + |\cS_{\sharp}|\sqrt{\frac{\log p}{m}}\right). 
\end{align}
\end{lemma}

\begin{lemma}
\label{Lemma_sigma_nu}
Under the conditions of Lemma~\ref{Lemma_bw_consistency}, we have 
\begin{align}\label{eq_sigma_nu}
    \left|\frac{\hat{\sigma}_{\nu}}{\sigma_{\nu}} - 1\right| = O_{\bbP}\left(\sqrt{\frac{|\cS_{\star}|\log p}{n}} + \frac{|\cS_{\sharp}|\log p}{m}\right).
\end{align}
\end{lemma}

\begin{proof}[Proof of Theorem~\ref{Theorem_test_average}]
Recall that $\bar{\pi}(0, \hat{\bgamma}) = K^{-1} \sum_{k = 1}^{K} \tilde{\bw}_{k}^{\top} \nabla \tilde{\cL}_{k}(0, \hat{\bgamma})$, where $\tilde{\bw}_{k} = (1, -\hat{\bw}_k)^{\top}$ and $\bvarphi_{k} = \hat{\bw}_{k} - \bw^{\star}$. Hence 
\begin{align*}
    \bar{\pi}(0, \hat{\bgamma}) - \tilde{\bw}^{\star \top} \hat{\nabla \cL}(0, \bgamma^{\star}) &= - \frac{1}{K} \sum_{k = 1}^{K} \bvarphi_{k}^{\top} \{\nabla_{\bgamma} \cL_{k} (0, \hat{\bgamma}) - \nabla_{\bgamma} \cL_{k}(0, \bgamma^{\star})\}\cr
    &- \frac{1}{K} \sum_{k = 1}^{K} \bvarphi_{k}^{\top} \nabla_{\bgamma} \tilde{\cL}_{k}(0, \bgamma^{\star}) + \tilde{\bw}^{\star \top} \left\{\hat{\nabla \cL}(0, \hat{\bgamma}) - \hat{\nabla \cL}(0, \bgamma^{\star})\right\}\cr
    &=: \Delta_{\pi, 1} + \Delta_{\pi, 2} + \Delta_{\pi, 3}. 
\end{align*}
Observe that $\sigma_{\nu}^{2} = \cH_{\nu\nu}^{\star} - \cH_{\bgamma\nu}^{\star\top} \cH_{\bgamma\bgamma}^{\star-1} \cH_{\bgamma\nu}^{\star} = \tilde{\bw}^{\star\top} \cH^{\star} \tilde{\bw}^{\star}$. Hence, by~\eqref{eq_clt_cond} 
and~Lemma~\ref{Lemma_gradient_clt},
\begin{align*}
    \sup_{z \in \bbR} \left|\bbP\left\{\frac{\sqrt{n}\tilde{\bw}^{\star\top}\hat{\nabla \cL} (\bbeta^{\star})}{\sigma_{\nu}} \leq z\right\} - \Phi(z)\right| \to 0. 
\end{align*}
Then it suffices to show that $|\Delta_{\pi, 1}| + |\Delta_{\pi, 2}| + |\Delta_{\pi, 3}| = o_{\bbP}\left(n^{-1/2}\right)$. 
For simplicity of notation, we write $\bDelta_{\bgamma} = \hat{\bgamma} - \bgamma^{\star}$ and $\tilde{\bDelta}_{\bgamma} = (0, \bDelta_{\bgamma}^{\top})^{\top}$. By Assumption~\ref{Assumption_covariate} and Corollary~\ref{Corollary_beta0_LASSO}, we have $\max_{k \in [K]} \nabla_{\nu\nu}^{2} \cL_{k}(\hat{\bbeta})(\hat{\nu} - \nu^{\star})^{2} \leq 4B^{2} (\hat{\nu} - \nu^{\star})^{2} = O_{\bbP}\{|\cS_{\star}|(\log p)/n\}$. Therefore 
\begin{align*}
    \max_{k \in [K]}\bDelta_{\bgamma}^{\top} \nabla_{\bgamma\bgamma}^{2} \cL_{k}(\hat{\bbeta})\bDelta_{\bgamma} &\leq 2 \max_{k \in [K]} \bdelta^{\top}\nabla^{2} \cL_{k}(\hat{\bbeta}) \bdelta + 2 \max_{k \in [K]} \nabla_{\nu\nu}^{2} \cL_{k}(\hat{\bbeta})(\hat{\nu} - \nu^{\star})^{2} \cr
    &= O_{\bbP}\left(\frac{|\cS_{\star}|\log p}{n}\right). 
\end{align*}
Consequently, by the mean value theorem, Lemma~\ref{Lemma_3.2_Huang} and Lemma~\ref{Lemma_bw_consistency}, we obtain 
\begin{align*}
    |\Delta_{\pi, 1}| &\leq \max_{k \in [K]} \sqrt{\bvarphi_{k}^{\top}\nabla_{\bgamma\bgamma}^{2} \cL_{k}(0, \bgamma^{\circ})\bvarphi_{k}} \sqrt{\bDelta_{\bgamma}^{\top}\nabla_{\bgamma\bgamma}^{2} \cL_{k}(0, \bgamma^{\circ})\bDelta_{\bgamma}}\cr
    &= O_{\bbP}\left(\frac{|\cS_{\star}|\log p}{n} + (\log p)\sqrt{\frac{|\cS_{\star}||\cS_{\sharp}|}{mn}}\right) = o_{\bbP}\left(n^{-1/2}\right), 
\end{align*}
where $\bgamma^{\circ} = \bgamma^{\star} + \upsilon \bDelta_{\bgamma}$ for some $\upsilon \in [0, 1]$. We now upper bound $|\Delta_{\pi, 2}|$. Using~\eqref{eq_uniform_deviation_bound} together with Lemma~\ref{Lemma_bw_consistency} and~\eqref{eq_clt_cond}, we obtain 
\begin{align*}
    |\Delta_{\pi, 2}| &\leq \max_{k \in [K]} \|\bvarphi_{k}\|_{1}\max_{k \in [K]} \|\nabla_{\bgamma}\tilde{\cL}_{k}(\bbeta^{\star})\|_{\infty} \cr
    &= O_{\bbP}\left(\frac{|\cS_{\star}|\log p}{n\sqrt{K}} + \frac{|\cS_{\sharp}|\log p}{m \sqrt{K}}\right) = o_{\bbP}\left(n^{-1/2}\right).
\end{align*}
Recall that $\tilde{\bDelta}_{\bgamma}^{\top} \cH^{\star}\tilde{\bw}^{\star} = 0$. Hence  
\begin{align*}
    \Delta_{\pi, 3} &= \tilde{\bw}^{\star \top}\left\{\hat{\nabla^{2} \cL}(0, \bgamma^{\star}) - \cH^{\star}\right\} \tilde{\bDelta}_{\bgamma} + \tilde{\bw}^{\star \top} \left\{\hat{\nabla^{2} \cL}(0, \bgamma^{\natural}) - \hat{\nabla^{2} \cL}(0, \bgamma^{\star})\right\} \tilde{\bDelta}_{\bgamma} \cr
    &=: \Delta_{\pi, 3}^{\diamond} + \Delta_{\pi, 3}^{\circ}, 
\end{align*}
where $\bgamma^{\natural} = \bgamma^{\star} + \upsilon_{\natural} \bDelta_{\bgamma}$ for some $\upsilon_{\natural} \in [0, 1]$. Then, by Assumption~\ref{Assumption_cM_bw} and Lemma~\ref{Lemma_Hessian}, 
\begin{align*}
    |\Delta_{\pi, 3}^{\diamond}| \leq \|\bDelta_{\bgamma}\|_{1} \max_{k \in [K]}\left\|\left\{\nabla^{2}\cL_{k}(\bbeta^{\star}) - \cH^{\star}\right\}\tilde{\bw}^{\star}\right\|_{\infty} = O_{\bbP}\left(\frac{|\cS_{\star}|\log p}{\sqrt{nm}}\right) = o_{\bbP}\left(n^{-1/2}\right). 
\end{align*}
Observe that $\bgamma^{\natural} - \bgamma^{\star} = \upsilon_{\natural} \bDelta_{\bgamma}$. Hence, by~\eqref{equation_HB} in Remark~\ref{remark_HB} and~\eqref{eq_clt_cond}, we obtain 
\begin{align*}
    |\Delta_{\pi, 3}^{\circ}| &\leq \max_{k \in [K]} \left|\tilde{\bw}^{\star \top}\left\{\nabla^{2} \cL_{k}(0, \bgamma^{\natural}) - \nabla^{2} \cL_{k}(0, \bgamma^{\star})\right\}\tilde{\bDelta}_{\bgamma}\right|\cr
    &= O_{\bbP}\left(\sqrt{\frac{\log p}{m}}\frac{|\cS_{\star}|^{2} \log p}{n} + \frac{|\cS_{\star}|\log p}{n}\right) = o_{\bbP}\left(n^{-1/2}\right). 
\end{align*}
Putting all these pieces together with Lemma~\ref{Lemma_sigma_nu}, we obtain the asymptotic normality in~\eqref{eq_CLT_pi_nu}. 
\end{proof}

\section{Distributed estimation of the baseline hazard function} 
In the context of Cox's proportional hazards model, the estimation of both the cumulative baseline hazard function $\Lambda_{0}(\cdot)$ and the baseline hazard function $\lambda_{0}(\cdot)$ is also of great interest and importance. Here, we introduce two communication-efficient estimators for $\Lambda_{0}(\cdot)$ and $\lambda_{0}(\cdot)$, respectively, in the distributed setting.

Motivated by Corollary~\ref{Corollary_beta0_LASSO}, we take $\tilde{\bbeta} = \bbeta_{T}$, where $T \geq \lceil (\log K)/(2\log(1/ \cA_{0}))\rceil$, to be the estimator for $\bbeta^{\star}$, which is communication-efficient and attains the same convergence rate as the full-sample estimator. Then our distributed Breslow-type estimator~\citep{Andersen1982cox} for $\Lambda_{0}(t)$ is defined by 
\begin{align}
\label{eq_Distributed_Breslow_Estimator}
    \hat{\Lambda}_{0}(t) = \frac{1}{K} \sum_{k = 1}^{K} \int_{0}^{t} \frac{d \bar{N}_{k}(s)}{\sum_{i \in \mathcal{I}_{k}} Y_{i}(s) \exp\{\bx_{i}(s)^{\top} \tilde{\bbeta}\}}, 
\end{align}
which can be easily computed by aggregating and averaging the local Breslow estimators. To estimate the baseline hazard function $\lambda_{0}(\cdot)$, we employ the kernel smoothing technique~\citep{Wells1994Biometrika, Fan1997AOS, Fan2006AOS} in view of the definition that $\Lambda_{0}(t) = \int_{0}^{t} \lambda_{0}(s) d s$. Specifically, let $\mathcal{K}(\cdot)$ be a kernel function and $h > 0$ be a bandwidth parameter. With $\tilde{\bbeta}$, we introduce the distributed nonparametric estimator for $\lambda_{0}(t)$ as 
\begin{align*}
    \hat{\lambda}_{0}(t; h) = \frac{1}{K} \sum_{k = 1}^{K} \int_{0}^{\tau} \frac{\mathcal{K}_{h}(t - s) d \bar{N}_{k}(s)}{\sum_{i \in \mathcal{I}_{k}} Y_{i}(s) \exp\{\bx_{i}(s)^{\top} \tilde{\bbeta}\}}, 
\end{align*}
where $\mathcal{K}_{h}(\cdot) = \mathcal{K}(\cdot/h)/h$. Similar to $\hat{\Lambda}_{0}(t)$, $\hat{\lambda}_{0}(t; h)$ is also convenient to compute and communication-efficient.

\begin{assumption}
\label{Assumption_Kernel_Function}
    $\mathcal{K}(\cdot)$ is a symmetric probability density function with $\int_{\mathbb{R}} u^{2} \mathcal{K}(u) d u < \infty$. There exists a positive constant $L_{\mathcal{K}} < \infty$ such that
    \begin{align*}
        |\mathcal{K}(u) - \mathcal{K}(u')| \leq L_{\mathcal{K}} |u - u'|, \enspace \forall u, u' \in \mathbb{R}. 
    \end{align*}
\end{assumption}

Assumption~\ref{Assumption_Kernel_Function} is fundamental regarding the kernel function and is frequently employed in the literature on nonparametric density estimation. It is satisfied by many commonly used kernel functions, such as the Epanechnikov kernel $\mathcal{K}(u) = 0.75 (1 - u^{2}) \mathbb{I}\{|u| \leq 1\}$ and the Gaussian kernel $\mathcal{K}(u) = \exp(-u^{2}/2)/\sqrt{2 \pi}$.

The convergence rates of these two estimators are established in the following theorem.

\begin{theorem}
\label{Theorem_Baseline_Hazard_Function_Estimation}
    Under the conditions of Theorem~\ref{Theorem_iteration}, we have  
    \begin{align*}
        \sup_{t \in [0, \tau]} |\hat{\Lambda}_{0}(t) - \Lambda_{0}(t)| = O_{\mathbb{P}}\left(|\mathcal{S}_{\star}|\sqrt{\frac{\log p}{n}}\right).
    \end{align*}
    Furthermore, assume that $|\mathcal{S}_{\star}| \sqrt{\log (p n)} = o(n h^{2})$, $h = o(1)$ and the baseline hazard function $\lambda_{0}(\cdot)$ is twice continuously differentiable. Then, under Assumption~\ref{Assumption_Kernel_Function}, we have  
    \begin{align*}
        \sup_{t \in [h, \tau - h]} |\hat{\lambda}_{0}(t; h) - \lambda_{0}(t)| = O_{\mathbb{P}}\left(|\mathcal{S}_{\star}|\sqrt{\frac{\log (p n)}{n h^{2}}} + h^{2}\right). 
    \end{align*}
\end{theorem}

\subsection{Proof of Theorem~\ref{Theorem_Baseline_Hazard_Function_Estimation}}
\begin{proof}[Proof of Theorem~\ref{Theorem_Baseline_Hazard_Function_Estimation}]
    Denote $\hat{\Lambda}_{0}(t) - \Lambda_{0}(t) = \Delta_{1}(t) + \Delta_{2}(t)$ for $t \in [0, \tau]$, where 
    \begin{align*}
        \Delta_{1}(t) &= \frac{1}{K} \sum_{k = 1}^{K} \int_{0}^{\tau} \left[\frac{d \bar{N}_{k}(s)}{\sum_{i \in \mathcal{I}_{k}} Y_{i}(s) \exp\{\bx_{i}(s)^{\top} \tilde{\bbeta}\}} - \frac{d \bar{N}_{k}(s)}{\sum_{i \in \mathcal{I}_{k}} Y_{i}(s) \exp\{\bx_{i}(s)^{\top} \bbeta^{\star}\}}\right], \cr
        \Delta_{2}(t) &= \frac{1}{K} \sum_{k = 1}^{K} \int_{0}^{t} \frac{d \bar{M}_{k}(s)}{\sum_{i \in \mathcal{I}_{k}} Y_{i}(s) \exp\{\bx_{i}(s)^{\top} \bbeta^{\star}\}}.  
    \end{align*}
    By Assumption~\ref{Assumption_at_risk} and Hoeffding's inequality,   
    \begin{align*}
        \mathbb{P}\left(\min_{k \in [K]} \frac{1}{m} \sum_{i \in \mathcal{I}_{K}} Y_{i}(\tau) \geq \frac{\rho_{0}}{2} \right) \leq K \exp\left(-\frac{m \rho_{0}^{2}}{2}\right).
    \end{align*}
    Then, as $\|\tilde{\bbeta} - \bbeta^{\star}\|_{1}^{2} = O_{\mathbb{P}} (|\mathcal{S}_{\star}| (\log p)/n)$, and using Assumption~\ref{Assumption_covariate}, it follows that 
    \begin{align}
    \label{eq_Delta_1_bound_Lambda}
        \sup_{t \in [0, \tau]} |\Delta_{1}(t)| &\leq \max_{i \in [n]} \sup_{t \in [0, \tau]} |\bx_{i}(t)^{\top} (\tilde{\bbeta} - \bbeta^{\star})| \frac{1}{K} \sum_{k = 1}^{K} \int_{0}^{\tau} \frac{d \bar{N}_{k} (s)}{\sum_{i \in \mathcal{I}_{k}} Y_{i}(s) \exp\{\bx_{i}(s)^{\top} \tilde{\bbeta}_{\upsilon}\}}\cr
        &\lesssim B \|\tilde{\bbeta} - \bbeta^{\star}\|_{1} = O_{\mathbb{P}} \left(|\mathcal{S}_{\star}|\sqrt{\frac{\log p}{n}}\right),  
    \end{align}
    where $\tilde{\bbeta}_{\upsilon} = \bbeta^{\star} + \upsilon (\tilde{\bbeta} - \bbeta^{\star})$ for some constant $\upsilon \in [0, 1]$.
    By~\eqref{eq_event_E} and Lemma~\ref{Lemma_3.3_Huang}, we have $\sup_{t \in [0, \tau]} |\Delta_{2}(t)| = O_{\mathbb{P}}(n^{-1/2})$ and 
    \begin{align*}
        \sup_{t \in [0, \tau]} |\hat{\Lambda}_{0}(t) - \Lambda(t)| \leq \sup_{t \in [0, \tau]} |\Delta_{1}(t)| + \sup_{t \in [0, \tau]} |\Delta_{2}(t)| = O_{\mathbb{P}}\left(|\mathcal{S}_{\star}| \sqrt{\frac{\log p}{n}}\right).  
    \end{align*}
    Similarly, we denote $\hat{\lambda}_{0}(t; h) - \lambda_{0}(t) = \Delta_{1}(t; h) + \Delta_{2}(t; h) + \Delta_{3}(t; h)$, where 
    \begin{align*}
        \Delta_{1}(t; h) &= \frac{1}{K} \sum_{k = 1}^{K} \int_{0}^{\tau} \left[\frac{\mathcal{K}_{h}(t - s) d \bar{N}_{k}(s)}{\sum_{i \in \mathcal{I}_{k}} Y_{i}(s) \exp\{\bx_{i}(s)^{\top} \tilde{\bbeta}\}} - \frac{\mathcal{K}_{h}(t - s) d \bar{N}_{k}(s)}{\sum_{i \in \mathcal{I}_{k}} Y_{i}(s) \exp\{\bx_{i}(s)^{\top} \bbeta^{\star}\}}\right], \cr 
        \Delta_{2}(t; h) &= \frac{1}{K} \sum_{k = 1}^{K} \int_{0}^{t} \frac{\mathcal{K}_{h}(t - s) d \bar{M}_{k}(s)}{\sum_{i \in \mathcal{I}_{k}} Y_{i}(s) \exp\{\bx_{i}(s)^{\top} \bbeta^{\star}\}}, \cr
        \Delta_{3}(t; h) &= \int_{0}^{\tau} \mathcal{K}_{h}(t - s) \lambda_{0}(s) ds - \lambda_{0}(t). 
    \end{align*}
    Following a similar argument as that in~\eqref{eq_Delta_1_bound_Lambda}, we have 
    \begin{align*}
        \sup_{t \in [h, \tau - h]} |\Delta_{1}(t; h)| &\lesssim B \|\tilde{\bbeta} - \bbeta^{\star}\|_{1} \frac{1}{K} \sum_{k = 1}^{K} \int_{0}^{t} \frac{\mathcal{K}_{h} (t - s) d \bar{N}_{k} (s)}{\sum_{i = 1}^{n} Y_{i}(s) \exp\{\bx_{i}(s)^{\top} \tilde{\bbeta}_{\upsilon}\}} \cr
        &= O_{\mathbb{P}} \left(|\mathcal{S}_{\star}|\sqrt{\frac{\log p}{n h^{2}}}\right). 
    \end{align*}
    By Assumption~\ref{Assumption_Kernel_Function}, Lemma~\ref{Lemma_3.3_Huang} and the baseline hazard function $\lambda_{0}(\cdot)$ is twice continuously differentiable, it is straightforward to obtain 
    \begin{align*}
        \sup_{t \in [h, \tau - h]} |\Delta_{2}(t; h)| = O_{\mathbb{P}} \left(\sqrt{\frac{\log n}{n h^{2}}}\right). 
    \end{align*}
    By Assumption~\ref{Assumption_Kernel_Function}, we have $\sup_{t \in [h, \tau - h]} |\Delta_{3}(t; h)| \lesssim h^{2}$. Putting all these pieces together, we obtain the convergence rate of $\hat{\lambda}(t; h)$ in Theorem~\ref{Theorem_Baseline_Hazard_Function_Estimation}.     
\end{proof}

\section{Additional simulated experiments}

In this Section, we add simulated experiments for $n = 240$ with a $50\%$ censoring rate and $p = 300$, and demonstrate the performance of our algorithms via various metrics as in Section~\ref{sec_simulations}. We use $K = 2$ centers and simulate $200$ i.i.d replications.

As in Section~\ref{subsubsection-estimation}, we study the median estimation error of estimators. Considering $\bbeta^{\star} = (0,2,2,2,{\bf{0}}_{p-4}^{\top})^{\top}$, we obtain Table~\ref{tab_median_estimation_error_appendix}. Our iterated procedure has an estimation error close to the full-sample benchmark, and is much smaller than other estimators.

\begin{table}[H]
\begin{center}
\caption{\small{Median estimation error of various estimators.}}
\label{tab_median_estimation_error_appendix}
\footnotesize
\vspace{2mm}
\begin{tabular}{ccccc}
\hline\hline
{5 iterations} & {10 iterations} & {Full-sample} & {One-center} & {Average-debiased}
\\\cline{1-5}
2.25 & 2.18 & 1.95 & 3.54 & 3.65
\\\hline\hline
\end{tabular}
\end{center}
\end{table}

Similar to Section~\ref{subsubsection-hyptesting}, we consider testing the hypothesis
\begin{align*}
    H_0 : \nu^{\star} = 0 \enspace \mathrm{versus} \enspace H_{1} : \nu^{\star} \neq 0.  
\end{align*}
The Q-Q plots of the $p$-values under $H_0$ in the current simulation setting are in Figure~\ref{fig_qqplot_appendix}. As expected, the $p$-values from our iterated estimator and from the full-sample estimator follow the Uniform $(0,1)$ distribution more closely than the $p$-values coming from the one-center and the average-debiased estimators.

\begin{figure}[H]
\centering
    \begin{subfigure}{\subfigfracin\linewidth}
        \includegraphics[width=\imgfrac\linewidth]{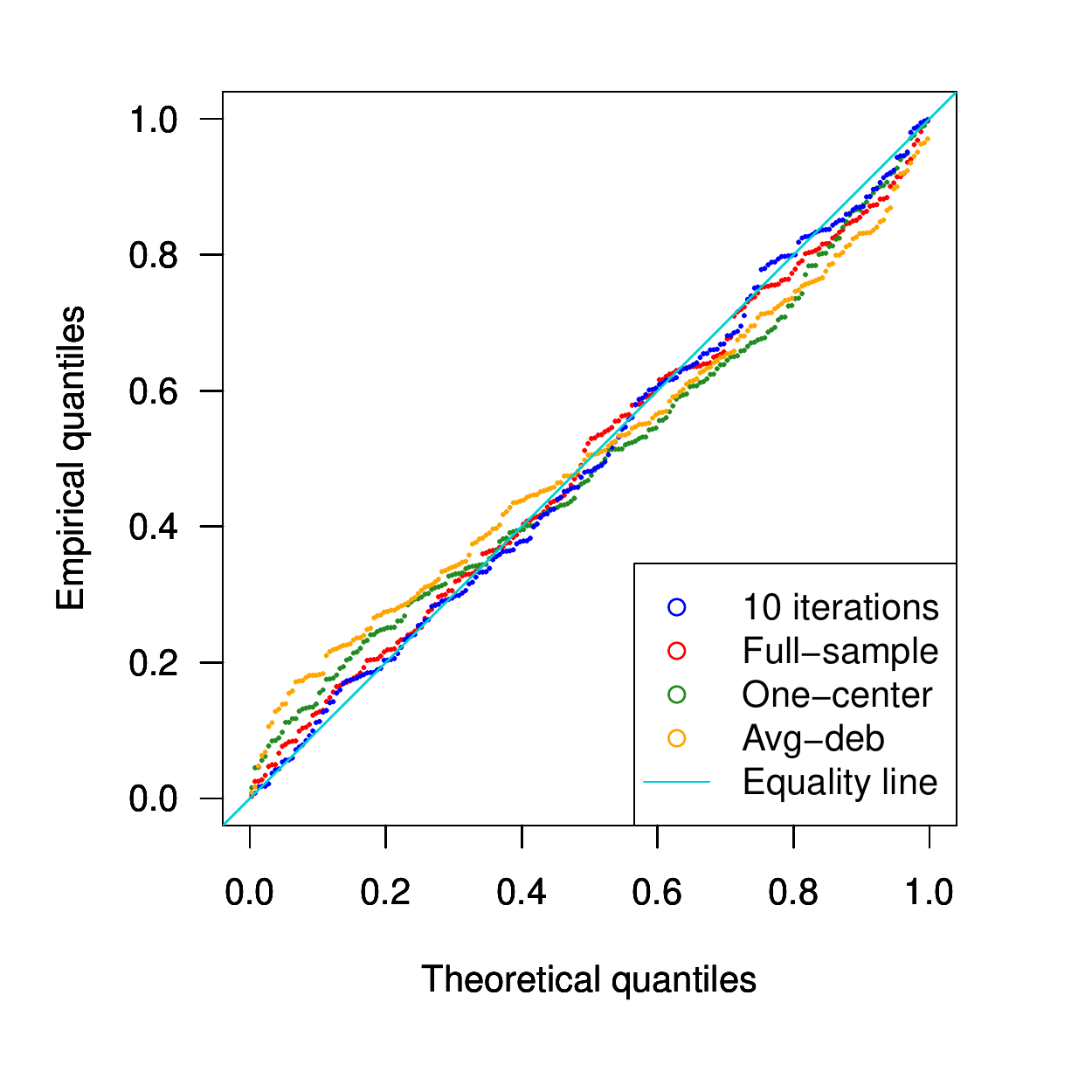}
    \end{subfigure}\hfill
    \caption{Q-Q plots of the $p$-values under $H_0$.}
    \label{fig_qqplot_appendix}
\end{figure}

Under the current ($n$, $p$) setting, the tests presented in the paper have their size displayed in Table~\ref{tab_size_appendix} and power in Table~\ref{tab_power_appendix}. The estimated test size for our estimator is the closest to the target level. It is also more powerful than the tests based on the one-center and average-debiased estimators.

\begin{table}[H]
\begin{center}
\caption{\small{Size of the tests with significance level $\alpha = 0.05$.}}
\label{tab_size_appendix}
\footnotesize
\vspace{2mm}
\begin{tabular}{ccccc}
\hline\hline
& {10 iterations} & {Full-sample} & {One-center} & {Average-debiased}
\\\cline{1-5}
{Average} & 0.055 & 0.025 & 0.035 & 0.010\\
{Standard error} & 0.02 & 0.01 & 0.01 & 0.01
\\\hline\hline
\end{tabular}
\end{center}
\end{table}

\begin{table}[H]
\begin{center}
\caption{\small{Power of the tests for $\nu^{\star} = 0.5$.}}
\label{tab_power_appendix}
\footnotesize
\vspace{2mm}
\begin{tabular}{ccccc}
\hline\hline
& {10 iterations} & {Full-sample} & {One-center} & {Average-debiased}
\\\cline{1-5}
{Average} & 0.60 & 0.61 & 0.22 & 0.38\\
{Standard error} & 0.04 & 0.04 & 0.03 & 0.04
\\\hline\hline
\end{tabular}
\end{center}
\end{table}

We now turn to confidence intervals for the linear functional $\bc^{\top} \bbeta^{\star}$ as in Section~\ref{subsubsection-linearfunc}. The estimated coverage probability is displayed in Table~\ref{tab_coverage_appendix} and the median interval width in Table~\ref{tab_width_appendix}. Our confidence interval has an estimated coverage probability closer to the $95 \%$ target than the intervals based on other estimators, and a width very similar to the one based on the full-sample estimator.
\begin{table}[H]
\begin{center}
\caption{\small{Coverage probability with target $95\%$.}}
\label{tab_coverage_appendix}
\footnotesize
\vspace{2mm}
\begin{tabular}{ccccc}
\hline\hline
& {10 iterations} & {Full-sample} & {One-center} & {Average-debiased}
\\\cline{1-5}
{Average} & 0.965 & 0.99 & 0.985 & 0.99\\
{Standard error} & 0.01 & 0.01 & 0.01 & 0.01
\\\hline\hline
\end{tabular}
\end{center}
\end{table}

\begin{table}[H]
\begin{center}
\caption{\small{Confidence interval width.}}
\label{tab_width_appendix}
\footnotesize
\vspace{2mm}
\begin{tabular}{cccc}
\hline\hline
{10 iterations} & {Full-sample} & {One-center} & {Average-debiased}
\\\cline{1-4}
0.53 & 0.52 & 0.74 & 0.45
\\\hline\hline
\end{tabular}
\end{center}
\end{table}

\end{document}